\def\IsFull{}

\ifdefined\IsSpringer{}
\documentclass[runningheads]{llncs}
\else
    \ifdefined\IsSP{}
    \documentclass[conference,compsoc]{IEEEtran}
    \usepackage{amsthm}
    \usepackage{amsmath}
    \usepackage[nocompress]{cite}
    \else
    \documentclass[8pt]{article}

\usepackage{fullpage} 
\usepackage{amsthm}
    \fi
\fi
\usepackage[T1]{fontenc}
\usepackage{booktabs}
\usepackage{arydshln}
\usepackage{orcidlink}
\usepackage{multicol,multirow}
\usepackage{thm-restate}
\usepackage{graphicx}
\usepackage{subcaption}
\usepackage{tikz}
\usepackage{multicol,multirow}
\usetikzlibrary{positioning}
\usetikzlibrary{shapes}
\usetikzlibrary{matrix,arrows}
\usetikzlibrary{decorations}
\usetikzlibrary{decorations.pathreplacing}
\usetikzlibrary{shapes.geometric, chains, calc, fit, matrix}
\tikzstyle{system}=[rectangle,draw,fill=lightgray,minimum height=0.8cm,minimum
            width=0.8cm,thick]
\tikzstyle{BC}=[system]
\tikzstyle{rzesource}=[system]
\tikzstyle{RO}=[resource, minimum width=1cm]
\tikzstyle{protocol}=[circle, inner sep=0.7mm, draw]
\tikzstyle{simulator}=[circle, inner sep=0.7mm, draw]
\tikzstyle{memory}=[resource]
\tikzstyle{distinguisher}=[resource,fill=white,minimum width=3.5cm,
minimum height=1.2cm]
\tikzstyle{link}=[]
\usepackage{thm-restate}
\usepackage{tabularray}
\UseTblrLibrary{booktabs}
\usepackage{pgfplots}
\usepackage{siunitx}
\usepackage{hyperref}
\hypersetup{
    colorlinks=true,
    linkcolor=blue,
    filecolor=magenta,
    citecolor=magenta,      
    urlcolor=cyan   
    }

\usepackage{pifont} 

\usepackage[
	n,
	operators,
	advantage,
	sets,
	adversary,
	landau,
	probability,
	notions,	
	logic,
	ff,
	mm,
	primitives,
	events,
	complexity,
	oracles,
	asymptotics,
	keys]{cryptocode}
\usepackage{enumitem}
\usepackage{algorithm}
\usepackage{algpseudocode}
\usepackage{float}
\usepackage{adjustbox}
 \RequirePackage{varwidth}
 \RequirePackage{color}
 \RequirePackage[most]{tcolorbox}
\date{}
\usepackage{cite}
\begin{document}
\def\ddefloop#1{\ifx\ddefloop#1\else\ddef{#1}\expandafter\ddefloop\fi}

\def\ddef#1{\expandafter\def\csname bf#1\endcsname{{\mbox{\bf #1}}}}
\ddefloop ABCDEFGHIJKLMNOPQRSTUVWXYZabcdefghijklmnopqrstuvwxyz\ddefloop

\def\ddef#1{\expandafter\def\csname bf#1\endcsname{\ensuremath{\pmb{\csname #1\endcsname}}}}
\ddefloop {alpha}{beta}{gamma}{delta}{epsilon}{varepsilon}{zeta}{eta}{theta}{vartheta}{iota}{kappa}{lambda}{mu}{nu}{xi}{pi}{varpi}{rho}{varrho}{sigma}{varsigma}{tau}{upsilon}{phi}{varphi}{chi}{psi}{omega}{Gamma}{Delta}{Theta}{Lambda}{Xi}{Pi}{Sigma}{varSigma}{Upsilon}{Phi}{Psi}{Omega}{ell}\ddefloop

\def\ddef#1{\expandafter\def\csname bb#1\endcsname{\ensuremath{\mathbb{#1}}}}
\ddefloop ABCDEFGHIJKLMNOPQRSTUVWXYZ\ddefloop

\def\ddef#1{\expandafter\def\csname c#1\endcsname{\ensuremath{\mathcal{#1}}}}
\ddefloop ABCDEFGHIJKLMNOPQRSTUVWXYZ\ddefloop

\newcommand{\getsr}{\xleftarrow{\$}}

\newcommand{\Floor}[1]{\left \lfloor #1\right\rfloor}
\renewcommand{\angle}[1]{\left\langle #1\right\rangle}
\newcommand{\Ceil}[1]{\ceil #1\rceil}
\renewcommand{\set}[1]{\left\{#1\right\}}
\newcommand{\hmin}{\text{H}_\infty}

\renewcommand{\prover}{\mathcal{P}}
\renewcommand{\verifier}{\mathcal{V}}
\newcommand{\ZKP}{\mathsf{ZKP}}
\newcommand{\zkio}{\mathsf{io}}
\renewcommand{\pp}{\mathsf{pp}}
\newcommand{\witness}{\mathsf{wt}}
\newcommand{\zkstate}{\mathsf{st}}
\newcommand{\Prime}{p}
\newcommand{\Gen}{\term{Gen}}
\newcommand{\LWR}{\texttt{LWR}}
\newcommand{\LWE}{\texttt{LWE}}
\newcommand{\lwr}{\term{LWR}}


\newcommand{\anti}[1]{\textcolor{magenta}{[Antigoni: #1]}}
\newcommand{\hnote}[1]{\textcolor{red}{[Harish: #1]}}

\newcommand{\term}[1]{\mathsf{#1}}
\newcommand{\sys}{\term{TAPAS}}
\newcommand{\Sys}{\textrm{Two-server Asymmetric Private Aggregation Scheme}}
\newcommand{\etal}{\textit{et al.}}
\newcommand{\Hybrid}{\term{Hybrid}}
\renewcommand{\negl}{\term{negl}}


\newcommand{\Ext}{\term{Ext}}
\newcommand{\psaaux}{\term{aux}}
\newcommand{\psact}{\term{ct}}
\newcommand{\Commit}{\term{PCommit}}
\newcommand{\commit}{\term{DLCommit}}

\newcommand{\bpsact}{\boldsymbol{\psact}}
\newcommand{\bpsain}{\boldsymbol{{x}}}

\newcommand{\cIf}{\mbox{\textbf{if }}}

\newcommand{\com}{\term{com}}
\newcommand{\coms}{\term{Setup}}
\newcommand{\comc}{\term{Commit}}
\renewcommand{\pp}{\term{pp}}

\newcommand{\eqq}{\stackrel{?}{=}}
\newtcolorbox{protocol}[1][]{%
width=\textwidth, boxsep=1mm, left= 0mm, right=0.5mm, title={\textbf{#1}}
  }

\newcommand{\aok}{\term{SoK}_\cR}
\newcommand{\aoks}{\mathsf{Setup}_\cR}
\newcommand{\aokp}{\term{Sign}_\cR}
\newcommand{\aokv}{\term{Verify}}
\newcommand{\aoksim}{\term{Sim}}
\newcommand{\stmt}{\term{stmt}}
\newcommand{\wit}{\term{wit}}
\newcommand{\aokcrs}{\term{crs}}
\newcommand{\tr}{\mathrm{tr}}

\newcommand{\cFor}{\mbox{\bf for }}
\newcommand{\cTo}{\mbox{\bf to }}
\newcommand{\cDownTo}{\mbox{\bf downto }}
\newcommand{\cDo}{\mbox{\bf do }}
\renewcommand{\cIf}{\mbox{\bf if }}
\newcommand{\cAnd}{\mbox{\bf and }}
\newcommand{\cThen}{\mbox{\bf then }}
\newcommand{\cElse}{\mbox{\bf else }}
\newcommand{\cWhile}{\mbox{\bf while }}
\newcommand{\cOr}{\mbox{\bf or }}
\newcommand{\cAbort}{\mbox{\bf abort }}

\ifdefined\IsSpringer{}
\newenvironment{claimproof}[1]{\par\noindent\underline{Proof:}\space#1}{}
\else
\makeatletter
\renewcommand{\paragraph}[1]{%
  \@startsection{paragraph}{4}{\z@}%
  {1.25ex \@plus1ex \@minus.2ex}%
  {-1em}%
  {\normalfont\normalsize\bfseries}%
  {\strip@period#1}%
}
\ifdefined\IsFull

\def\strip@period#1.{#1}
\makeatother
\theoremstyle{definition}
\newtheorem{definition}{Definition}
\newtheorem{construction}{Construction}
\theoremstyle{plain}
\newtheorem{theorem}{Theorem}
\newtheorem{lemma}{Lemma}
\newtheorem{claim}{Claim}
\newenvironment{claimproof}[1]{\par\noindent\underline{Proof:}\space#1}{}
\theoremstyle{remark} 
\newtheorem{remark}{Remark} 
\fi
\fi

\newcommand{\tSim}{\term{Sim}}
\newcommand{\lab}{\ell}


\newcommand{\kappaSec}{\kappa}                      
\newcommand{\NP}{\mathsf{NP}}                           
\newcommand{\PPT}{\mathsf{PPT}}                         
\newcommand{\Prb}{\Pr}                                  
\newcommand{\given}{\;|\;}                              
\newcommand{\advA}{\mathcal{A}}                         
\newcommand{\advD}{\mathcal{D}}                         
\newcommand{\SimA}{\mathcal{S}}                         
\newcommand{\ExtA}{\mathcal{E}}                         
\newcommand{\Experiment}{\mathsf{Exp}}                  
\newcommand{\Game}{\mathsf{Game}}                       
\newcommand{\CRS}{\mathsf{crs}}                         
\newcommand{\TD}{\mathsf{td}}                           
\newcommand{\SimO}{\mathcal{O}_{\Sim}}                  
\newcommand{\SimQ}{\mathcal{Q}}                         
\newcommand{\stA}{\mathsf{st}_{\advA}}                  
\newcommand{\Trans}{\mathcal{T}}                        

\newcommand{\probexp}[2]{\Prb\!\left[#1:#2\right]}      
\newcommand{\ExpSE}[1]{\Experiment^{\mathsf{se}}_{#1}}  
\newcommand{\ExpSHVZK}[1]{\Experiment^{\mathsf{shvzk}}_{#1}} 

\newcommand{\Setup}{\mathsf{Setup}}
\newcommand{\Prove}{\mathsf{Prove}}
\newcommand{\Verif}{\mathsf{Verif}}
\newcommand{\Sim}{\mathsf{Sim}}

\newcommand{\Rel}{R}
\newcommand{\Lang}{L_\Rel}

\DeclareRobustCommand{\lesssim}{%
  \mathrel{\mathpalette\lowersim\originallesssim}%
}
\DeclareRobustCommand{\gtrsim}{%
  \mathrel{\mathpalette\lowersim\originalgtrsim}%
}

\makeatletter
\renewcommand{\paragraph}[1]{\smallskip\noindent{\it #1}}
\makeatother


%
\title{\textrm{TAPAS}: Efficient Two-Server Asymmetric Private Aggregation Beyond Prio(+)}
%
%
\ifdefined\IsSpringer{}
    \ifdefined\IsSub{}
    \author{}
    \else
    \author{Harish Karthikeyan\orcidID{0000-0002-1787-4906} \and
    Antigoni Polychroniadou\orcidID{0009-0003-0125-2971}}
           \institute{JPMorgan AI Research, JPMorgan AlgoCRYPT CoE, New York, NY, USA}
    \authorrunning{H. Karthikeyan and A. Polychroniadou}
    \fi
\else
\author{%
\normalsize
 \begin{tabular}{c c}
Harish Karthikeyan~\orcidlink{0000-0002-1787-4906} & Antigoni Polychroniadou~\orcidlink{0009-0003-0125-2971} \\
\texttt{harish.karthikeyan@jpmchase.com} & \texttt{antigoni.polychroniadou@jpmorgan.com}
\end{tabular}
\\[2ex]
\normalsize JPMorgan AI Research, JPMorgan AlgoCRYPT CoE
}
\fi
%

%
%
\maketitle              

\begin{abstract}

Privacy‑preserving aggregation is a cornerstone for AI systems that learn from distributed data without exposing individual records, especially in federated learning and telemetry. Existing two‑server protocols (e.g., Prio and successors) set a practical baseline by validating inputs while preventing any single party from learning users’ values, but they impose symmetric costs on both servers and communication that scales with the per‑client input dimension $L$. Modern learning tasks routinely involve dimensionalities $L$ in the tens to hundreds of millions of model parameters.

We present TAPAS, a two‑server asymmetric private aggregation scheme that addresses these limitations along four dimensions: (i) no trusted setup or preprocessing, (ii) server‑side communication that is independent of $L$ (iii) post‑quantum security based solely on standard lattice assumptions (LWE, SIS), and (iv) stronger robustness with identifiable abort and full malicious security for the servers. A key design choice is intentional asymmetry: one server bears the $O(L)$ aggregation and verification work, while the other operates as a lightweight facilitator with computation independent of $L$. This reduces total cost, enables the secondary server to run on commodity hardware, and strengthens the non‑collusion assumption of the servers. One of our main contributions is a suite of new and efficient lattice-based zero-knowledge proofs; to our knowledge, we are the first to establish privacy and correctness with identifiable abort in the two-server setting.

\end{abstract}
\ifdefined\IsTPMPC 
\else
\section{Introduction}
AI-driven systems increasingly rely on privacy-preserving aggregation to unlock collective insights without exposing individual data, most notably in federated learning and telemetry~\cite{firefox,PoPETS:BKMGBE22,CCS:BBGLR20,CCS:BIKMMP17}. Two-server protocols such as Prio~\cite{Prio}, Prio+~\cite{SCN:AGJOP22}, and Elsa~\cite{SP:RSWP23} have established a practical baseline: in these designs, the two servers are symmetric, receiving and computing the same information, typically under secret sharing, so no single party learns users’ inputs, and malformed contributions are validated. Prio~\cite{Prio} is under active standardization consideration within the IETF~\cite{patton-cfrg-vdaf-01}. However, in modern workloads where per-client vectors are extremely large (millions of parameters) and the number of clients $n$ is comparatively modest (often $50-5000$ as reported in Table 2 of~\cite{Kairouz}), these symmetric designs exhibit several bottlenecks. First, each server incurs bandwidth and computation that scale linearly with the vector dimension $L$, so for a model of dimension $L$ both servers must provision network throughput $O(n\cdot L)$ and memory $O(L)$, effectively doubling deployment cost and forcing the secondary server to match the primary aggregator’s capacity. Second, server-side communication is tied to $L$ rather than $n$, misaligning cost with the true bottleneck in federated learning and straining network and memory budgets. Third, setup or preprocessing requirements (for example, trusted setup or client precomputation) complicate operations and introduce assumptions that reduce flexibility. Fourth, security relies on non–post-quantum, limiting long-term robustness against quantum adversaries. Finally, robustness against malicious servers is weaker: once servers begin exchanging messages, a malicious server can trigger aborts that appear indistinguishable from client misbehavior, thereby enabling blame shifting and denial-of-service without attributable fault. Prior work does not provide comprehensive guarantees for malicious security in server-side computation. 

Addressing these limitations calls for an asymmetric two-server design that removes trusted setup and preprocessing, decouples server-side communication from $L$ so costs scale with $n$ rather than vector dimension, bases security solely on standard post-quantum lattice assumptions (LWE, SIS), and strengthens robustness with identifiable abort that attributes faults to misbehaving servers once inter-server messaging begins, thereby aligning private aggregation with real-world constraints in federated learning and telemetry. The recent Heli~\cite{Heli} work explores asymmetric server roles but retains the aforementioned bottlenecks (setup/preprocessing overhead, communication scaling with $L$, non–post-quantum assumptions, and weak robustness under malicious servers). Moreover, Heli imposes a small-plaintext constraint: each vector coordinate must lie in a limited range to enable decryption via discrete-log recovery, which is impractical for high-precision or large-dimension aggregations common in federated learning. 
\fi
\ifdefined\IsTPMPC
AI-driven systems increasingly rely on privacy-preserving aggregation to unlock collective insights without exposing individual data, most notably in federated learning and telemetry~\cite{firefox,PoPETS:BKMGBE22,CCS:BBGLR20,CCS:BIKMMP17}. Two-server protocols such as Prio~\cite{Prio,patton-cfrg-vdaf-01}, Prio+~\cite{SCN:AGJOP22}, and Elsa~\cite{SP:RSWP23}
\footnote{See Section~\ref{sec:related_work_comparison} for a detailed discussion about related work.}
have established a practical baseline: in these designs, the two servers are symmetric, receiving and computing the same information, typically under secret sharing, so no single party learns users’ inputs, and malformed contributions are validated. However, in modern workloads where per-client vectors are extremely large (millions of parameters) and the number of clients $n$ is comparatively modest (often $50-5000$ as reported in Table 2 of~\cite{Kairouz}), these symmetric designs exhibit the following bottlenecks:

\begin{itemize}[leftmargin=*, noitemsep, topsep=0pt]
    \item \textbf{Server Asymmetry:} Both servers must provision $O(nL)$ bandwidth and $O(L)$ memory, doubling deployment costs. Furthermore, Server-to-Server traffic is tied to the vector dimension $L$ rather than the client count $n$.
    \item \textbf{Trusted Setup:} Reliance on trusted setups or heavy client precomputation limits deployment flexibility.
    \item \textbf{Post-quantum Insecurity:} Most existing protocols lack post-quantum security, limiting long-term robustness.
    \item \textbf{Weak Security against Malicious Server:} Robustness against malicious servers is weak~\cite{SP:RSWP23,Heli}: once servers begin exchanging messages, a malicious server can trigger aborts that appear indistinguishable from client misbehavior, thereby enabling blame shifting and denial-of-service without attributable fault. Prior work does not provide comprehensive guarantees against malicious adversaries for server-side computation. 
\end{itemize}

\fi
\ifdefined\IsTPMPC
\begin{table*}[!b]
\centering
\caption{ Asymptotic costs, omitting constant factors, are reported for $n$ clients and vector length $L$, where inputs $\in[0,2^t-1)$. Communication cost omits the computational security parameter ($\lambda$), which accounts for the bit-length of field and group elements. $\kappa$ is the statistical security parameter. For Heli, $\gamma$ denotes the fraction of dropped clients from computation. Server Security has three flavors: semi-honest (SH) server(s), malicious privacy (MP), where an honest server can detect deviation, and its strengthened version, identifiable abort (IA), where the honest server can identify the misbehaving party. Note that only Heli requires a trusted setup phase. $\sys$ provides improved security guarantees over prior work while ensuring server asymmetry and comparable (and, in some cases, improved) asymptotic costs. 
}
\label{tab:costs}
\resizebox{0.85\textwidth}{!}{%
\begin{tabular}{@{}lccc|ccc|cc@{}}
\toprule
 & \multicolumn{3}{c|}{\textbf{Communication Cost}} 
 & \multicolumn{3}{c|}{\textbf{Computation Cost}} 
 & \textbf{Server} & \textbf{Mal.} \\ 
 \cmidrule(lr){2-4} \cmidrule(lr){5-7} 
 & Server 1 & Server 2 & Client 
 & Server 1 & Server 2 & Client 
 & \textbf{Security} & \textbf{Clients}\\ \midrule
\textsc{Elsa} 
      & $n + n L t^{2}$ 
      & $n + n L t^{2}$ 
      & $tL + \lambda + \kappa+L \log L$ 
      & $n(tL + \lambda + \kappa)$ 
      & $n(tL + \lambda + \kappa)$ 
      & $tL + \lambda + \kappa $ 
      & MP & Yes \\[0.5em]
\textsc{Prio} 
      & $ntL\lambda+tL\log tL$ 
      & $ntL+tL\log tL$ 
      & $tL$ 
      & $n tL\log tL $ 
      & $n tL\log tL $ 
      & $tL\log tL$ 
      & SH & Yes \\[0.5em]
\textsc{Prio+} 
      & $n L t $ 
      & $n L t $ 
      & $L t $ 
      & $n L t\lambda$ 
      & $n L t \lambda$ 
      & $L$ 
      & SH & Yes \\[0.5em]
\textsc{Heli} 
      & $n (L+\log t)$ 
      & $(L+\gamma n\log n)$
      & $(L+\log tL)$ 
      & $ntL+L\sqrt{n\cdot 2^t}$
      & $L+ \gamma n$  
      & $tL$ 
      & MP & Yes \\
\midrule 
\textsc{$\sys$} 
      & $ nL$ 
      & $n\lambda$ 
      & $ L+\lambda$ 
      & $nL$ 
      & $n$ 
      & $L$ 
      & SH & No \\[0.5em]
\textsc{$\sys$$^{\text{BP}}$} 
      & $n (L+\log (L+\kappa))$ 
      & $n\lambda$ 
      & $ L+ \log (L+\kappa)+\lambda$ 
      & $nL + n\kappa+ n \lambda$ 
      & $n\lambda$ 
      & $L + \kappa+\lambda$ 
      & SH & Yes \\ [0.5em]
\textsc{$\sys$$^{\text{BP}}$} 
      & $n (L+\log (L+\kappa))$ 
      & $n\lambda$ 
      & $ L+ \log (L+\kappa)+\lambda$ 
      & $nL + n\kappa+ n \lambda$ 
      & $n\lambda+nL$ 
      & $L + \kappa+\lambda$ 
      & IA & Yes \\[0.5em]
\midrule
\textsc{$\sys$$^{\text{LWE}}$} 
      & $n \kappa L$ 
      & $n \cdot \kappa \lambda$ 
      & $\kappa L $ 
      & $n\kappa L$ 
      & $n\kappa \lambda$ 
      & $L \kappa$ 
      & SH & Yes \\ [0.5em]
\textsc{$\sys$$^{\text{LWE}}$} 
      & $n \kappa L$ 
      & $n\lambda$ 
      & $\kappa L $ 
      & $n\kappa L$ 
      & $n\kappa \lambda$ 
      & $L \kappa$ 
      & IA & Yes \\ [0.5em]
\bottomrule
\end{tabular}%
}
\end{table*}
\fi
\ifdefined\IsTPMPC\subsubsection*{Our Contributions}\else\subsection{Our Contributions}\fi\label{sec:contrib}
We formally introduce $\sys$, a Two-Server Asymmetric Private Aggregation Scheme that addresses these limitations along four key dimensions:
\begin{itemize}[leftmargin=*, noitemsep, topsep=0pt]
    \item No setup or preprocessing: $\sys$ operates without trusted setup or client-side preprocessing phases beyond the round’s standard participation, simplifying deployment and removing assumptions
\item  Communication independent of vector length: TAPAS’s server-side communication is independent of the per-client vector dimension $L$, aligning costs with the number of participants $n$ rather than with model size. This is crucial in federated learning, where $L$ can be on the order of millions of parameters while $n$ is comparatively modest. Asymptotic costs for $\sys$ and prior two-server protocols are summarized in Table~\ref{tab:costs}.

\item Post-quantum foundations:  $\sys$ is the first two-server private aggregation protocol built solely on standard lattice assumptions (LWE, SIS), providing post-quantum security without resorting to composite assumptions or ad hoc primitives.

\item Lattice-based Zero Knowledge Proof System: Standard lattice-based proofs suffer from a soundness slack proportional to the square root of the witness dimension, $\sqrt{L}$~\cite{C:LyuNguPla22}. While decomposing the witness into $k$ chunks of dimension $d$ theoretically tightens this slack to $\propto \sqrt{d}$, requiring simultaneous acceptance in across all blocks results in a negligible joint success probability. $\sys$ combines the witness decomposition trick and an \textit{Independent Restart} strategy, where each chunk runs rejection sampling asynchronously until acceptance. This ensures that the success of a block is independent of the success of other blocks, thereby significantly improving performance. This restores linear expected runtime while maintaining a tighter $\sqrt{d}$- factor slack. We prove security under a \textit{relaxed soundness definition}: we guarantee that a witness $\bfs_j$ can be extracted for every block $j$, but make no guarantee that $\bfs_1=\bfs_2=\ldots$. Fortunately, for our intended application of building a two-server secure aggregation protocol, this relaxed soundness definition can be easily rectified within the overall protocol.
\item Stronger robustness with identifiable abort: $\sys$ provides identifiable abort once servers begin exchanging messages. If a server acts maliciously, the protocol produces evidence that attributes the fault to that server (rather than “blaming” honest clients), thereby closing a gap in prior systems in which a malicious server could trigger aborts that appear indistinguishable from client misbehavior. 
The recent work of LZKSA~\cite{CCS:LuLu25} proposes a lattice-based ZK proof tailored to secure aggregation; however, we found that the protocol does not provide zero-knowledge or soundness. 
We defer more information in Section~\ref{sub:lzksa}.
\item Concrete performance gains: Experimental results convincingly demonstrate the significant performance gains that $\sys$ achieves, a 93$\times$ speedup over Prio and 17.5$\times$ over Elsa for vector dimensions of $L=2^{18}$, maintaining high efficiency for million-parameter models by decoupling secondary server overhead from the vector length.


\end{itemize}
\ifdefined\IsTPMPC
\begin{theorem}[Informal Theorem]
\small
    Assuming the hardness of the Hint Learning With Errors (LWE) problem, the Short Integer Solution (SIS) problem, the existence of collision-resistant hash functions, and simulation-extractable NIZKs, there exists a two-server secure aggregation protocol in the programmable random oracle model for high-dimensional vectors of size $L$ that achieves the following:
    \begin{itemize}[leftmargin=*, noitemsep, topsep=0pt]
    \item \textbf{Asymmetric Efficiency:} The protocol concentrates the heavy $O(L)$ bandwidth and computation costs on a single server, while the secondary server serves as a lightweight facilitator with costs \emph{independent} of the model dimension $L$ (scaling only with the number of clients $n$).
    \item \textbf{Malicious Security with Identifiable Abort:} The protocol preserves the privacy of honest clients' inputs against a malicious server and arbitrary malicious client collusion. Furthermore, it enforces \emph{Identifiable Abort}: any party that causes the protocol to halt is cryptographically identified.
\end{itemize}
\end{theorem} 
\fi
\ifdefined\IsTPMPC \else A central design choice that enables these advances is asymmetry between the two servers: one server performs lighter computation and lower-bandwidth tasks, while the other undertakes the heavier verification and aggregation work. The lighter Server receives only $\mathcal{O}(\lambda)$ data per client (seeds) rather than $\mathcal{O}(L)$ (gradients), which for transformer-based models with $L \approx 100-120M$ cuts bandwidth by orders of magnitude. and performs computation independent of $L$, making commodity hardware (e.g., a laptop or \texttt{t3.micro}) sufficient. In particular, we present two versions, $\sys^{\text{BP}}$ and $\sys^{\text{LWE}}$, whose underlying proof systems rely, respectively, on bulletproofs and on lattice-based assumptions. We further subscript the two versions with $\text{SH}$ and $\text{MAL}$ to denote whether the protocol provides security against semi-honest servers (with malicious clients) or malicious servers, respectively. \fi
\ifdefined\IsTPMPC
\else
We state our main result informally as follows.

\begin{theorem}[Informal Theorem]
    Assuming the hardness of the Hint Learning With Errors (LWE) problem, the Short Integer Solution (SIS) problem, the existence of collision-resistant hash functions, and simulation-extractable NIZKs, there exists a two-server secure aggregation protocol in the programmable random oracle model for high-dimensional vectors of size $L$ that achieves the following:
    \begin{itemize}
    \item \textbf{Asymmetric Efficiency:} The protocol concentrates the heavy $O(L)$ bandwidth and computation costs on a single server, while the secondary server serves as a lightweight facilitator with costs \emph{independent} of the model dimension $L$ (scaling only with the number of clients $n$).
    \item \textbf{Malicious Security with Identifiable Abort:} The protocol preserves the privacy of honest clients' inputs against a malicious server and arbitrary malicious client collusion. Furthermore, it enforces \emph{Identifiable Abort}: any party that causes the protocol to halt is cryptographically identified.
\end{itemize}
\end{theorem} 
\fi

\ifdefined\IsTPMPC\else
\begin{table*}[!t]
\centering
\caption{ Costs are reported for $n$ clients and vector length $L$, where inputs $\in[0,2^t-1)$ Computation and communication (vector count) are given in $\mathcal(\cdot)$ notation, omitting constant factors. Communication cost omits the computational security parameter ($\lambda$) factor that accounts for bit-length in field and group elements. $\kappa$ is the statistical security parameter. For Heli, $\gamma$ denotes the fraction of dropped clients from computation. Server Security has three flavors: semi-honest (SH) server(s), malicious privacy (MP) where honest server can detect deviation, and its strengthened version identifiable abort (IA) where the honest server can identify the misbehaving party. 
}
\label{tab:costs}
\resizebox{\textwidth}{!}{%
\begin{tabular}{@{}lccc|ccc|cc@{}}
\toprule
 & \multicolumn{3}{c|}{\textbf{Communication Cost}} 
 & \multicolumn{3}{c|}{\textbf{Computation Cost}} 
 & \textbf{Server} & \textbf{Mal.} \\ 
 \cmidrule(lr){2-4} \cmidrule(lr){5-7} 
 & Server 1 & Server 2 & Client 
 & Server 1 & Server 2 & Client 
 & \textbf{Security} & \textbf{Clients}\\ \midrule
\textsc{Elsa} 
      & $n + n L t^{2}$ 
      & $n + n L t^{2}$ 
      & $tL + \lambda + \kappa+L \log L$ 
      & $n(tL + \lambda + \kappa)$ 
      & $n(tL + \lambda + \kappa)$ 
      & $tL + \lambda + \kappa $ 
      & MP & Yes \\[0.5em]
\textsc{Prio} 
      & $ntL\lambda+tL\log tL$ 
      & $ntL+tL\log tL$ 
      & $tL$ 
      & $n tL\log tL $ 
      & $n tL\log tL $ 
      & $tL\log tL$ 
      & SH & Yes \\[0.5em]
\textsc{Prio+} 
      & $n L t $ 
      & $n L t $ 
      & $L t $ 
      & $n L t\lambda$ 
      & $n L t \lambda$ 
      & $L$ 
      & SH & Yes \\[0.5em]
\textsc{Heli} 
      & $n (L+\log t)$ 
      & $(L+\gamma n\log n)$
      & $(L+\log tL)$ 
      & $ntL+L\sqrt{n\cdot 2^t}$
      & $L+ \gamma n$  
      & $tL$ 
      & MP & Yes \\
\midrule 
\textsc{$\sys$} 
      & $ nL$ 
      & $n\lambda$ 
      & $ L+\lambda$ 
      & $nL$ 
      & $n$ 
      & $L$ 
      & SH & No \\[0.5em]
\textsc{$\sys$$^{\text{BP}}$} 
      & $n (L+\log (L+\kappa))$ 
      & $n\lambda$ 
      & $ L+ \log (L+\kappa)+\lambda$ 
      & $nL + n\kappa+ n \lambda$ 
      & $n\lambda$ 
      & $L + \kappa+\lambda$ 
      & SH & Yes \\ [0.5em]
\textsc{$\sys$$^{\text{BP}}$} 
      & $n (L+\log (L+\kappa))$ 
      & $n\lambda$ 
      & $ L+ \log (L+\kappa)+\lambda$ 
      & $nL + n\kappa+ n \lambda$ 
      & $n\lambda+nL$ 
      & $L + \kappa+\lambda$ 
      & IA & Yes \\[0.5em]
\midrule
\textsc{$\sys$$^{\text{LWE}}$} 
      & $n \kappa L$ 
      & $n \cdot \kappa \lambda$ 
      & $\kappa L $ 
      & $n\kappa L$ 
      & $n\kappa \lambda$ 
      & $L \kappa$ 
      & SH & Yes \\ [0.5em]
\textsc{$\sys$$^{\text{LWE}}$} 
      & $n \kappa L$ 
      & $n\lambda$ 
      & $\kappa L $ 
      & $n\kappa L$ 
      & $n\kappa \lambda$ 
      & $L \kappa$ 
      & IA & Yes \\ [0.5em]
\bottomrule
\end{tabular}%
}
\end{table*}
\fi
\ifdefined\IsTPMPC
\bibliographystyle{abbrv}
\bibliography{bib/abbrev2,bib/crypto_crossref,bib/sample-base}
\appendix 
\section*{Appendix Overview}
The Appendix is organized as follows: 
\begin{itemize}
    \item Section~\ref{sec:related_work_comparison} provides a comprehensive comparison with related work and a detailed technical overview, including the protocol sequence diagram. 
    \item[] Section~\ref{sub:tech} provides a technical overview of our construction.
    \item Section~\ref{sec:prelims} establishes the necessary notations and mathematical preliminaries
    \item Section~\ref{sec:design} presents the full design specifications for the \textsc{TAPAS}$^{\text{BP}}$ and \textsc{TAPAS}$^{\text{LWE}}$ variants, covering both semi-honest and malicious server settings. 
    \item Section~\ref{sec:exp} provides detailed experimental evaluations of $\sys$. 
    \item Section~\ref{sec:crypto-app} details the formal cryptographic assumptions, including Learning With Errors (LWE) and Hint-LW.
    \item Section~\ref{sub:lzksa} offers a critical analysis of the LZKSA lattice-based proof system. 
    \item Section~\ref{sec:deferred} contains the complete security proofs for our core theorems and auxiliary lemmas.
\end{itemize}
\section{Detailed Comparison with Related Work}
\label{sec:related_work_comparison}
\else
\subsection{Detailed Comparison with Related Work}
\label{sec:related_work_comparison}
\fi

Table~\ref{tab:costs} summarizes the asymptotic complexity of $\sys$ against state-of-the-art protocols. The column labeled server security admits three values: {MP} (Malicious Privacy) indicates whether honest clients' inputs remain private even if the servers are malicious. SH (semi-honest) denotes privacy against semi-honest servers (but malicious clients), and IA (identifiable abort) means that an honest server can detect malicious behavior and abort, revealing the identity of the misbehaving party. The final column indicates whether the protocol is robust against active client deviations, such as malformed ciphertexts or invalid plaintexts. 
As can be seen: $\sys$ offers the highest level of security among all the other candidates. \textsc{$\sys$} is distinguished by its communication and computational efficiency, making it well-suited for large-scale secure aggregation. Its variants (\textsc{$\sys^{\text{BP}}$} and \textsc{$\sys^{\text{LWE}}$}) offer enhanced security and robustness, with only modest increases in cost, and remain competitive with or superior to existing two-server protocols in most settings. We discuss in detail below:

\begin{description}[nosep, leftmargin=*]
    \item[Comparison with Prio~\cite{Prio} and Prio+~\cite{SCN:AGJOP22}:]  Prio-based systems employ Secret-Shared Non-Interactive Proofs (SNIPs). Clients split their data into additive shares and prove validity properties (e.g., range limits) by evaluating an arithmetic circuit over these shares. Prio+ optimizes this by performing a Boolean conversion to reduce communication. These systems impose a {symmetric computational burden}. To verify the SNIPs, every server must process the full vector of dimension $L$. Further, these works only provide privacy or guarantees against a semi-honest server. 
    \item[Comparison with Elsa~\cite{SP:RSWP23}:] Elsa employs an interactive "trust-but-verify" approach in which servers collectively verify client inputs and reject batches that fail consistency checks.
Elsa achieves Malicious Privacy (MP), a property in which even if a server is malicious, an honest client's input remains private. If a consistency check fails, the honest server cannot distinguish between a malicious client sending bad data and a malicious server falsely reporting a check failure. Elsa also does not provide malicious server security for the final output of the computation.
\item[Comparison with Heli~\cite{Heli}:] Heli is an aggregation protocol based on the Discrete Logarithm Assumption (DL). At its core, Heli builds atop the key-homomorphic PRF-based private stream aggregation~\cite{NDSS:SCRCS11}. 
Heli faces two critical limitations due to its DLP foundation. First, recovering the aggregate requires solving the discrete log problem, which necessitates using small integers to ensure tractability. Second, validity proofs rely on bit-decomposition, expanding the vector size from $L$ to $L \cdot t$ (where $t$ is the bit-width). It should also be noted that a client in Heli communicates only with the heavy load server. 

These are the improvements that $\sys$ offers over Heli: 
\begin{enumerate}
    \item \textbf{Better Security:} $\sys$ achieves \textbf{Identifiable Abort (IA)}. This is a significant improvement over prior work. By enforcing a signed binding phase before aggregation, $\sys$ ensures that any protocol deviation is cryptographically attributable. If the protocol halts, the honest parties output the cheater's specific identity, thereby preventing anonymous stalling attacks.
    \item \textbf{Asymmetry:} Unlike symmetric workload protocols such as \cite{Prio,SCN:AGJOP22,SP:RSWP23}, the asymmetric Heli optimizes for client count $n$ (handling dropouts), $\sys$ optimizes for model dimension $L$. Given that modern FL models are massive ($L \gg n$), our $L$-independent facilitator significantly reduces the system-wide cost in deep learning contexts.
    \item \textbf{Input Independence:} $\sys$ supports arbitrary integer inputs without decryption overhead, as lattice decryption is a linear operation rather than an exponential search. Meanwhile, Heli restricts the input domain for efficient recovery of the aggregate. 
    \item \textbf{Proof Efficiency:} While both Heli and a variant of $\sys$ rely on using bulletproofs~\cite{SP:BBBPWM18} as a core building block for zero-knowledge proofs, $\sys$'s careful optimization ensures that the length of the proof vector does not incur a multiplicative $t$ factor where $t$ is the bit length of the inputs. 
    By integrating Bulletproofs over lattice commitments, our proof sizes scale with $\mathcal{O}(L)$ rather than the bit-expanded $\mathcal{O}(tL)$ required by Heli. Furthermore, we present a fully lattice-based proof system that avoids expensive group exponentiation and relies solely on field arithmetic. 
\end{enumerate}
\end{description}




\begin{figure}[!tb]
    \centering
    \includegraphics[width=\linewidth]{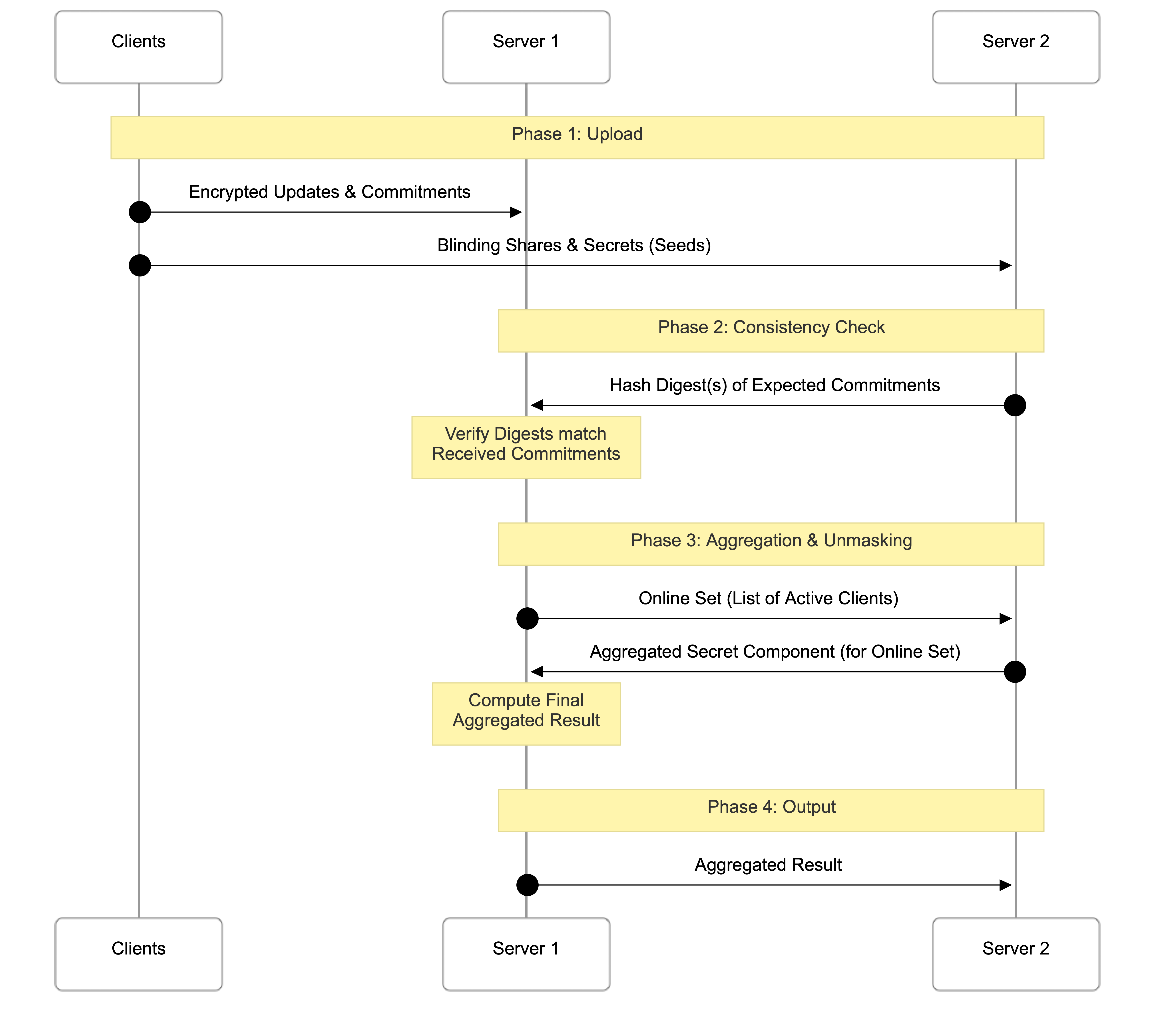}
    \caption{Sequence Diagram of $\sys$ showing client communication with the two servers and additional rounds of communication between the two servers. By updates, we are indicating the model updates that is to be sent to the server and are the inputs to our secure aggregation protocol. The unmasked result is communicated with all the users, in the semi-honest version. The version that is secure against malicious servers also allow the participants to verify the aggregate using commitments. }
    \label{fig:comm-model}
\end{figure}

\subsection{Technical Overview}
\label{sub:tech}
In this section, we outline our approach towards building $\sys$. As is typical in a two-server setting, a client communicates directly with both
servers. However, in a departure from prior work, we will avoid the standard approaches of sharing the inputs between the two servers but instead
 rely on a careful combination of encryption and homomorphism. The protocol's sequence diagram is described in Figure~\ref{fig:comm-model}, proceeding in 4 phases:
\begin{itemize}
    \item \textbf{Phase 1 (Upload):} Clients distribute encrypted inputs to Server 1 ($S_1$) and a much shorter secret component (seed) to Server 2 ($S_2$).
    \item \textbf{Phase 2 (Consistency Check):} $S_1$ validates the received commitments against hash digests provided by $S_2$ to detect malformed inputs.
    \item \textbf{Phase 3 (Unmasking):} Upon identifying the active client set, $S_2$ provides the necessary aggregated secret components to $S_1$, which then computes the final plaintext result.
    \item \textbf{Phase 4 (Output):} The unmasked aggregate is finalized in the semi-honest version. The version that is secure against malicious servers also allow the participants to verify the aggregate using commitments. 
\end{itemize}

\paragraph{Masking the Inputs.} The core technique will employ the Regev's encryption scheme based on the \LWE\ Assumption, where the plaintext modulus is $\bbZ_p$ and the ciphertext modulus is $\bbZ_q$, with $\Delta$ as the encoding parameter where $\Delta = \lfloor q/p \rfloor$.  Each client $C_i$ holds a private input vector $\mathbf{x}_i\in\bbZ_p^L$. $C_i$ samples a secret seed $\mathbf{s}_i \in \mathbb{Z}_q^{\lambda}$ and a noise vector $\mathbf{e}_i$ from a discrete Gaussian distribution $\chi$. The client computes the LWE ciphertext $\mathbf{y}_i$:
    \begin{equation}
        \mathbf{y}_i = \mathbf{A}\mathbf{s}_i + \mathbf{e}_i + \Delta \cdot \mathbf{x}_i \in \mathbb{Z}_q^L
    \end{equation}
    Client $C_i$ transmits $\mathbf{y}_i$ to the Server $S_1$.
For efficient computation, note that $\bfA$ can be set to be the public parameter which can be generated by hashing various common parameters of the scheme. In the Federated Learning setting, this can be the global model weights. By binding $\bfA$ to these global model weights, we can prevents attacks such as those listed by Pasquini~\etal~\cite{CCS:PasFraAte22}.  

\paragraph{Additive Homomorphism.} Now, observe that Regev encryption is homomorphic in the following sense: 
\[
\sum_{i=1}^{n} \bfy_i = \sum_{i=1}^n (\bfA\bfs_i+\bfe_i+\Delta\cdot \bpsain_i)=\bfA\cdot \sum_{i=1}^{n} \bfs_i + \sum_{i=1}^{n} \bfe_i + \Delta\cdot \sum_{i=1}^{n} \bpsain_i
\]
Note that the server needs to recover $\sum_{i=1}^{n} \bpsain_i$. For which it needs \emph{only} $\sum_{i=1}^{n} \bfs_i$. We will sample $\bfe_i$ in such a way that even with $\sum_{i=1}^{n} \bfe_i$, the decryption can still succeed even if the server does not receive $\sum_{i=1}^{n} \bfe_i$. Here $n$ is the number of clients.  

\paragraph{Role of the Second Server.}
The additive homomorphism guarantees that unmasking at server $S_1$ requires only the aggregate mask $\sum_{i=1}^{n} \bfs_i$, not the individual masks $\bfs_i$. Indeed, the knowledge of the individual $\bfs_i$ by the first server will enable it to recover individual inputs, violating the privacy of the inputs. Therefore, one needs to ensure that the first server only learns the sum of $\bfs_i$. This is the role undertaken by the second server $S_2$. Each client sends $\bfs_i$ to the second server, which simply adds up the received values and forwards them to the server. While this seems straightforward, one needs to include additional communication to ensure that both servers agree on a single online set of clients with respect to which both inputs and the seeds are summed. This is for the semi-honest server setting. 

\paragraph{Incorporating zero-knowledge proofs.}
Consistent with prior work, $\sys$ allows the client to generate a Zero-Knowledge Proof (ZKP) attesting to two main properties: (1) the client has executed the masking behavior correctly, and (2) the input vector $\bfx_i$ adheres to the infinity norm bound $\|\bfx_i\|_\infty \leq t$ for a public value $t$ (or the $L_2$ norm of the witness $\bfx_i$). The checks on norm bounds were considered to be sufficient to mitigate poisoning attacks~\cite{NDSS:NasHayCri22,SP:SHKR22}. To enforce these checks, $\sys$ utilizes a ``commit-and-prove'' strategy. We offer two distinct instantiations of this protocol - one based on bulletproofs over groups of prime order where DDH is hard and the other relying on lattice-based hardness assumptions. Each variant offer a tradeoff and users can choose variants based on specific performance or security needs. We now look at the variants in greater detail. We begin by looking at the group-based approach which is a warm up to our lattice-based zero knowledge proof system. The lattice-based proof systems come with inherent challenges that we tackle in novel ways, as described later in this section. Our focus over the next two subsections (Section~\ref{subsub:group},\ref{subsub:lwe}) are to create proof systems that protect against malicious client behavior and against a malicious server $S_2$. We describe how we achieve privacy against malicious server $S_1$ in Section~\ref{subsub:malicious}

\subsubsection{Group-Based Zero Knowledge Proof Instantiation.} 
\label{subsub:group}
This approach utilizes Bulletproofs and performs proofs over a group $\bbG$ of prime order $q'$.  It guarantees the exact satisfaction of the bound $\|\bfx_i\|_\infty \leq t$ (i.e., there is zero slack in the proof). It is highly efficient in communication, using Pedersen commitments, a single group element can commit to an entire vector. It requires computationally expensive group exponentiations. Additionally, it necessitates a group order $q' \gg q$, meaning the proofs cannot occur directly over the smaller LWE modulus $q$. We employ a suite of zero-knowledge building blocks to verify the correctness of LWE-based encryptions. The core of our system relies on {Pedersen Commitments}~\cite{C:Pedersen91} and {Bulletproofs} \cite{SP:BBBPWM18} to handle linear and quadratic relations over committed vectors in $O(\log L)$ communication. We adapt the techniques pointed out in the prior works of ACORN~\cite{USENIX:BGLLMR23} and Armadillo~\cite{CCS:MGKP25}. The proof proceeds in three parts: 

\begin{itemize}
    \item Exact Norm Proofs: To ensure security and correctness, we must prove that witnesses $\mathbf{w}$ satisfy $\| \mathbf{w} \|_2 \leq B$ or $\| \mathbf{w} \|_\infty \leq B$ where $B$ is an integer constant, specific to the underlying use case. Following \cite{EC:GenHalLyu22}, we use Lagrange's Four-Square Theorem to reduce these to algebraic relations. For example, $\| \mathbf{w} \|_2 \leq B$ is proven by showing $B^2 - \|\mathbf{w}\|^2 = \sum_{j=1}^4 \alpha_j^2$ for auxiliary integers $\alpha_j$. 
However, since these relations are checked \textit{modulo $q$}, the prover must also demonstrate that no modular wrap-around has occurred. This necessitates proving that all witness components are "sufficiently small" relative to $q$.
\item Approximate Proof of Smallness: Because exact range proofs for every component are expensive, we use an \emph{approximate proof} as a bridge. This system shows $\| \mathbf{w} \|_\infty < B'$ for some gap $B' \gg B$ to ensure the algebraic norms remain valid over the integers. 
The prover masks the witness with a high-entropy vector $\mathbf{y}$ (sampled according to a statistical security parameter $\kappa$) and uses rejection sampling. The verifier checks a random projection:
\[ \langle \mathbf{R}^\top \mathbf{r}, \mathbf{w} \rangle + \langle \mathbf{r}, \mathbf{y} \rangle = \langle \mathbf{z}, \mathbf{r} \rangle \]
where $\mathbf{R} \in \mathbb{Z}^{\kappa \times L}$. This ensures that the exact norm proofs are sound without requiring the overhead of bit-decomposition for every vector entry and $\bfr$ is a random challenge vector. 
\item Finally, we verify the LWE relation $\mathbf{y} = \mathbf{A}\mathbf{s} + \mathbf{e} + \Delta \mathbf{x} \pmod q$ by collapsing the $L$ equations into a single inner product $\langle \mathbf{y}, \mathbf{r} \rangle$ using a random challenge vector $\mathbf{r}$. This reduced statement is then proven using a linear proof $\Pi_{\text{linear}}$ over the combined witness $(\mathbf{s} \| \mathbf{e} \| \mathbf{x})$.
\end{itemize}

While the Zero-Knowledge Proof (ZKP) convinces $S_1$ that the client's ciphertext $\mathbf{y}_i$ is well-formed with respect to a commitment $\mathbf{c}_i$, we must introduce two additional mechanisms to prevent attacks on the secret keys and the aggregation process.

\begin{itemize}
    \item \textbf{Input Consistency (Preventing Client Equivocation).} 
    To ensure a malicious client does not provide a secret $\mathbf{s}_i$ to $S_2$ that differs from the one committed to $S_1$ during the ZKP phase, we leverage the binding property of the Pedersen commitment.
    \begin{enumerate}
        \item The client sends the explicit opening $(\mathbf{s}_i, r_i)$ to $S_2$.
        \item $S_2$ recomputes the commitment $\mathbf{c}'_i = \mathsf{Com}(\mathbf{s}_i; r_i)$ and sends a cryptographic digest $h_i = H(\mathbf{c}'_i)$ to $S_1$.
        \item $S_1$, who already holds the original commitment $\mathbf{c}_i$ from the ZKP, computes $h'_i = H(\mathbf{c}_i)$ and verifies that $h_i = h'_i$. Under the collision resistance of the hash function, we are guaranteed that consistency is maintained. 
    \end{enumerate}
    This ensures that $S_1$ and $S_2$ hold consistent views of the client's secret. $S_1$ trusts the value via the ZKP verification, while $S_2$ gains assurance via the digest comparison with $S_1$.

    \item \textbf{Verifiable Aggregation (Preventing malicious $S_2$).} 
    Server $S_2$ is tasked with computing the aggregate secret $\mathbf{s}_\Sigma = \sum_{i \in \mathcal{U}} \mathbf{s}_i$ for the valid set of clients $\mathcal{U}$. To prevent a malicious $S_2$ from injecting errors into this sum, we utilize the homomorphic property of Pedersen commitments.
    \begin{enumerate}
        \item In addition to $\mathbf{s}_\Sigma$, $S_2$ computes and sends the sum of the randomness $r_\Sigma = \sum_{i \in \mathcal{U}} r_i$.
        \item $S_1$ computes the commitment to the received aggregate: $\mathbf{c}_\Sigma = \mathsf{Com}(\mathbf{s}_\Sigma; r_\Sigma)$.
        \item $S_1$ verifies that this matches the product of the individual client commitments it stored previously:
        \[ \mathbf{c}_\Sigma \overset{?}{=} \prod_{i \in \mathcal{U}} \mathbf{c}_i \]
    \end{enumerate}
    Since the individual commitments $\mathbf{c}_i$ were verified against the ZKP and cross-checked for consistency in the previous step, $S_1$ is assured that $\mathbf{s}_\Sigma$ is the correct sum of the valid secrets. Even if $S_2$ colludes with a client, they cannot generate a valid opening $(\mathbf{s}_\Sigma, r_\Sigma)$ that satisfies the homomorphic relation without breaking the binding property of the commitment scheme.
\end{itemize}
We will defer a discussion on security against malicious server $S_1$ to later in this exposition. 


\subsubsection{Lattice-Based Zero Knowledge Proof Instantiation.}
\label{subsub:lwe} 
This approach is designed to align more closely with the system's underlying cryptographic assumptions. We first formally show that the LWE masking is inherently a commitment to vector $\bfx_i$ using randomness $\bfs_i,\bfe_i$. We then carefully build a lattice-based zero-knowledge proof specifically tailored to our intended protocol. This includes careful optimizations. All proofs are performed directly over the LWE modulus $q$, avoiding the need for a separate, larger group field. The proof guarantees are looser, proving only that $\|\bfx_i\|_\infty \leq \gamma \cdot t$, where $\gamma > 1$ represents a soundness slack, which is sufficient for most applications.

A central challenge in LWE-based aggregation is proving that the noise vector $\mathbf{e} \in \mathbb{Z}^L$ is small without revealing it. Standard lattice $\Sigma$-protocols operate monolithically: the prover masks the entire vector and performs rejection sampling on the full response $\mathbf{z}$.
This creates a harsh trade-off: to ensure the rejection probability is non-negligible for high dimensions (for instance $L \approx 10^5$), the masking noise $\sigma$ must be set extremely large. This introduces a polynomial "slack" ($\sqrt{L}$), yielding a bound that is far looser than the actual witness, thereby degrading model accuracy. $\sys$ adopts a block-wise rejection sampling paradigm that decouples the protocol's convergence from the vector dimension $L$.
\begin{itemize}
    \item \textbf{Decomposition \& Independent Restarts:} We partition the witness into small blocks of size $d \ll L$. Crucially, we treat the rejection condition for each block independently. If a specific block fails the rejection check (i.e., falls into the tail of the distribution), the prover restarts the protocol \emph{only for that block}, leaving successful blocks untouched. This transforms the convergence time from exponential in $L$ to linear in $L$, allowing us to target significantly tighter distributions. It is to be noted that this decomposition and independent parallel execution of $\Sigma$ could allow a malicious client to use a different $\bfs_j$, per block. However, for our intended application to building a two-server secure aggregation protocol, we have built in additional checks to remedy this attack. 
    \item \textbf{Bimodal Gaussian Targeting:} Leveraging the independent restart policy, we employ Bimodal Rejection Sampling~\cite{C:DDLL13}. By targeting a distribution that matches the shape of the witness-plus-noise, we reduce the required standard deviation $\sigma$. 
    \item \textbf{Fixing LZKSA:} We note that the recent LZKSA protocol~\cite{CCS:LuLu25} attempts a similar efficiency gain but fails to achieve soundness. LZKSA permits norm checks to wrap around modulo $q$, allowing adversaries to hide unbounded inputs. $\sys$ enforces strict integer-based checks over the decomposed blocks, providing the first sound realization of this high-efficiency proof strategy. We discuss more in Section~\ref{sub:lzksa}
\end{itemize}
As before, we need a crucial line of defense against malicious clients sending inconsistent $\bfs_i$ to server $S_2$. We observe that the $\Sigma$-protocol response can be viewed as a commitment to the secret $\bfs$. In particular, if server $S_2$ is given both $\bfs$ and $\bfr_s$, then it can compute the response $\bfz_s = \bfr_s + c \bfs$. $\bfz_s$ can now be hashed to produce a digest, which in turn, can be used by $S_1$ to validate against the choice of $\bfz_s$ it has received. Unfortunately, the $\Sigma$-protocol actually proceeds in $k$ blocks where $k=L/\lambda$ or $k\lambda=L$. In order for $S_2$ to compute each response $\bfz_{j,s}$ for $j=1,\ldots,k$, it will require $\bfr_{j,s}$ for $h=1,\ldots,k$. In other words, $S_2$ would require $k\cdot \lambda = L$ elements in total, a situation that we want to avoid. Instead, we adopt a randomization strategy: we use a randomized linear combination to compress the proof witnesses. The client aggregates randomness vectors across $k$ blocks using verifiably random binary coefficients $\alpha_{j}$ derived from the proof transcript. Specifically, the client computes $\bfr_s^\ast = \sum_{j=1}^{k}\alpha_j\cdot \bfr_s$ and provides $S_2$ with $\bfr_s^\ast$ and $\alpha_j$ values (which in turn can be generated from a single seed). $S_2$ can now compute $\bfz_{s2}^\ast = \bfr_s^\ast + \bfs_i\cdot \sum_{j=1}^{k} \alpha_j c_j $. Note that using $\bfz_{j,s}$, server $S_1$ can compute a $\bfz_{s1}^\ast=\sum_{j=1}^{k} \alpha_j \cdot \bfz_s$. A hash digest, as before, can be used to establish correctness. Thus, we only send one vector $\bfr_s^\ast$, of length $\lambda$, towards verification. 

\paragraph{Using Ajtai Commitments.}
Unfortunately, using the response of a $\Sigma$ protocol as a commitment to $\bfs$ fails when $S_1$ needs to verify if $S_2$ has behaved correctly. 
To resolve this, we leverage the homomorphic properties of the Short Integer Solution (SIS) problem~\cite{Ajtai}. 
We introduce a deterministic commitment (or homomorphic hash) $\bft := \bfB\bfs$, which is binding under the SIS assumption but generally not hiding.
This necessitates augmenting the client's $\Sigma$-protocol, which originally proved knowledge of the LWE secret $\bfs$, to also prove the correct computation of $\bft$ with respect to $\bfs$. Verification then proceeds in two steps. 
First, for \textit{consistency}, $S_2$ recomputes the expected tag $\bft^\ast := \bfB\bfs$ using the secret $\bfs$ received from the client. $S_2$ sends a hash digest of $\bft^\ast$ to $S_1$, who compares it against the digest of the received $\bft$. This ensures that any client participating in the aggregation has provided consistent inputs to both servers.
Second, for \textit{aggregation}, $S_1$ verifies the additive homomorphism by checking if $\sum_i \bft^{(i)} = \bfB\bfs^{\Sigma}$, where $\bfs^{\Sigma}$ is the aggregate sum returned by $S_2$. 
This binding mechanism prevents a malicious $S_2$ from forging an invalid sum; any deviation would break the SIS hardness. While this, as presented, is sufficient for security against malicious server 2, unfortunately for security against malicious server 1 we will need a stronger property of the commitment scheme - that of equivocation, which we will address next. 

\subsubsection{Security Against Malicious Server $S_1$.} 
\label{subsub:malicious} Proving privacy against a malicious server typically requires the simulator to manipulate the aggregate sum. This is a feature commonly observed in single-server protocols~\cite{CCS:BBGLR20,SP:MWAPR23,C:KarPol25,C:BGLRS25} where this is achieved by adding a secondary mask produced by a function $\mathcal{H}$, modeled as a programmable random oracle.  The simulator waits until the adversary chooses the corrupt set, then programs $\mathcal{H}(x)$ to inject a specific correction term into the sum. The core philosophy behind this approach is to ensure that the communication is independent of the vector length as $|x|\ll L$ where $L$ is the length of the vector being masked. 

This same approach can be undertaken for our asymmetric two-server case for malicious $S_1$. For our protocol, the new masking function is $\bfy=\bfA\bfs+\bfe+\Delta \bfx+\cH'(\term{seed})$ for some $\term{seed}$ that is expanded by $\cH'$. The client sends $\term{seed}$ to $S_2$. Once the online set is fixed, $S_2$ sends individual seeds to $S_1$. $S_1$ now applies $\cH'$ on each seed $\term{seed}$ to remove from the sum of the ciphertexts. If $S_1$ is malicious, the simulator waits until the online set is fixed to query the functionality and then correct, with the help of random oracle programming, to ensure that the sum of ciphertexts decrypts to the actual result expected by the adversary controlling $S_1$. However, this approach of using a second mask appears to collapse when clients are required to prove correctness of their inputs where the witness now includes $\term{seed}$ and the proof $\pi$ is to show that $\bfy$ is correctly formed which involves proof of correct computation of $\cH'(\term{seed})$. Implicitly, it appears to ask for a proof $\pi$ to show that there is a witness $x$ such that $\pi=\mathcal{H}(x)$. Leaving aside the fact that zero-knowledge proofs to prove correct computation of hash function~\cite{USENIX:GiaMadOrl16} are expensive, in our case after $S_1$ has verified $\pi$, it receives both the input $x$ and the newly programmed value $y'$ with the expectation that $\pi$ verifies even given witness $x$ and $y'$. Typically, $\pi$ contains a binding commitment to $y$ which might fail with $y'$. Prior works on single-server secure aggregation with input validation~\cite{USENIX:BGLLMR23,CCS:MGKP25} have constructions only based on security against semi-honest server.

In this work, we circumvent this barrier by introducing an equivocable commitment scheme, restoring the simulator's ability to program the oracle consistently. Further, we also avoid the expensive proof of correct computation of the hash function, by using the asymmetry. Instead of proving the computation of the hash function inside the ZK proof (which is both expensive and restrictive), we decouple the \textit{value} from its \textit{derivation}. We now look at the way in which we achieve privacy against malicious servers in both our variants:

\begin{itemize}
    \item \textbf{Group-Based Construction:} In our group based version (See Figure~\ref{fig:mal-bp-protocol}), the client commits to the output and not the preimage. In other words, the client commits directly to the mask value $\bfw$ using a Pedersen commitment $C_w = \Commit(\bfw;r_w)$. The ZK proof $\pi$ only attests that the ciphertext $\bfy$ is well-formed relative to the value inside $C_w$ (i.e., $\bfy = \dots + \bfw$). Crucially, $\pi$ implies nothing about how $\bfw$ was generated. $S_2$ is provided both $\term{seed}$ and $\bfw$ which it can use to generate $C_w^\ast$ on its own. Again, a digest allows $S_2$ to convince $S_1$ that the commitment $C_w$ is valid against $S_2$'s own inputs. Once the online set is fixed, malicious server $S_1$ now expects $\term{seed}$ along with $R_w=\sum r_w^{(i)}$. Fortunately, the Pedersen commitment is equivocal - specifically the simulator chooses the generators such that it knows the discrete logarithm of all the bases. If $\mathcal{H}(\term{seed}^{(i)}$ was programmed to be ${\bfw^{(i)}}^\ast$, for some client $i$, the simulator (which controls $S_2$) can compute the correct ${r_w^{(i)}}^\ast$ such that $\Commit(\bfw^{(i)};r_w^{(i)})=\Commit({\bfw^{(i)}}^\ast;{r_w^{(i)}}^\ast)$. This allows for the simulator to successfully equivocate on the commitment. 

\item \textbf{Lattice-Based Construction:} Recall that the basic Ajtai commitment, defined as $\bft:=\bfB\bfs$ where $\bfB$ is a public matrix $\bfs$ is the message, and $\bft$ is the commitment, is not equivocal. To achieve malicious security in the lattice setting, we face the same challenge: ensuring inputs are correctly formed (i.e., $\bfw = \mathcal{H}'(\term{seed})$) without binding the simulator to the hash output before the adversary is defined.
Standard approaches that prove the hash computation in Zero-Knowledge are incompatible with the Programmable Random Oracle Model (PROM) used for simulation.

We resolve this by employing an equivocal SIS-based commitment scheme that decouples the validity of the mask from its generation.
Instead of using a standard binding matrix $\bfB$, we construct a wide, underdetermined commitment matrix $\bfB = [\bfB_w | \bfB_\rho]$.
\begin{enumerate}
    \item \textbf{Trapdoor Setup:} The simulator generates $\bfB_\rho$ with an associated lattice trapdoor $\mathbf{T}_\rho$ (using algorithms like \cite{GPV08}). This trapdoor allows finding short preimages for any syndrome.
    \item \textbf{Equivocal Opening:} In the real protocol, clients simply commit to their mask $\bfw$ as $\bft_w = \bfB_w \bfw + \bfB_\rho \bfrho$. The ZK proof $\pi$ only attests that the ciphertext $\bfy$ is consistent with the committed $\bfw$, without proving anything about $\mathcal{H}'(\rho)$.
\end{enumerate}
As before, $S_2$ compares digests with $S_1$ to ensure that the $\term{seed}$ eceived by $S_2$ satisfies the commitment $\bft_w$ received by $S_2$. Crucially, this structure enables the simulator to "break" binding. 
It generates a dummy commitment $\bft_w$ at the start. Later, when it programs the oracle $\mathcal{H}'(\term{seed}) \to \bfw_{\text{target}}$ to fix the aggregate sum, it uses the trapdoor $\mathbf{T}_\rho$ to sample a short randomness $\bfrho'$ such that $\bfB_\rho \bfrho' = \bft_w - \bfB_s \bfw_{\text{target}}$.
This valid opening makes the transcript indistinguishable from an honest execution where the mask was genuinely derived from the seed, satisfying the simulation requirements.
\end{itemize}

\section{Preliminaries} 
\label{sec:prelims}
We consider a federated learning framework. There exist $n$ clients, with each client $C^{(i)}$ owning a dataset $D^{(i)}$. The server holds the ML model $\Theta$. 
In FL, the server first sends $\Theta$ to the clients, and each client trains on its local dataset $D^{(i)}$ to generate the updated weights $m^{(i)}$. Meanwhile, the server computes the average of these model updates $\set{m^{(i)}}$ to update its global model to $\Theta'$. In the next iteration, $\Theta'$ is sent back for the next update. The goal is to enable the server to compute the average and, thereby, the new model $\Theta'$, while keeping the model weights confidential. To this end, we employ an intermediary party, the second server.  

\paragraph{Communication and Threat Model.} 
$\sys$ operates within the traditional two-server setting with $n$ clients~\cite{SP:RSWP23,Prio,SCN:AGJOP22}, where the threat model assumes that the two servers cannot collude. We designate these as servers $S_1$ and $S_2$. In contrast to prior two-server protocols, our second server $S_2$ performs a lighter computational load than server $S_1$, introducing an asymmetry in the server roles. Additionally, unlike traditional two-server protocols where both servers receive client inputs proportional to the input length, $\sys$ reduces the communication received by $S_2$ to be independent of the input length. 






\paragraph{Modeling Security.} Like prior works, we identify two kinds of behavior. A semi-honest party is one that adheres to the protocol description but remains curious and attempts to compromise the inputs of the honest parties. A malicious/corrupted party is one that is allowed to arbitrarily deviate from the protocol description. $\sys$ is designed to offer the following feature: honest clients' inputs remain private as long as at least one server is honest and does not collude with the other server. Furthermore, once the servers begin interacting, we guarantee an identifiable abort, i.e., malicious clients (or the server) are detected. 

\paragraph{Notations.}
For a distribution $X$, we use $x\getsr X$ to denote that $x$ is a random sample
drawn from the distribution $X$. We denote by $\bfu$ a vector and by $\bfA$ a matrix. For a set $S$ we use $x\getsr S$ to denote
that $x$ is chosen uniformly at random from the set $S$. By $[n]$ for some integer $n$, we denote the set $\set{1,\ldots,n}$. 
Due to space constraints, we defer a refresh of some relevant cryptographic preliminaries to the appendix. They can be found in Section~\ref{sec:crypto-app}

\subsection{Lattice-based Masking}
\label{sub:lattice}
\label{sub:lwe}
Our masking scheme is based on the original LWE-based encryption Scheme~\cite{STOC:Regev05,ACM:Regev09}, albeit extended to the setting where the plaintext space is over $\bbZ_p$ for some integer $p<q$ with $q$ being the ciphertext modulus. In more detail, the scheme is parameterized by a security parameter $\lambda$, a plaintext modulus $p$, and
a ciphertext modulus $q$, and number of LWE samples $L$. Given a secret key $\bfs \xleftarrow{\$}
\bbZ_q^{\lambda}$, the encryption of a vector $\bfx \in \bbZ_p^L$ is 
\[ (\bfA, \bfc) := (\bfA, \bfA \bfs + \bfe + \Delta \cdot \bfx), \] where $\bfA
\xleftarrow{\$} \bbZ_q^{L\times \lambda}$ is a random matrix ($m>\lambda$), $ \bfe
\xleftarrow{\$} \chi^L$ is an error vector and $\chi$ is a discrete Gaussian
distribution, and $\Delta := \floor{q/p}$. Decryption is computed as $(\bfc -
\bfA \bfs) \bmod q$ and rounding each entry to the nearest multiples of
$\Delta$, and then divide the rounding result by $\Delta$. The decrypted result
is correct if entries in $\bfe$ are less than $\Delta/2$.

\subsection{Commitment Schemes and Zero-Knowledge Proofs}
\label{sub:com}
In this section, we review the necessary cryptographic building blocks, with particular emphasis on zero-knowledge proof systems and commitment schemes. 

\subsubsection{Group Based Constructions}
\allowdisplaybreaks

\paragraph{Pedersen Commitment.} 
We first review the scalar Pedersen Commitment~\cite{C:Pedersen91}. Let $\mathbb{G}$ be a group of prime order $q$ where the Discrete Logarithm problem is hard. Let $g,h \xleftarrow{\$} \mathbb{G}$. The commitment scheme $\Commit((g,h), x; r) = g^x \cdot h^r$ is conditionally binding (under the discrete logarithm assumption) and perfectly hiding. 

We extend this to the vector setting. For a randomly chosen generator vector $\mathbf{G}=(g_1, \ldots, g_L) \in \mathbb{G}^L$ and a base $h \in \mathbb{G}$, the commitment to a vector $\mathbf{x} = (x^{(1)}, \ldots, x^{(L)}) \in \mathbb{Z}_q^L$ with randomness $r \in \mathbb{Z}_q$ is defined as:
\[
\Commit((\mathbf{G}, h), \mathbf{x}; r) = h^r \cdot \prod_{i=1}^L {g_i}^{x_i}.
\]
Going forward, we suppress the public parameters and denote this commitment as $\Commit_{\mathbf{G}}(\mathbf{x};r)$, or simply $\Commit(\mathbf{x})$.

\paragraph{Inner-product Zero-Knowledge Proof.}
An inner product proof allows a prover to convince a verifier that $\langle \mathbf{a}, \mathbf{b} \rangle = c$ for committed vectors $\mathbf{a}, \mathbf{b}$ and a public scalar $c$. Bulletproofs~\cite{SP:BBBPWM18} provide a non-interactive protocol for this relation using the Fiat-Shamir transform, achieving $O(\log L)$ proof size and $O(L)$ prover/verifier cost for vector length $L$. We denote by $\pi_{\text{ip}}$ the inner product proof where both $\bfa$ and $\bfb$ are witnesses and by $\pi_{\text{linear}}$ we denote the version where only one of the vectors (say $\bfb$) is a witness. 



\paragraph{Group-Based Proof of Correct LWE Encryption.} 
Due to space constraints, we defer a review of the techniques to the appendix in Section~\ref{sub:lwe-app}. We obtain the following Lemma:
\begin{lemma}[Cost of Norm proofs]
\label{lem:norm-proofs}
\leavevmode
\begin{itemize}
  \item An $L_2$ norm bound requires one quadratic proof of length $(m+4)$ and one approximate proof of length $(m+4)$ (instantiated via a linear proof of length $m+4+\kappa$).
  \item An $L_\infty$ norm bound requires one quadratic proof of length $(5m)$ and one linear proof of length $(5m)$.
  \item With one-bit leakage~\cite{EC:GenHalLyu22}, lengths reduce to $(\ell+3)$ for $L_2$ and $(4\ell)$ for $L_\infty$.
\end{itemize}
\end{lemma}

Combining these components, we construct the proof system $\mathbb{CS}_{\text{enc}}$:
\begin{equation}
    \label{eq:cs-enc}
    \begin{aligned}
    \mathbb{CS}_{\text{enc}} : \{ &\mathsf{zk.io}: (\Commit(\mathbf{s}), \Commit(\mathbf{x}), \Commit(\mathbf{e})), \\
        &\mathsf{zk.stmt}: \mathbf{y} = \mathbf{A} \mathbf{s} + \mathbf{e} + \Delta \mathbf{x}, \ \|\mathbf{x}\|_2 < B_{\mathbf{x}}(L_2), \\
        &\qquad \qquad \|\mathbf{x}\|_\infty < B_{\mathbf{x}}(L_\infty), \ \|\mathbf{e}\|_\infty < B_{\mathbf{e}}(L_\infty), \\
        &\mathsf{zk.wit}: (\mathbf{s}, \mathbf{x}, \mathbf{e}) \}. 
    \end{aligned}
\end{equation}

\subsubsection{Lattice-Based Commitments and Zero-knowledge Proofs}
\paragraph{Lattice-based Commitment.}
We first formalize the LWE-based encryption from Section~\ref{sub:lwe} as a homomorphic commitment scheme.

\begin{definition}[Commitment Scheme]
Let $q,p \in \mathbb{Z}$ with $p < q$, and $L,\lambda \in \mathbb{N}$ be positive integers.
Let $\bfA \in \mathbb{Z}_q^{L \times \lambda}$ be a public matrix sampled uniformly at random.
Let $\Delta = \lfloor q/p \rfloor$ be the public scaling factor.
The message space is $\mathcal{X} = \mathbb{Z}_p^L$.
To commit to a message $\bfx \in \mathcal{X}$:
\begin{enumerate}
    \item Sample secret $\bfs \sample D_s$ and error $\bfe \sample D_e$ from distributions over short vectors, bounded by $\|\bfs\|_\infty \le B_s$ and $\|\bfe\|_\infty \le B_e$.
    \item Compute the commitment:
    \[ \bfc = \bfA\bfs + \bfe + \Delta \cdot \bfx \pmod{q}. \]
\end{enumerate}
The output is the commitment $\bfc \in \mathbb{Z}_q^L$. The opening is the tuple $(\bfx, \bfs, \bfe)$.
\end{definition}

We now have the following lemma (whose proof can be found in \hyperref[proof:com]{Section~\ref{sec:deferred}}), which states that the scheme satisfies the standard properties of \emph{computational hiding} and \emph{computational binding}.

\begin{restatable}[Hiding and Binding]{lemma}{commitmentSecurity}
\label{lem:commitment-security}
\label{lem:binding}
The commitment scheme is computationally hiding under the Decision-LWE assumption and computationally binding under the Short Integer Solution (SIS) assumption.
\end{restatable}

\paragraph{Bimodal Rejection Sampling.}
To optimize the efficiency of our Zero-Knowledge proofs, we employ Bimodal Rejection Sampling~\cite{C:DDLL13}. This technique symmetrizes the rejection condition, allowing the standard deviation $\sigma$ of the masking noise to be set proportional to the witness norm, rather than exponentially larger.

\begin{lemma}[Bimodal Rejection Sampling]
\label{lem:bimodal}
Let $V \subseteq \mathbb{Z}^d$ be a set of vectors such that for all $\bfv \in V$, $\|\bfv\|_2 \le T$. Let $\sigma \ge \xi T$ be a standard deviation parameter for some constant $\xi$ (typically $\xi \approx 11$).
Consider the following algorithm for a prover who holds a shift vector $\bfv \in V$:
\begin{enumerate}
    \item Sample $\bfz \sample D_{\sigma}^d$.
    \item Compute acceptance probability $p(\bfz) = \frac{1}{M \exp(-\|\bfv\|^2/(2\sigma^2))} \cosh\left(\frac{\langle \bfz, \bfv \rangle}{\sigma^2}\right)$, where $M = \exp\left(\frac{T^2}{2\sigma^2}\right)$.
    \item Output $\bfz$ with probability $p(\bfz)$, otherwise restart.
\end{enumerate}
Then, the distribution of the output $\bfz$ is statistically close to the centered Gaussian distribution $D_{\sigma}^d$, independent of the shift $\bfv$. The expected number of trials is $M$.
\end{lemma}

\paragraph{Block-wise Zero-Knowledge Protocol.}
We construct a Zero-Knowledge Argument of Knowledge (ZKAoK) for the relation underlying the commitment. To improve bandwidth efficiency and soundness, we decompose the large message $\bfx$ into $k$ blocks.

\begin{definition}[Relation $\mathcal{R}_{\term{Block}}$]
The prover proves knowledge of witness components $\{(\bfs_j, \bfe_j, \bfx_j)\}_{j=1}^k$ such that for all blocks $j \in [k]$:
\begin{enumerate}
    \item \textbf{Validity:} $\bfy_j \equiv \bfA_j\bfs_j + \bfe_j + \Delta \bfx_j \pmod q$.
    \item \textbf{Norms:} $\|\bfs_j\|_2 \le \beta_{s}, \|\bfe_j\|_2 \le \beta_{e}, \|\bfx_j\|_2 \le \beta_{x}$.
\end{enumerate}
\textit{Note:} While the honest prover uses a consistent secret $\bfs_j = \bfs$ across all blocks, this relation allows distinct $\bfs_j$. Global consistency is enforced by the separate 2-Server Consistency Check (see Protocol $\sys^{\text{LWE}}$).
\end{definition}

\begin{remark}[Trade-off: Communication vs. Modulus Size]
\label{rem:tradeoff}
Standard lattice proofs incur a soundness slack proportional to the square root of the witness dimension. For a model of size $L \approx 10^6$, a monolithic proof would require a large modulus $q$ to accommodate this slack, inflating the size of every ciphertext element.
By decomposing the proof into $k$ blocks of size $d=L/k$, we reduce the slack to $\propto \sqrt{d}$, allowing for a significantly smaller $q$.
The cost of this approach is the redundant transmission of the secret response $\bfz_{j,s}$ for each block (an overhead of $k \cdot \lambda$ elements). However, since the model dimension $L$ is orders of magnitude larger than the secret dimension $\lambda$ ($L \gg \lambda$), the bandwidth savings from reducing the modulus $q$ on the massive vector $\bfy$ far outweigh the cost of the extra secret responses.
\end{remark}

The protocol $\Pi_{\term{Enc}}$ is detailed in Figure~\ref{fig:optimized-proto}, a standard $\Sigma$-protocol, adapted to the setting of lattices. As before, let $\lambda$ be the LWE secret dimension and let $d = L/k$ be the block size. Let $C$ be the $L_\infty$ bound of the challenge $c_j$ (e.g., $c_j \in \{-1, 0, 1\}$ implies $C=1$). Let $\xi \approx 1$ be the bimodal tuning constant.
Further parameters are defined in Table~\ref{tab:params}.  Here, $\tau$ is the tail cut factor (typically 1.1) used to determine the acceptance bound $\beta$, ensuring that honest responses sampled from the Gaussian distribution are accepted with overwhelming probability.

\begin{table}[!htbp]
\caption{Parameters for the Zero-knowledge Argument of Knowledge}
\label{tab:params}
\centering
\renewcommand{\arraystretch}{1} 
\setlength{\tabcolsep}{9pt}      
\resizebox{0.75\textwidth}{!}{\begin{tabular}{@{} l c c c c @{}}
\toprule
\textbf{Component} & \textbf{Base} ($L_\infty$) & \textbf{Shift Norm} $(T)$ & \textbf{Sigma} $(\sigma = \xi T)$& \textbf{Verification} ($\beta$)\\
\midrule
Secret $\bfs$ 
    & $B_s$ 
    & $C \sqrt{n} B_s$ 
    & $\xi C \sqrt{\lambda} B_s$ 
    & $\tau \sqrt{\lambda} \sigma_s$ \\
Error $\bfe$ 
    & $B_e$ 
    & $C \sqrt{d} B_e$ 
    & $\xi C \sqrt{d} B_e$ 
    & $\tau \sqrt{d} \sigma_e$ \\
Message $\bfx$ 
    & $B_x$ 
    & $C \sqrt{d} B_x$ 
    & $\xi C \sqrt{d} B_x$ 
    & $\tau \sqrt{d} \sigma_x$ \\
\bottomrule
\end{tabular}}
\end{table}

\begin{figure*}[!t]
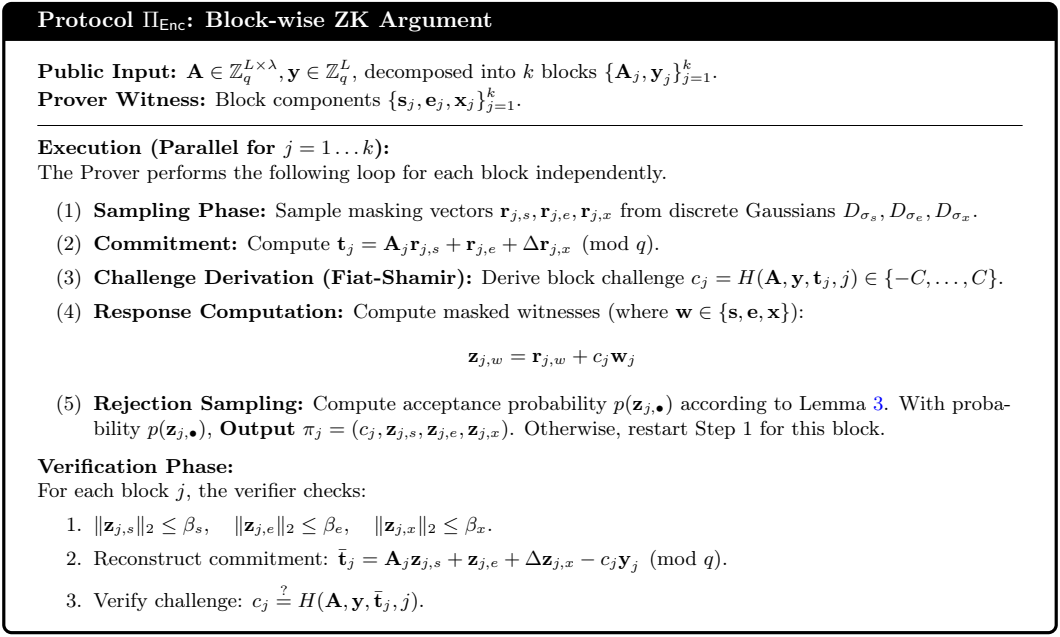

 \centering
 \resizebox{0.85\textwidth}{!}{\begin{tcolorbox}[colback=white, colframe=black, title={\textbf{Protocol $\Pi_{\term{Enc}}$: Block-wise ZK Argument}}]
 \centering
 \small
 \begin{varwidth}{\textwidth}
 \textbf{Public Input:} $\bfA \in \bbZ_q^{L \times \lambda}, \bfy \in \bbZ_q^L$, decomposed into $k$ blocks $\{\bfA_{j}, \bfy_{j}\}_{j=1}^k$. \\
 \textbf{Prover Witness:} Block components $\{\bfs_j, \bfe_{j}, \bfx_{j}\}_{j=1}^k$.
 \vspace{2mm}
 \hrule
 \vspace{2mm}
 
 \textbf{Execution (Parallel for $j = 1 \dots k$):} \\
 The Prover performs the following loop for each block independently.
 
 \begin{enumerate}[label=({\arabic*})]
     \item \label{step1} \textbf{Sampling Phase:}
     Sample masking vectors $\bfr_{j,s}, \bfr_{j,e}, \bfr_{j,x}$ from discrete Gaussians $D_{\sigma_s}, D_{\sigma_e}, D_{\sigma_x}$.
     
     \item \label{step2} \textbf{Commitment:}
     Compute $\bft_j = \bfA_{j}\bfr_{j,s} + \bfr_{j,e} + \Delta\bfr_{j,x} \pmod q$.
     
     \item \label{step3} \textbf{Challenge Derivation (Fiat-Shamir):}
     Derive block challenge $c_j = H(\bfA, \bfy, \bft_j, j) \in \{-C, \dots, C\}$.
     
     \item \label{step4} \textbf{Response Computation:}
     Compute masked witnesses (where $\bfw \in \{\bfs, \bfe, \bfx\}$):
     \[ \bfz_{j,w} = \bfr_{j,w} + c_j \bfw_j \]
     
     \item \label{step5} \textbf{Rejection Sampling:}
     Compute acceptance probability $p(\bfz_{j,\bullet})$ according to Lemma~\ref{lem:bimodal}.
     With probability $p(\bfz_{j,\bullet})$, \textbf{Output} $\pi_j = (c_j, \bfz_{j,s}, \bfz_{j,e}, \bfz_{j,x})$.
     Otherwise, restart Step 1 for this block.
 \end{enumerate}
 
 \textbf{Verification Phase:} \\
 For each block $j$, the verifier checks:
 \begin{enumerate}
     \item $\|\bfz_{j,s}\|_2 \le \beta_s, \quad \|\bfz_{j,e}\|_2 \le \beta_e, \quad \|\bfz_{j,x}\|_2 \le \beta_x$.
     \item Reconstruct commitment: $\bar{\bft}_j = \bfA_{j}\bfz_{j,s} + \bfz_{j,e} + \Delta\bfz_{j,x} - c_j \bfy_{j} \pmod q$.
     \item Verify challenge: $c_j \stackrel{?}{=} H(\bfA, \bfy, \bar{\bft}_j, j)$.
 \end{enumerate}
 \end{varwidth}
 \end{tcolorbox}}
 \caption{Lattice-based Zero-Knowledge Argument of Knowledge for Relation $\mathcal{R}_{\term{Block}}$. Here $C$ is the maximum value of the challenge set. We sample challenges between $\{-C,\ldots,C\}$.}
 \label{fig:optimized-proto}
\end{figure*}

\begin{restatable}[Security of $\Pi_{\term{Enc}}$]{theorem}{securityEnc}
The protocol $\Pi_{\term{Enc}}$ is a complete, Special Honest-Verifier Zero-Knowledge (SHVZK) Argument of Knowledge for the relation $\mathcal{R}_{\term{Block}}$.
\end{restatable}
\noindent The proof is deferred to \hyperref[proof:sec-enc]{Section~\ref{sec:deferred}}.

\subsection{Simulation-Based Proof of Security}
\label{sub:sim}
\label{sec:functionality}
Our proof follows the standard \emph{simulation-based paradigm}~\cite{Goldreich,Lindell}. The objective is to show that any real-world adversary attacking our protocol can be simulated by an adversary~$\cS$ in an \emph{ideal world}, where a \emph{trusted party}~$\cT$ securely computes a function~$F$ on the clients' behalf. In our setting, this function simply aggregates (sums) the clients' input vectors.

We consider an adversary $\cA$ that can corrupt up to $\eta \cdot n$ clients and at most one server. The ideal-world execution proceeds as follows, specialized to the case where only the server receives the final output:

\begin{enumerate}[label=(\alph*)]\itemsep0em
    \item \label{stepa} Each honest client submits its input to the trusted party~$\cT$.
    \item \label{stepb} The simulator $\cS$ determines which corrupted clients provide valid inputs and which ones abort.
    \item \label{stepc} If a server is corrupted,~$\cS$ may decide either to abort the execution or to allow it to proceed.
    \item \label{stepd} If no abort occurs,~$\cT$ computes the function~$F(X)$ and sends the result to the server.
    \item \label{stepe} If the server is honest, it outputs whatever value it receives from $\cT$.
\end{enumerate}
We formally define the functionality in Figure~\ref{fig:ideal-func-secagg}. We allow for the following adversarial behavior:
\begin{itemize}\itemsep0em
    \item Clients can be malicious, i.e., deviate from protocol descriptions arbitrarily. 
    \item We require exactly one server to be honest, and the other server to be adversarially controlled. Looking ahead, we model this adversarial control to be either semi-honest or malicious and still offer privacy. Observe that a malicious server can ``forge'' or provide inputs on behalf of the honest clients. To avoid this trivial attack, we assume that clients' messages are signed. We assume that there exists a Public Key Infrastructure (PKI) that allows servers to verify and validate signatures. Later, we additionally leverage these signatures to achieve security with identifiable abort. 
\end{itemize}
\begin{figure}[!tb]
\centering
\begin{tcolorbox}[
    colback=white, 
    colframe=black, 
    boxrule=0.8pt, 
    arc=0pt, 
    left=6pt, right=6pt, top=6pt, bottom=6pt,
]
\centerline{\textbf{Functionality $\cF_{\text{SecAgg}}$}}
\vspace{4pt}

The functionality is parameterized by the maximum allowed dropout fraction $\delta \in [0,1)$ and corruption fraction $\eta\in [0,1)$ (subject to $\delta+\eta<1/3$) and the iteration identifier $\lab$. It interacts with clients $i \in [n]$ and a simulator $\cS$.

\vspace{4pt}
\hrule
\vspace{6pt}

\begin{description}[leftmargin=1em, style=nextline]

\item[Corruption Setup:] 
    Upon receiving a message \textsf{(Corrupt, $\lab$, $\cC_\text{corr}$, $\{\bfx^{(i)}\}_{i \in \cC_\text{corr}}$)} from $\cS$:
    Check if the size of the corrupted set $|\cC_\text{corr}| \le \eta n$. If true, record $C$ as the set of corrupted clients and record their respective adversarial inputs $\{\bfx^{(i)}\}_{i \in \cC_\text{corr}}$ into the input set $X$. 
    
    \item[Honest Input Collection:] 
    Upon receiving a message \textsf{(Input, $\lab$, $\bfx^{(i)}$)} from an honest client $i \in [n] \setminus \cC_\text{corr}$, record $\bfx^{(i)}$ in the set of client inputs $X$.

    \item[Dropout Specification:] 
    Upon receiving a message \textsf{(Drop, $\lab$, $D$)} from the $\cS$ (where $D \subseteq [n]$), record the set of dropped-out clients $D$. Here, we mark every client dropped out from the computation as ``corrupt''.\footnotemark

    \item[Aggregation \& Output:] 
    Upon receiving a message \textsf{(Aggregate, $\lab$)} from the servers, define: $\cC:=\{i|i\in[n]\wedge i\not \in D\wedge\term{Valid}(\bfx^{(i)})\}$ where $\term{Valid}$ is a boolean function that determines whether the input $\bfx^{(i)}$ satisfies predetermined, publicly known constraints (such as bounded $L_\infty$ and $L_2$ norms) and evaluate the aggregate based on the recorded inputs $X$ as:
    
    \begin{equation}
    \cF_{\text{SecAgg}}^{D,\delta} =
    \begin{cases}
        \displaystyle \sum_{i \in \cC} \bfx^{(i)}, & \text{if } |\cC| \geq (1-\delta-\eta) n,\\[1.5em]
        \bot, & \text{otherwise.}
    \end{cases}
    \end{equation}
    
    Output \textsf{(Result, $\lab$, $\cF_{\text{SecAgg}}^{D,\delta}$)} to the evaluating servers.
\end{description}
\end{tcolorbox}
\caption{The ideal functionality for Secure Aggregation handling dropouts and validity constraints.}
\label{fig:ideal-func-secagg}
\end{figure}
\footnotetext{Prior works~\cite{SP:RSWP23,Heli} have also treated honest clients who drop out or fail to participate as corrupt. While this approach is typical in secure computation literature, Fitzi et al.~\cite{C:FitHirMau98} propose a mixed model where an adversary can induce communication failures among honest clients without those clients being flagged as corrupt.}
\section{Designing $\sys$}
\label{sec:design}
We will now describe how we build $\sys$.

\subsection{Building $\sys^{\text{BP}}$}
\label{sub:g-zkp}
We now focus on constructing $\sys^{\text{BP}}$, which instantiates the proof system using Bulletproofs (BP). We first build the protocol secure against a semi-honest server and then show how to extend it for a malicious server.

\subsubsection{Semi-Honest Servers and Malicious Clients.}
In this section, we outline the protocol for semi-honest servers, while allowing for malicious clients. Figure~\ref{fig:sh-bp-protocol} describes the protocol $\sys_{\text{SH}}^{\text{BP}}$ for secure aggregation with semi-honest servers. Each client masks its input $\mathbf{x}^{(i)}$ using LWE encryption, computes commitments to its secrets, and generates a NIZK proof of correct encryption and bounded input. The client sends the ciphertext, proof, and commitments to Server 1, and the blinding vector with its opening to Server 2.

To ensure consistency, Server 2 hashes its commitment and sends the digest to Server 1, which verifies both the hash and the proof. Only clients passing both checks are included in the valid set $\mathcal{V}$. Server 2 aggregates the secrets from $\mathcal{V}$ and sends the sum to Server 1, which computes the final aggregate. The protocol requires that enough clients are valid to ensure robustness.

In the semi-honest version, the final communication of the sum from Server 1 to Server 2 is omitted from the protocol description. Observe that the server 1 receives $\bfs$ from Server 2 which can be viewed as a leakage on each individual LWE sample. Therefore, the earlier argument that $\bfA\bfs+\bfe_i$ is pseudorandom and can mask the input value is insufficient. This has been observed in prior works~\cite{USENIX:BGLLMR23,C:KarPol25,C:BGLRS25}. Intuitively, the Hint-LWE (See Definition~\ref{def:hint-lwe}) assumption states that $\mathbf{y}_0$ appears pseudorandom to an adversary, even when given randomized leakage on the secret and error vectors. Kim et al.~\cite{C:KLSS23b} demonstrate that solving Hint-LWE is no easier than solving the LWE problem. We formally have the following theorem statement:

\begin{restatable}[Simulatability under Hint-LWE]{theorem}{thmSimulatability}
\label{thm:simulatability}
Let $(\lambda, L, q, \chi)$ be a secure LWE instance where $\chi$ is a discrete Gaussian distribution with standard deviation $\sigma$, and let $\chi'$ be a discrete Gaussian distribution with standard deviation $\sigma' = \sigma/\sqrt{2}$. Let $\mathbf{A} \in \mathbb{Z}_q^{L \times \lambda}$ be a public matrix, $\mathbf{X} \in \mathbb{Z}_q^L$ be a message vector, and $n \geq 2$ be the number of parties.

Then there exists a probabilistic polynomial-time simulator $\textsf{Sim}$ such that for all probabilistic polynomial-time adversaries $\mathcal{A}$, the following distributions are computationally indistinguishable under the Hint-LWE assumption:
\begin{align*}
\mathcal{D}_R &= \left\{ \bfs, \mathbf{y}_1, \ldots, \mathbf{y}_n : \begin{array}{l}
\bfs_i \getsr \mathbb{Z}_q^\lambda, \mathbf{e}_i\getsr\chi^L \text{ for } i \in [n], \\
\bfs = \sum_{i=1}^n \bfs_i \bmod q, \\
\mathbf{y}_i = \mathbf{A}\bfs_i + \mathbf{e}_i + \Delta\mathbf{x}_i \text{ for } i \in [n], \\
\end{array} \right\}, & 
\mathcal{D}_{\textsf{Sim}} &= \left\{ \textsf{Sim}(\mathbf{A}, \mathbf{X}) \right\}
\end{align*}
That is, for all PPT adversaries $\mathcal{A}$:
\(
\left| \Pr[\mathcal{A}(\mathcal{D}_R) = 1] - \Pr[\mathcal{A}(\mathcal{D}_{\textsf{Sim}}) = 1] \right| \leq \textsf{negl}(\lambda)
\)
where the probability is taken over the randomness of the distributions and the adversary.
\end{restatable}
\noindent The proof of Theorem~\ref{thm:simulatability} can be found in \hyperref[proof:sim]{Section~\ref{sec:deferred}}
\begin{restatable}{theorem}{shbp}
    \label{thm:atlas-lwe-sh}
    Let $\lambda$ and $\kappa$ be the computational and statistical security parameters, and let $n$ be the number of clients chosen for summation. Let $L$ be the vector length. The construction defined in Figure~\ref{fig:sh-client} securely realizes the functionality $\cF_{\text{SecAgg}}(X)$ (defined in Figure~\ref{fig:ideal-func-secagg}) achieving semi-honest security for $X=\{\bfx_{i}\}_{i\in[n]\setminus (\cC_{\text{corr}}\cup \cC_{\text{drop}}}$ where $\cC_{\text{corr}}$ is the set of corrupted clients and $\cC_{\text{drop}}$ denotes the set of dropped clients, under the Hint-\LWE\ assumption and the existence of SHVZK proofs and secure commitments.
\end{restatable}
\noindent The proof of Theorem~\ref{thm:atlas-lwe-sh} can be found in \hyperref[proof:sh-bp]{Section~\ref{sec:deferred}}.

\begin{figure*}[!t]
\centering
\resizebox{0.85\textwidth}{!}{\begin{tcolorbox}[
    colback=white, 
    colframe=black, 
    boxrule=0.8pt, 
    arc=2pt, 
    left=4pt, right=4pt, top=4pt, bottom=4pt,
    title={\textbf{Protocol $\sys_{\text{SH}}^{\text{BP}}$: Group-Based Aggregation (Semi-Honest)}},
    fonttitle=\bfseries\small,
    colbacktitle=gray!15!white, 
    coltitle=black
]
\footnotesize
\textbf{Setup:} Matrix $\mathbf{A} \in \mathbb{Z}_q^{L \times \lambda}$
\hfill
\textbf{Inputs:} Client $i$ has $\mathbf{s}^{(i)}, \mathbf{e}^{(i)}, \mathbf{x}^{(i)}$.

\tcbline

\textbf{Phase 1: Client Computation}
\begin{itemize}[nosep, leftmargin=*]
     \item \textbf{Masking \& Commitments:} Client $i$ computes:
        \( \bfy_i = \bfA\cdot\bfs_i + \bfe_i + \Delta \cdot \bfx_i \bmod q \)
        Compute commitments: $C_s^{(i)} = \Commit(\bfs_i; r_s^{(i)})$, $C_x^{(i)} = \Commit(\bfx_i; r_x^{(i)})$, $C_e^{(i)} = \Commit(\bfe_i; r_e^{(i)})$.
        
        \item \textbf{Proof Generation:} Generate NIZK proof $\pi_{\text{Enc}}^{(i)}$ satisfying Eq~\ref{eq:cs-enc}.
        \item \textbf{Distribution:}
        \begin{itemize}
            \item \textbf{To Server 1:} Sends $\bfy_i$, Proof $\pi_{\text{Enc}}^{(i)}$, and Commitments $\{C_s^{(i)}, C_x^{(i)}, C_e^{(i)}\}$.
            \item \textbf{To Server 2:} Sends Blinding Factor $\bfs_i$ and Opening $\{r_s^{(i)}\}$.
        \end{itemize}
\end{itemize}

\tcbline

\textbf{Phase 2: Consistency Check}
\begin{itemize}[nosep, leftmargin=*]
    \item \textbf{$S_2 \to S_1$ (Digest):} $S_2$ computes $z^{(i)} = \mathcal{H}(\Commit(\bfs_i;r_s^{(i)})$ and sends list $\{z^{(i)}\}$ to $S_1$.
    \item \textbf{$S_1$ Verification:} $S_1$ checks $z^{(i)} \stackrel{?}{=} \mathcal{H}(C_s^{(i)})$ and verifies proof $\pi^{(i)}$.
    \item \textbf{Filtering:} Define valid set $\mathcal{V} = \{ i \mid \text{Checks Pass} \}$. Send $\mathcal{V}$ to $S_2$.
\end{itemize}

\tcbline

\textbf{Phase 3: Aggregation}
\begin{itemize}[nosep, leftmargin=*]
    \item \textbf{$S_2 \to S_1$:} Compute aggregate secret $\mathbf{s}^{\Sigma} = \sum_{i \in \mathcal{V}} \mathbf{s}^{(i)}$ and send to $S_1$.
    \item \textbf{$S_1$ Finalize:} Compute $\mathbf{x}^{\Sigma} = \mathsf{Decode}\left( (\sum_{i \in \mathcal{V}} \mathbf{y}^{(i)}) - \mathbf{A} \mathbf{s}^{\Sigma} \right)$.
\end{itemize}
\end{tcolorbox}}
\caption{Protocol $\Sys_{\text{SH}}^{\text{BP}}$ for Semi-Honest Servers. Let $L,\lambda,q$ be integers such that they are secure parameter choices of Hint LWE Assumption.}
\label{fig:sh-bp-protocol}
\label{fig:sh-server}
\label{fig:sh-client}
\end{figure*}

       \subsubsection{Security Against Malicious Server.}
\label{sub:malicious}
In a two-server setting, at most one server can be malicious. However, 
the asymmetry of $\sys{}$ implies the following:

\begin{itemize}
    \item The client's input is handled solely by $S_1$. Therefore, a malicious server $S_1$, while attempting to leverage the corrupt clients and the semi-honest server $ S_2$, attempts to compromise the privacy of the inputs of the honest clients. Additionally, it can provide an incorrect sum of the inputs it has recovered. 
    \item Meanwhile, a malicious server $S_2$ can provide inconsistent inputs to server 1, affecting the output of the computation. 
\end{itemize}
We leverage this asymmetry to propose a protocol in which any communication from $S_2$ can be validated by $S_1$, with the aid of the client. However, a semi-honest $S_1$ cannot determine whether the inconsistency is due to $S_2$'s maliciousness or the client's. $S_1$ can preemptively abort or drop the client from the computation. This approach is similar to those taken by prior works, such as Elsa~\cite{SP:RSWP23}. However, in a major improvement, we can guarantee security with identifiable abort once the servers begin interacting. The protocol is formally described in Figure~\ref{fig:mal-bp-protocol}. 

\begin{figure*}[!ht]
\centering
\resizebox{0.85\textwidth}{!}{\begin{tcolorbox}[
    colback=white, 
    colframe=black, 
    boxrule=0.8pt, 
    arc=2pt,
    left=2pt, right=2pt, top=2pt, bottom=2pt,
    title={\textbf{Protocol $\sys_{\text{MAL}}^{\text{BP}}$: Group-Based Aggregation (Malicious)}},
    fonttitle=\bfseries\small,
    colbacktitle=gray!15!white, 
    coltitle=black
]
\footnotesize

\textbf{Parameters:} Matrix $\mathbf{A}$, Group $\mathbb{G}$, Hash $\mathcal{H}'$. 
\hfill 
\textbf{Input:} Client $i$ has $\mathbf{s}^{(i)}, \mathbf{e}^{(i)}, \mathbf{x}^{(i)}$.

\tcbline

\textbf{Phase 1: Client Computation}
\begin{itemize}[nosep, leftmargin=*]
    \item \textbf{Setup:} Sample seed $\term{seed}^{(i)} \in \{0,1\}^\lambda$. Derive mask $\mathbf{w}^{(i)} = \mathcal{H}'(\term{seed}^{(i)})$.
    \item \textbf{Masking:} Compute $\mathbf{y}^{(i)} = \mathbf{A}\mathbf{s}^{(i)} + \mathbf{e}^{(i)} + \Delta \mathbf{x}^{(i)} + \mathbf{w}^{(i)} \bmod q$.
    \item \textbf{Commitments:} Compute $C_s^{(i)}, C_x^{(i)}, C_e^{(i)}$ (as above) and $C_w^{(i)} = \Commit(\mathbf{w}^{(i)}; r_w^{(i)})$.
    \item \textbf{Proof:} Generate $\pi^{(i)}$ proving well-formedness of $\mathbf{y}^{(i)}$ and consistency of commitments.
    \item \textbf{Submission:}
    \begin{itemize}[nosep, leftmargin=1.5em]
        \item \textbf{To $S_1$:} $\mathbf{y}^{(i)}$, $\pi^{(i)}$, $\{C_s^{(i)}, C_x^{(i)}, C_e^{(i)}, C_w^{(i)}\}$.
        \item \textbf{To $S_2$:} $\mathbf{s}^{(i)}$, $\term{seed}^{(i)}$, Openings $\{r_s^{(i)}, r_w^{(i)}\}$.
    \end{itemize}
\end{itemize}

\tcbline

\textbf{Phase 2: Consistency Check}
\begin{itemize}[nosep, leftmargin=*]
    \item \textbf{$S_2 \to S_1$:} Recompute commitments $C_s' = \Commit(\mathbf{s}^{(i)}; r_s^{(i)})$ and $C_w' = \Commit(\mathcal{H}'(\term{seed}^{(i)}); r_w^{(i)})$. Send digests $z_s^{(i)} = \mathcal{H}(C_s')$ and $z_w^{(i)} = \mathcal{H}(C_w')$ to $S_1$.
    \item \textbf{$S_1$ Verify:} Check $z_s^{(i)} \stackrel{?}{=} \mathcal{H}(C_s^{(i)})$ and $z_w^{(i)} \stackrel{?}{=} \mathcal{H}(C_w^{(i)})$. If valid, add $i$ to $\mathcal{V}$.
\end{itemize}

\tcbline

\textbf{Phase 3: Secure Aggregation}
\begin{itemize}[nosep, leftmargin=*]
    \item \textbf{$S_2 \to S_1$:} Aggregate $\mathbf{s}^{\Sigma} = \sum \mathbf{s}^{(i)}$ and send seeds $\{\term{seed}^{(i)}\}_{i \in \mathcal{V}}$. Also send aggregate openings $R_s = \sum r_s^{(i)}$ and $R_w = \sum r_w^{(i)}$.
    \item \textbf{$S_1$ Verify:} Check homomorphic property:
    \[ \Commit(\mathbf{s}^{\Sigma}; R_s) \stackrel{?}{=} \prod_{i \in \mathcal{V}} C_s^{(i)} \quad \text{and} \quad \Commit(\sum \mathcal{H}'(\term{seed}^{(i)}); R_w) \stackrel{?}{=} \prod_{i \in \mathcal{V}} C_w^{(i)} \]
    \item \textbf{Finalize:} $\mathbf{x}^{\Sigma} = \mathsf{Decode}(\sum \mathbf{y}^{(i)} - \mathbf{A}\mathbf{s}^{\Sigma} - \sum \mathbf{w}^{(i)})$.
\end{itemize}
\end{tcolorbox}}
\caption{Protocol $\sys_{\text{MAL}}^{\text{BP}}$. Adds mask $\mathbf{w}$ and homomorphic verification for malicious security.}
\label{fig:mal-bp-protocol}
\label{fig:combined-mal-protocol}
\label{fig:mal-server}
\label{fig:mal-client}
\end{figure*}
We focus on the malicious security against server 1. The security against malicious server 2 follows from the binding property of commitments and collision-resistance of hash functions. A malicious server $S_2$ could either collude with corrupt clients in the digest phase where it provides convincing digests to convince $S_1$. However, due to the binding property of commitments, $S_2$ will not be able to provide correct openings to the malicious $\bfs{(i)}$ or $\term{seed}^{(i)}$. We have the following theorem whose proof can be found in \hyperref[proof:mal-bp]{Section~\ref{sec:deferred}}.

\begin{restatable}[Security Against Malicious Server 1]{theorem}{maliciousSOne}
\label{thm:malicious-s1}
Assuming the hardness of the hint Learning With Errors (LWE) problem, the existence of simulation-sound Non-Interactive Zero-Knowledge Arguments of Knowledge (NIZK-AoK) in the Random Oracle Model, and the security of perfectly hiding, computationally binding, additively homomorphic commitments (e.g., Pedersen), and collision-resistant hash functions, the protocol $\sys^{\text{BP}}$ (Figure~\ref{fig:mal-client}) securely realizes the ideal functionality $\cF_{\text{SecAgg}}$ (see Figure~\ref{fig:ideal-func-secagg}) against a malicious Server 1 and a subset of malicious clients, provided Server 2 is semi-honest and non-colluding.
\end{restatable}

\subsection{Building $\sys^\text{LWE}$}

We now focus on building $\sys^{\text{LWE}}$, which instantiates the proof system by relying on lattice-based proof systems. We first build the protocol secure against a semi-honest server and then show how to extend it for a malicious server.
\subsubsection{Security Against Semi-Honest Servers and Malicious Clients.}
We define the semi-honest protocol in Figure~\ref{fig:sh-lattice}. The process flow is similar to the bulletproof but with critical technical challenges. 

A core security requirement is ensuring \emph{input consistency}: a malicious client must not be able to prove correctness of a valid masking secret $\bfs$ to Server~1 ($S_1$) via the ZK-AoK, while revealing a different, incorrect secret $\bfs' \neq \bfs$ to Server~2 ($S_2$). A standard mitigation, having $S_2$ recompute the full set of ZK responses $\{\bfz_{j,u,s}\}$ to verify consistency against $S_1$'s view, would require the client to send the underlying randomness $\{\bfr_{j,u,s}\}$ for all $k$ blocks. Since $k = L/\lambda$, this would result in communication complexity scaling linearly with the vector dimension $L$, negating the protocol's efficiency goals. We now detail two approaches towards circumventing this bottleneck:

\paragraph{Folding Argument via Randomized Aggregation.} To circumvent this bottleneck, we employ a \emph{randomized aggregation technique}. Instead of verifying block-wise consistency, we check consistency over a random linear combination of the blocks. 
\begin{enumerate}
    \item \textbf{Randomized Compression:} The client aggregates the randomness vectors across all $k$ blocks into a single compact vector $\bfr_{u,i}^* = \sum_{j=1}^k \alpha_{j,u} \cdot \bfr_{j,u,s}^{(i)}$ for each repetition $u$. The coefficients $\alpha_{j,u}$ are derived from a PRG seeded by the proof transcript $\rho_i$, ensuring they are verifiably random and bound to the proof sent to $S_1$.
    \item \textbf{Split Verification:} 
    \begin{itemize}
        \item $S_1$ computes the \emph{aggregated proof response} $\bfz_{u, \text{s1}}^*$ by applying the same linear combination to the proof components $\bfz_{j,u,s}^{(i)}$ it received.
        \item $S_2$ computes the \emph{expected aggregate response} $\bfz_{u, \text{s2}}^*$ using the compressed randomness $\bfr_{u,i}^*$ and the secret $\bfs_i$ it received.
    \end{itemize}
    \item \textbf{Consistency:} The servers verify that $\bfz_{u, \text{s1}}^* = \bfz_{u, \text{s2}}^*$ by using a hash function to produce a digest.
\end{enumerate}
This reduces the consistency check traffic from $O(\kappa L)$ to $O(\kappa \cdot \lambda)$, making it independent of the vector dimension $L$ while retaining binding security due to the unpredictability of $\alpha_{j,u}$. We argue that a malicious client cannot successfully satisfy the consistency check at $S_1$ and $S_2$ simultaneously while providing different inputs.

\begin{restatable}[Folding via Randomized Aggregation]{lemma}{bindinglemma}
\label{lem:binding-aggregation}
Let $\Delta \mathbf{s} = \mathbf{s} - \mathbf{s}' \neq \mathbf{0}$ be the difference between the secret used in the proof to $S_1$ and the secret sent to $S_2$. For a client to pass the consistency check for all $t$ repetitions, they must present fake aggregated randomness $\mathbf{r}^{*\prime}_{u,i}$ to $S_2$ satisfying:
\begin{equation}
    \label{eq:fake-rand}
    \mathbf{r}^{*\prime}_{u,i} = \mathbf{r}^*_{u,i} + c^*_{u,i} \cdot \Delta \mathbf{s}
\end{equation}
where $\mathbf{r}^*_{u,i}$ is the valid aggregated randomness and $c^*_{u,i} = \sum_{j=1}^k \alpha_{j,u} c_{j,u}$ is the aggregated challenge. The probability that the norm bound check $\|\mathbf{r}^{*\prime}_{u,i}\|_2 \le \tau \sigma_s \sqrt{k \lambda}$ holds for all $u \in [t]$ is negligible in $\kappa$.
\end{restatable}
We defer the proof to Section~\ref{sec:deferred}.

\paragraph{Seed-based Digest Approach.}Instead of algebraic folding, we utilize a deterministic reconstruction approach. The Client samples a seed $\rho$, which is expanded to produce the masking vectors:
\begin{equation*}
    \mathbf{r}_{j,u,s} := \mathcal{H}_{r}(\rho, u, j, \text{attempt})
\end{equation*}
where $\text{attempt}$ is the index of the first successful rejection sampling iteration for block $j$, in iteration $u$. With $\rho$ and the attempt values, $S_2$ can independently reconstruct the client's state.

\begin{itemize}
    \item \textbf{Protocol:} The Client sends to $S_2$ the secret $\mathbf{s}^{(i)}$, seed $\rho$, challenges $\{c_{j,u}\}$, and an attempt matrix $\mathbf{Att}$. $S_2$ reconstructs $\mathbf{z}_{j,u,s} := \mathbf{r}_{j,u,s} + c_{j,u} \cdot \mathbf{s}^{(i)}$ and computes a digest $h^{(i)} := \mathcal{H}(\{\mathbf{z}_{j,u,s}, c_{j,u}\}_{j,u})$. $S_1$ computes a matching digest using the values received directly from the Client.
    \item \textbf{Complexity:} We eliminate the overhead of random linear combinations and complex norm-bound checks. While this introduces an $O(\kappa \cdot L/\lambda)$ communication cost for the matrix $\mathbf{Att}$ (as opposed to simply sending $O(\kappa \lambda)$ in the previous approach), this matrix is typically sparse and highly compressible, ensuring minimal impact on total bandwidth.
\end{itemize}
The formal proof of security in the semi-honest setting is largely similar to that of the bulletproof-based counterpart. The key idea is that the simulator can emulate the ciphertexts of the honest users by querying the functionality to obtain the sum of their inputs. Further, using the special-honest verifier zero-knowledge proofs, the zero-knowledge proofs can be simulated.

\begin{figure*}[!t]
\centering
\resizebox{0.85\textwidth}{!}{\begin{tcolorbox}[
    colback=white, 
    colframe=black, 
    boxrule=0.8pt, 
    arc=2pt, 
    left=4pt, right=4pt, top=4pt, bottom=4pt,
    title={\textbf{Protocol $\sys_{\text{SH}}^{\text{LWE}}$: LWE Aggregation (Semi-Honest Servers)}},
    fonttitle=\bfseries\small,
    colbacktitle=gray!15!white, 
    coltitle=black
]
\footnotesize
\textbf{Setup:} Public matrix $\mathbf{A} \in \mathbb{Z}_q^{L \times \lambda}$ (parsed as $k$ blocks), Hash $\mathcal{H}$, $\term{PRG}:\{0,1\}^\lambda \to \{0,1\}^{kt}$.
\\\hfill
\textbf{Inputs:} Client $i$ has secret $\mathbf{s}^{(i)} \in \mathbb{Z}_q^\lambda$, error $\mathbf{e}^{(i)}$, input $\mathbf{x}^{(i)}$.

\tcbline

\textbf{Phase 1: Client Computation (Local)}
\begin{itemize}[nosep, leftmargin=*]
    \item \textbf{Masking:} Compute LWE ciphertext $\mathbf{y}^{(i)} = \mathbf{A}\mathbf{s}^{(i)} + \mathbf{e}^{(i)} + \Delta \mathbf{x}^{(i)} \bmod q$.
    \item \textbf{Proof Generation:} :
        Generate $\pi^{(i)}=\left\{ c_{j,u}, \bfz_{j,u,s}^{(i)}, \bfz_{j,u,e}^{(i)}, \bfz_{j,u,x}^{(i)} \right\}_{j,u}$ (as described in Figure~\ref{fig:optimized-proto}) for blocks $j \in \{1,\dots,k\}$ and repetitions $u \in \{1,\dots,t\}$ where $\bfz_{j,u,\bullet}:=\bfr_{j,u,\bullet}+c_{j,u}\cdot \bullet_j$ and $\bullet\in\{s,e,x\}$
    \item \textbf{Aggregation Coefficients:} Derive $\rho^{(i)} = \mathcal{H}(\mathbf{A}, \mathbf{y}^{(i)}, \pi^{(i)})$ and get: $\{\alpha_{j,u}\}_{j,u}\gets\term{PRG}(\rho^{(i)})$
    \item For each repetition $u \in [t]$, aggregate randomness over blocks $j$:
             \( {\bfr_u^{(i)}}^\ast = \sum_{j=1}^k \alpha_{j,u} \cdot \bfr_{j,u,s}^{(i)} \)
             Restart if $\|{\bfr_u^{(i)}}^\ast\|_2>\tau \sigma_s \sqrt{k \lambda}$. 
    \item \textbf{Submission:}
    \begin{itemize}[nosep, leftmargin=1.5em]
        \item \textbf{To Server 1 ($S_1$):} Ciphertext $\mathbf{y}^{(i)}$, Proof commitment $\pi^{(i)}$, Seed $\rho^{(i)}$..
        \item \textbf{To Server 2 ($S_2$):} Secret $\mathbf{s}^{(i)}$, Compact response ${\bfr_u^{(i)}}^\ast$, Seed $\rho^{(i)}$, Challenges $\{c_{j,u}\}$.
    \end{itemize}
\end{itemize}

\tcbline

\textbf{Phase 2: Verification \& Filtering}
\begin{itemize}[nosep, leftmargin=*]
    \item \textbf{$S_2$ Verification:} Reconstruct challenge $\alpha_{j,u}$ from $\rho^{(i)}$. Compute $c_{u,i}^* = \sum_{j=1}^k \alpha_{j,u} c_{j,u}$ for $u \in [t]$.
    \item  If $\|{\bfr_{u}^{(i)}}^*\|_2 \le \tau \sigma_s \sqrt{k \lambda}$, then compute: $\bfz_{u, \text{s2}}^* = {\bfr_{u}^{(i)}}^* + c_{u, i}^* \bfs^{(i)}$.
        \item Send digest $h^{(i)} = \cH(\rho^{(i)},\{c_{j,u}\}_{j,u},\{ \bfz_{u, \text{s2}}^* \}_{u=1}^t)$ to $S_1$.
    \item \textbf{$S_1$ Verification:} $S_1$ does: (1) $\pi^{(i)}$ verification, and (2) computes $\{\bfz_{u,s1}^\ast\}_u$ from $\bfz_{j,u,s}^{(i)}$ and verifies against $h^{(i)}$ received from $S_2$. If valid, add $i \text{ to }\mathcal{V}$.
\end{itemize}

\tcbline

\textbf{Phase 3: Final Aggregation}
\begin{itemize}[nosep, leftmargin=*]
    \item \textbf{Aggregation:} $S_2$ computes $\mathbf{s}^{\Sigma} = \sum_{i \in \mathcal{V}} \mathbf{s}^{(i)}$ and sends to $S_1$.
    \item \textbf{Decryption:} $S_1$ computes $\mathbf{x}^{\Sigma} = \mathsf{Decode}\left( (\sum_{i \in \mathcal{V}} \mathbf{y}^{(i)}) - \mathbf{A} \mathbf{s}^{\Sigma} \right)$.
\end{itemize}
\end{tcolorbox}}
\caption{The Semi-Honest Protocol $\sys_{\text{SH}}^{\text{LWE}}$. Here $t,L,d,k$ are integers such that $L=dk$. Further, $L,\lambda,q$ satisfies the Hint LWE Assumption.} 
\label{fig:sh-lattice}
\end{figure*}

\subsubsection{Security Against Malicious Servers.}
 Before we present the malicious version of the protocol it is necessary
to discuss the additional properties needed in this version. They are:
\begin{itemize}
    \item As shown in the malicious version of $\sys^{\text{BP}}$, we need a secondary mask $\bfw^{(i)}$ that is the output of a hash function $\cH'$ on $\term{seed}^{(i)}$. The purpose of this is two fold: (a) use this to program the random oracle model once the online set is finalized by server $S_1$, and (b) to ensure that the communication from $S_2$ to $S_1$ is independent of the vector length. 
    \item Similar to the other malicious version, we need server $S_2$ to prove to $S_1$ that it has provided the correct $\bfs_{\text{agg}}$, with respect to the online set. Unfortunately, with the information provided to $S_1$, it cannot be convinced. To this end, we need a second commitment that is additively homomorphic, lattice-based and compatible with the rest of the proof protocol. 
    \item Finally, we need an equivocable commitment. This is critical because while we provide a commitment to $\bfw^{(i)}$ and prove that $\bfy^{(i)}$ is correctly computed with respect to this commitment, $S_1$ needs to be convinced that indeed the $\term{seed}^{(i)}$ provided by $S_2$ when expanded opens the commitment to $\bfw^{(i)}$. We need this level of equivocation. To this end we will use the Ajtai commitment~\cite{Ajtai,C:LyuNguPla22}. One can commit to $\bfs_1$ using randomness $\bfs_2$ as $\bfB \bfs_1 + \bfC \bfs_2$ where we require that norm of $\|\bfs_1\|$ and $\|\bfs_2\|$ are small. This is binding under the SIS assumption and hiding aspect of the commitment scheme is based on the indistinguishability of $\bfC,\bfC\bfs_2$ from uniform. A key artifact that will be useful for our protocol is that the dimension of the message $\bfs_1$ does not increase the commitment size. Further, since SIS hardness is fairly independent of the size of the solution, increasing the length of $\bfs_1$ does not impact the security. However, we do require that the norm of $\bfs_1$ is bounded. 
\end{itemize}
We first show how to extend Figure~\ref{fig:optimized-proto} to now work for our desired malicious setting, specifically for the following relation:

\begin{definition}[Relation $\mathcal{R}_{\term{Enc}}^\ast$]
\label{def:relation}
The prover $\mathcal{P}$ proves knowledge of a witness 
$\mathfrak{w} = (\mathbf{s}, \mathbf{e}, \mathbf{x}, \mathbf{w}, \boldsymbol{\rho}_s, \boldsymbol{\rho}_w)$ 
satisfying the following conditions with respect to the public instance 
$\mathfrak{x} = (\mathbf{A}, \mathbf{B}_{s}, \bfB_w,\bfB_\rho, \mathcal{C})$:

\begin{enumerate}\itemsep0em
    \item \textbf{Ciphertext Validity (Payload Masking):} $\mathbf{ct} \equiv \mathbf{A}\mathbf{s} + \mathbf{e} + \Delta (\mathbf{x} + \mathbf{w}) \pmod q$
    where $\Delta = \lfloor q/p \rfloor$.

    \item \textbf{Commitment Binding:} $\cC \equiv \bfB_s \bfs + \bfB_w \bfw + \bfB_\rho \bfrho \pmod q$
    \item \textbf{Norm Constraints:} $\|\mathbf{s}\|_2\leq \beta_s, \|\mathbf{e}\|_2 \le \beta_{s}, 
        \|\mathbf{x}\|_2 \le \beta_{x},
        \|\mathbf{w}\|_2 \le \beta_{w} ,
        \|\boldsymbol{\rho}\| \le \beta_{\rho}$

\end{enumerate}
\end{definition}
The proposed protocol $\Pi^\ast$ (defined in Figure~\ref{fig:unified-proto}) extends the standard lattice-based $\Sigma$-protocol defined from Figure~\ref{fig:optimized-proto}. We introduce two critical enhancements:
(1) We augment the relation to include a witness $\bfw$ (client weights/metadata), proving it satisfies the linear relation $\bfy = \bfA\bfs + \bfe + \Delta \bfw + \Delta \bfx$ without revealing $\bfw$. Looking ahead $\bfw$ will be the output of the programmable hash function. 
(2) To ensure the prover is bound to their specific weight $\bfw$ and gradient $\bfs$ across the protocol, we add an auxiliary compressing Ajtai commitment $\cC = \bfB_s \bfs + \bfB_w \bfw + \bfB_\rho \bfrho$. This binds the ephemeral LWE secrets to a long-term identity or pre-declared metadata.
\begin{figure*}[!h]
\centering
\resizebox{0.95\textwidth}{!}{%
\begin{tcolorbox}[
    colback=white, colframe=black, 
    title={\textbf{Protocol $\Pi^\ast$: Extended $\Sigma$ Protocol}},
    fonttitle=\bfseries\scriptsize,
    colbacktitle=gray!15!white, coltitle=black,
    boxrule=0.7pt, arc=2pt, left=2pt, right=2pt, top=2pt, bottom=2pt
]
\scriptsize
\textbf{Public Parameters (CRS):} 
LWE Matrix $\bfA \in \bbZ_q^{L \times \lambda}$ (blocks $\{\bfA^{(j)}\}$); 
Commitment Matrices $\bfB_{s}\in \bbZ_q^{\lambda \times \lambda}$, $\bfB_{w}\in\bbZ_q^{\lambda \times d}$, $\bfB_{\rho} \in \bbZ_q^{\lambda \times m}$; 
Static Commitments $\{\cC^{(j)}\}_{j=1}^k$, $\cC^{(j)} \in \bbZ_q^\lambda$.

\textbf{Prover Witness:} 
$\bfs \in \bbZ^{\lambda}$, $\{\bfe^{(j)}, \bfx^{(j)}, \bfw^{(j)}\}_{j=1}^k$, $\{\bfrho^{(j)} \in \bbZ^m\}_{j=1}^k$

 \vspace{2mm}
 \hrule
 \vspace{2mm}

\textbf{Execution (for $j = 1 \dots k$ in parallel):}
\begin{enumerate}[itemsep=0pt, leftmargin=*, label=({\arabic*})]
    \item \textbf{Sample:} $\bfr_{j,s} \gets D_{\sigma_s}^\lambda$, $\bfr_{j,e} \gets D_{\sigma_e}^d$, $\bfr_{j,x} \gets D_{\sigma_x}^d$, $\bfr_{j,w} \gets D_{\sigma_w}^d$, $\bfr_{j, \rho} \gets D_{\sigma_\rho}^m$
    \item \textbf{Prover Msgs:} $\bft_{j} = \bfA^{(j)}\bfr_{j,s} + \bfr_{j,e} + \Delta\bfr_{j,x} + \Delta \bfr_{j,w} \pmod q$; $\bfv_{j} = \bfB_{s}\bfr_{j,s} + \bfB_{w}\bfr_{j,w} + \bfB_{\rho}\bfr_{j, \rho} \pmod q$
    \item \textbf{Challenge:} $c_j = H(\bfA, \bfB_{s}, \bfB_{w}, \bfB_{\rho}, \cC^{(j)}, \bft_{j}, \bfv_{j}, j) \in \{-1, 0, 1\}$
    \item \textbf{Response:} 
    $\bfz_{j,s} = \bfr_{j,s} + c_j \bfs$, 
    $\bfz_{j,e} = \bfr_{j,e} + c_j \bfe^{(j)}$, 
    $\bfz_{j,x} = \bfr_{j,x} + c_j \bfx^{(j)}$, 
    $\bfz_{j,w} = \bfr_{j,w} + c_j \bfw^{(j)}$, 
    $\bfz_{j, \rho} = \bfr_{j, \rho} + c_j \bfrho^{(j)}$
    \item \textbf{Rejection Sampling:} If $\|\bfz_{\bullet}\|$ too large, reject and resample; else output $\pi_j = (c_j, \bfz_{j,s}, \bfz_{j,e}, \bfz_{j,x}, \bfz_{j,w}, \bfz_{j, \rho})$
\end{enumerate}

 \vspace{2mm}
 \hrule
 \vspace{2mm}

\textbf{Verification (for $j = 1 \dots k$):}
\begin{enumerate}[itemsep=0pt, leftmargin=*, label=({\arabic*})]
    \item $\bar{\bft}_{j} = \bfA^{(j)}\bfz_{j,s} + \bfz_{j,e} + \Delta\bfz_{j,x} + \Delta \bfz_{j,w} - c_j \bfy^{(j)} \pmod q$
    \item $\bar{\bfv}_{j} = \bfB_{s}\bfz_{j,s} + \bfB_{w}\bfz_{j,w} + \bfB_{\rho}\bfz_{j, \rho} - c_j \cC^{(j)} \pmod q$
    \item Check $c_j \stackrel{?}{=} H(\bfA, \bfB_{s}, \bfB_{w}, \bfB_{\rho}, \cC^{(j)}, \bar{\bft}_{j}, \bar{\bfv}_{j}, j)$ and $\|\bfz_{\bullet}\|_2 \le \beta_{\bullet}$
\end{enumerate}
\end{tcolorbox}
}
\caption{Extended $\Sigma$-protocol for the relation in Definition~\ref{def:relation} using Ajtai Commitment.}
\label{fig:unified-proto}
\end{figure*}

\begin{restatable}[Security Against Malicious Server 1]{theorem}{maliciousSOneLat}
\label{thm:malicious-s1-lat}
Assuming the hardness of the {Hint Learning With Errors (Hint-LWE)} problem, the {Short Integer Solution (SIS)} problem (for the binding of Ajtai commitments), and the existence of a {Random Oracle} $\cH'$, the protocol $\sys^{\text{LWE}}$ securely realizes the ideal functionality $\cF_{\text{SecAgg}}$ in the presence of a malicious Server 1 controlling a static subset of corrupt clients $\bfK$, provided that Server 2 is semi-honest and non-colluding.
\end{restatable}
\noindent The proof is deferred to \hyperref[proof:mal-lat]{Section~\ref{sec:deferred}}.

\begin{figure*}[!ht]
\centering
\resizebox{0.85\textwidth}{!}{\begin{tcolorbox}[
    colback=white, 
    colframe=black, 
    boxrule=0.8pt, 
    arc=2pt,
    left=4pt, right=4pt, top=4pt, bottom=4pt,
    title={\textbf{Protocol $\sys_{\text{LWE}}^{\text{MAL}}$: LWE Aggregation (Malicious Security)}},
    fonttitle=\bfseries\small,
    colbacktitle=gray!15!white, 
    coltitle=black
]
\footnotesize

\textbf{Parameters:} LWE Matrix $\bfA \in \bbZ_q^{L \times \lambda}$, Commitment Matrices $\bfB_{s}\in \bbZ_q^{\lambda \times \lambda}, \bfB_{w}\in\bbZ_q^{\lambda \times d}, \bfB_{\rho} \in \bbZ_q^{\lambda \times m}$, Hashes $\cH:\{0,1\}^\ast\to \bbZ_q,\cH':\{0,1\}^\ast \to \bbZ_p^L$. 
\hfill
\textbf{Inputs:} Client $i$ has $\mathbf{s}^{(i)}, \mathbf{e}^{(i)}, \mathbf{x}^{(i)}$.

\tcbline

\textbf{Phase 1: Client Computation}
\begin{itemize}[nosep, leftmargin=*]
    \item \textbf{Setup:} Sample $\term{seed}^{(i)}\getsr \{0,1\}^\lambda, \bfrho^{(i)} \getsr D_\sigma^{km}$. Derive vector $\mathbf{w}^{(i)} \leftarrow \cH'(\term{seed}^{(i)})$.
           \item \textbf{LWE Encryption:} Compute $\bfy^{(i)} = \bfA\bfs^{(i)} + \bfe^{(i)} + \Delta \bfx^{(i)} + \Delta \bfw^{(i)} \bmod q$.
        \item \textbf{Commitment:} Compute $\cC_j^{(i)} = \bfB_s \bfs^{(i)} + \bfB_w \bfw_j^{(i)} + \bfB_\rho \bfrho_j^{(i)}$ where $\bfw^{(i)}=[\bfw_1^{(i)}|\ldots|\bfw_k^{(i)}]$ and $\bfrho^{(i)}=[\bfrho_1^{(i)}|\ldots|\bfrho_k^{(i)}]$
        \item \textbf{Proof Generation:} Generate ZK proof (based on Figure~\ref{fig:unified-proto}) $\pi^{(i)} = \left\{ c_{j,u}, \bfz_{j,u,s}^{(i)}, \bfz_{j,u,e}^{(i)}, \bfz_{j,u,x}^{(i)},\bfz_{j,u,w}^{(i)}, \bfz_{j,u, \rho}^{(i)} \right\}_{j,u} $ for the relation between $\bfy^{(i)}$ and $\cC_j^{(i)}$.
        \item \textbf{Submission:}
        \begin{itemize}[leftmargin=*]
            \item \textbf{To $S_1$:} Ciphertext $\bfy^{(i)}$, Commitments $\{\cC_j^{(i)}\}$, Proof $\pi^{(i)}$. 
            \item \textbf{To $S_2$:} Secret $\bfs^{(i)}$, Randomness $\{\bfrho_j^{(i)}\}$, $\term{seed}^{(i)}$.
        \end{itemize}
\end{itemize}

\tcbline

\textbf{Phase 2: Consistency Check ($S_2$-aided)}
\begin{itemize}[nosep, leftmargin=*]
            \item \textbf{$S_2$ Verification:} Reconstruct $\bar{\bfw}_j^{(i)} = \cH'(\term{seed}^{(i)})$. Recompute $\bar{\cC}_j^{(i)} = \bfB_s \bfs^{(i)} + \bfB_w \bar{\bfw}_j^{(i)} + \bfB_\rho \bfrho_j^{(i)}$.
        \item \textbf{Digest Exchange:} $S_2$ computes individual digests $h_{S2}^{(i)} = \cH(\{\bar{\cC}_{j}^{(i)}\}_{j})$ for all $i$ and sends list $\{h_{S2}^{(i)}\}_i$ to $S_1$.
        \item \textbf{$S_1$ Verification:} $S_1$ computes local digests $h_{S1}^{(i)} = \cH(\{\cC_{j}^{(i)}\}_{j})$.
    \item \textbf{Filtering:} Construct valid set $\mathcal{V} = \{ i \mid \text{Digests match} \land \text{Verify}(\pi^{(i)}) = 1 \}$.
\end{itemize}

\tcbline

\textbf{Phase 3: Secure Aggregation}
\begin{itemize}[nosep, leftmargin=*]
    \item \textbf{$S_2 \to S_1$:} Send aggregated secret $\mathbf{s}^\Sigma = \sum_{i \in \mathcal{V}} \mathbf{s}^{(i)}$ and seeds $\{\term{seed}^{(i)}\}_{i \in \mathcal{V}}$.
    \item \textbf{$S_1$ Finalize:}
    \begin{itemize}[nosep, leftmargin=1em]
        \item Reconstruct aggregate mask $\mathbf{w}^\Sigma = \sum \mathcal{H}'(\term{seed}^{(i)})$.
        \item Decrypt: $\mathbf{x}^{\Sigma} = \mathsf{Decode}\left( \sum_{i \in \mathcal{V}} \mathbf{y}^{(i)} - \mathbf{A} \mathbf{s}^\Sigma - \Delta \mathbf{w}^\Sigma \right)$.
    \end{itemize}
\end{itemize}
\end{tcolorbox}}
\caption{Protocol $\sys_{\text{LWE}}^{\text{MAL}}$. Secure against Malicious server. Let $L,k,m,d,t,\lambda$ be integers such that such that $L=k\cdot d,m=\omega(\lambda \log q)$. and $q,L,\lambda$ satisfy the Hint-LWE Assumption. }
\label{fig:mal-protocol}
\end{figure*}

\subsection{Malicious Security with Identifiable Abort}
\label{sec:identifiable-abort}

Existing efficient aggregation protocols, such as Elsa~\cite{SP:RSWP23}, typically achieve only \emph{Security with Abort} (SwA). In these systems, if a server detects an inconsistency (e.g., a checksum failure), it must abort to preserve privacy, but it cannot cryptographically prove \emph{who} caused the failure. This creates a ``framing'' vulnerability: a malicious server can falsely report a failure to halt the protocol, effectively censoring honest clients without being identified.

We demonstrate that $\sys$ achieves the stronger property of \emph{Security with Identifiable Abort} (IA) by incorporating digital signatures. This ensures that any abort is accompanied by cryptographic evidence identifying the cheater (whether a client or a server).

\paragraph{Ideal Functionalities: SwA vs. IA.}
Let $n$ be the number of clients and $\mathcal{A}$ be the adversary controlling a subset of corrupt clients $\mathcal{C} \subset [n]$.

\begin{itemize}
    \item \textbf{Standard Security with Abort ($\mathcal{F}_{\text{SwA}}$):} 
    If $\mathcal{A}$ sends a malformed input, the functionality outputs $\bot$ to all parties.
    \emph{Deficiency:} Honest parties learn that \emph{someone} cheated, but learn nothing about the identity of the malicious party.
    
    \item \textbf{Identifiable Abort ($\mathcal{F}_{\text{IA-Agg}}$):} 
    If the functionality detects an invalid input or protocol violation from client $j \in \mathcal{C}$ (resp. from server $S_i$), it halts and outputs $(\bot, j)$ (resp. $(\bot,S_i)$
    \emph{Advantage:} The server $S_{1-i}$ (and other honest parties) obtains the identity. If it is client $j$, then the servers can evict the malicious client and then restart the protocol.
\end{itemize}

\paragraph{The Mechanism.} 
The core intuition is to enforce a \emph{Binding Agreement Phase} before aggregation. Clients sign their inputs, and servers must commit to the set of valid signatures they have received \emph{before} processing data. This prevents a malicious server from initially accepting a client's data and later claiming it was malformed to force an abort. Let $\Sigma = (\mathsf{KeyGen}, \mathsf{Sign}, \mathsf{Verify})$ be an EU-CMA secure signature scheme. The protocol proceeds as follows:

\begin{enumerate}[leftmargin=*]
    \item \textbf{Authenticated Submission.} 
    Client $C_i$ signs their payload. For a message $m_{i, \ell}$ destined for Server $\ell$, the client sends $\bar{m}_{i, \ell} = (m_{i, \ell}, \sigma_{i, \ell})$, where $\sigma_{i, \ell} \leftarrow \mathsf{Sign}(sk_i, m_{i, \ell} || \ell)$.

    \item \textbf{The Handshake (Binding Agreement).} 
    Before exchanging any user data, the servers exchange their sets of locally valid clients, $\mathcal{V}_1$ and $\mathcal{V}_2$.
    \begin{itemize}
        \item They compute the intersection $\mathcal{V} = \mathcal{V}_1 \cap \mathcal{V}_2$.
        \item \emph{Critical Step:} Each server signs and exchanges the set $\mathcal{V}$.
    \end{itemize}
    By signing $\mathcal{V}$, Server $\ell$ explicitly attests: ``I possess a validly signed message from every client in $\mathcal{V}$.''

    \item \textbf{Blame Attribution.} 
    During aggregation, if a consistency check fails for Client $i \in \mathcal{V}$, the servers enter a \emph{Blame Phase} by exchanging the signed transcripts $(m_{i,1}, \sigma_{i,1})$ and $(m_{i,2}, \sigma_{i,2})$.
    
    \begin{itemize}
        \item \textbf{Case A: Malicious Client (Equivocation).} 
        If both signatures $\sigma_{i,1}, \sigma_{i,2}$ are valid but the messages are inconsistent (e.g., $m_{i,1}$ and $m_{i,2}$ do not sum to a valid vector), Client $i$ is identified as malicious. The signed conflicting messages serve as proof of misbehavior.
        
        \item \textbf{Case B: Malicious Server (Framing/Forgery).} 
        If Server 1 claims an error but cannot produce a valid signature $\sigma_{i,1}$ for Client $i$, Server 1 is identified as malicious (having violated the Handshake agreement). Similarly, if Server 1 provides a message that is actually \emph{consistent} with Server 2's message (contradicting the error claim), Server 1 is identified as malicious for raising a false alarm.
    \end{itemize}
\end{enumerate}

By leveraging the non-repudiation of signatures, $\sys$ ensures that a malicious server cannot simply assert a fault; it must cryptographically prove it.
\begin{table}[!htbp]
\centering
\renewcommand{\arraystretch}{1.2}
\setlength{\tabcolsep}{3pt}
\caption{\textbf{Scalability Analysis ($L=2^{16}$).} 
Direct comparison of computation time (in seconds) between Elsa and our proposed protocol as the number of clients ($n$) increases. 
\textbf{Bold} indicates the best performance in each category.}
\label{tab:scalability_comparison_restructured}
\resizebox{0.5\textwidth}{!}{\begin{tabular}{c|cc|cc|cc}
\toprule
\multirow{2}{*}{\textbf{Clients ($n$)}} & \multicolumn{2}{c|}{\textbf{Client (s)}} & \multicolumn{2}{c|}{\textbf{Server 1 (s)}} & \multicolumn{2}{c}{\textbf{Server 2 (s)}} \\ 
 & Elsa & \textbf{Ours} & Elsa & \textbf{Ours} & Elsa & \textbf{Ours} \\ \midrule
10  & 0.550 & \textbf{0.197} & 0.908 & \textbf{0.563} & 0.912 & \textbf{0.015} \\ 
20  & 0.805 & \textbf{0.200} & 1.178 & \textbf{0.762} & 1.180 & \textbf{0.031} \\ 
30  & 1.003 & \textbf{0.193} & 1.384 & \textbf{0.778} & 1.380 & \textbf{0.048} \\ 
40  & 1.650 & \textbf{0.198} & 1.752 & \textbf{1.388} & 1.751 & \textbf{0.063} \\ 
50  & 1.707 & \textbf{0.193} & 2.098 & \textbf{1.433} & 2.098 & \textbf{0.077} \\ 
100 & 3.139 & \textbf{0.202} & 3.743 & \textbf{2.270} & 3.734 & \textbf{0.162} \\ \bottomrule
\end{tabular}}
\end{table}
\section{Experiments}
\label{sec:exp}
In this section, we microbenchmark the performance of $\sys_{\text{MAL}}^{\text{LWE}}$. We implemented our protocol in Rust (v1.75) to leverage Elsa's codebase~\cite {SP:RSWP23}. All experiments were conducted on a single \textbf{AWS c5.9xlarge} instance equipped with 36 vCPUs (Intel Xeon Platinum 8000-series @ 3.0 GHz) and 72 GiB of RAM. The statistical security parameter was fixed at $\kappa = 128$ bits. The challenges were drawn from $\{-1,0,1\}$. The computational security parameter was also set at 128. 

We choose Elsa~\cite{SP:RSWP23}
as our baseline. This protocol is the current ``best'' protocol offering
malicious privacy as a security guarantee while also being efficient for large vector lengths $L$. Additionally, we use the current state-of-the-art Rust implementation Prio~\cite{Prio} dubbed ``prio3''~\cite{libprio_rs} to create an implementation to measure the running time. We observe that $\sys_{\text{MAL}}^{\text{LWE}}$ naively offers $L_2$ norm bound check, albeit with soundness slack. We compare only with the Elsa version that does not provide this additional $ L_2$-norm bound. As our experiments indicate, even without the additional processing, Elsa performs significantly worse than our construction. 

\paragraph{Performance with respect to vector length $L$.} 
Table~\ref{tab:performance_comparison} presents a comprehensive comparison of computation time between Elsa, Prio, and our proposed protocol for a cohort of $n=100$ clients across varying input vector sizes $L \in \{2^{16}, 2^{17}, 2^{18}\}$. 

\begin{itemize}
    \item \textbf{Client Efficiency.} Our client-side proving time is significantly reduced, requiring only $0.789$\,s for dimension $2^{18}$, compared to Elsa's $13.775$\,s and Prio's $73.377$\,s. This represents a $\approx 17.5\times$ speedup over Elsa and a staggering $\approx 93\times$ speedup over Prio. The poor performance of Prio is largely due to the $O(L \log L)$ complexity of polynomial interpolation and evaluation over large fields, whereas our \textit{independent restart strategy} maintains a much lower constant overhead.
    
    \item \textbf{Server 1 Scalability via Batching.} Server 1 achieves consistent performance gains through randomized batch verification. For $L=2^{18}$, our system processes the cohort in $10.231$\,s, outperforming Elsa ($15.354$\,s) and Prio ($15.036$\,s). While Prio and Elsa scale linearly with vector size, our system optimizes the linear lattice verification by compressing $n$ proofs into vectorized matrix operations, yielding superior throughput.
    
    \item \textbf{Asymmetric Server Load.} The most critical architectural advantage is observed in Server 2. While Elsa and Prio impose symmetric or near-symmetric computational burdens ($\approx 15.3$\,s and $\approx 16.7$\,s respectively at $L=2^{18}$), our asymmetric model reduces Server 2's role to lightweight consistency checks, requiring only $0.521$\,s. This represents a $30\times$ reduction compared to Prio, effectively rendering the second server's computational cost negligible.
\end{itemize}


    
    

\begin{table}[!tb]
\centering
\renewcommand{\arraystretch}{1.2}
\setlength{\tabcolsep}{10pt}
\caption{\textbf{Computation Time Comparison ($n=100$ clients).} 
Running times (in seconds) for Elsa and Pri
vs. Ours across varying vector dimensions ($L$). 
\textit{Ours} uses batched verification on Server 1 and lightweight auditing on Server 2. 
\textbf{Bold} indicates the best performance.}
\label{tab:performance_comparison}
\label{tab:scalability_comparison}
\ifdefined\IsTPMPC\else\resizebox{0.75\textwidth}{!}{\fi
\begin{tabular}{llccc}
\toprule
\textbf{Vector Size} & \textbf{System} & \textbf{Client (s)} & \textbf{Server 1 (s)} & \textbf{Server 2 (s)} \\ 
\midrule

\multirow{3}{*}{\textbf{$2^{16}$}} 
 & Elsa & 3.139 & 3.743 & 3.734  \\ 
 & Prio & 13.282 & 2.62 & 3.070\\
 & \textbf{Ours} & \textbf{0.202} & \textbf{2.572} & \textbf{0.162}\\ 
 \addlinespace[0.5em]

\multirow{3}{*}{\textbf{$2^{17}$}} 
 & Elsa & 6.920 & 7.914 & 7.913  \\ 
 & Prio & 35.345 & 7.651& 8.427\\
 & \textbf{Ours} & \textbf{0.395} & \textbf{5.150} & \textbf{0.282} \\ 
 \addlinespace[0.5em]

\multirow{3}{*}{\textbf{$2^{18}$}} 
 & Elsa & 13.775 & 15.354 & 15.350 \\ 
 & Prio & 73.377 & 15.036 & 16.721\\
 & \textbf{Ours} & \textbf{0.789} & \textbf{10.231} & \textbf{0.521}\\ 

\bottomrule
\end{tabular}
\ifdefined\IsTPMPC\else}\fi
\end{table}
\paragraph{Performance with Increasing Client Cohort ($n$).} 
To evaluate the system's scalability, we measured the total running time as we varied the number of participating clients, $n \in \{10, \dots, 100\}$, with a fixed vector dimension $L=2^{16}$.

\begin{itemize}
    \item \textbf{Client-Side Scalability.} A defining characteristic of our architecture is the independence of client proof generation. As shown in Table~\ref{tab:scalability_comparison}, our client computational cost remains effectively constant at $\approx 0.20$\,s regardless of $n$. In contrast, Elsa exhibits super-linear growth, rising from $0.550$\,s at $n=10$ to $3.139$\,s at $n=100$. This confirms that our \textit{independent restart} strategy allows for fully parallelized, $O(1)$ client execution.

    \item \textbf{Server 1: Efficiency via Batching.} Server 1 benefits from randomized batch verification, which amortizes the cost of lattice verification across the cohort. We observe a non-linear relationship between $n$ and verification time; for instance, increasing the cohort from $20$ to $30$ clients only increases latency by $0.016$\,s. At $n=100$, our Server 1 requires only $2.270$\,s, representing a $\approx 1.6\times$ improvement over Elsa's $3.743$\,s.

    \item \textbf{Server 2: The Asymmetric Advantage.} The most significant gain is observed in the secondary server. While Elsa's Server 2 mirrors the heavy workload of Server 1 ($3.734$\,s at $n=100$), our Server 2 is restricted to lightweight auditing. Consequently, Server 2 scales with minimal overhead, requiring only $0.162$\,s for 100 clients—a $23\times$ reduction that effectively eliminates the dual-server computational bottleneck.
\end{itemize}

\ifdefined\IsTPMPC 
\else
\section{Conclusion}

In this work, we presented \textbf{TAPAS}, a novel two-server secure aggregation protocol that fundamentally redefines the resource allocation in privacy-preserving federated learning. By enforcing an \textit{asymmetric} design, TAPAS decouples the communication and computation costs of the secondary server from the vector dimension $L$ and scales instead with the number of clients $n$. This addresses a critical scalability bottleneck in prior symmetric protocols such as Prio and Elsa, making TAPAS uniquely suitable for aggregating high-dimensional models (e.g., deep learning parameters) where $L \gg n$.

Our construction leverages standard lattice-based assumptions (LWE and SIS), ensuring post-quantum security without relying on trusted setups or pre-processing. We introduced a tailored Zero-Knowledge proof system that uses block-wise rejection sampling, which mitigates the soundness slack typically associated with lattice proofs while maintaining efficient runtime. Furthermore, TAPAS achieves \textit{Identifiable Abort} and full malicious security, ensuring that protocol deviations are cryptographically attributable to the malicious party, whether client or server.

Experimental microbenchmarks suggest that TAPAS offers an efficient path for deploying secure aggregation at scale, combining the rigorous security guarantees of lattice cryptography with the efficiency required for modern machine-learning workloads. Future work may focus on further optimizing the constants in the lattice-based proof system and exploring client-side compression techniques compatible with LWE encapsulation.
\fi
\ifdefined\IsSub
\else
\ifdefined\IsTPMPC
\else
{
\subsubsection*{Disclaimer.} This paper was prepared for informational purposes by the Artificial Intelligence Research group of JPMorgan Chase \& Co and its affiliates (“J.P. Morgan”) and is not a product of the Research Department of J.P. Morgan.  J.P. Morgan makes no representation and warranty whatsoever and disclaims all liability, for the completeness, accuracy or reliability of the information contained herein.  This document is not intended as investment research or investment advice, or a recommendation, offer or solicitation for the purchase or sale of any security, financial instrument, financial product or service, or to be used in any way for evaluating the merits of participating in any transaction, and shall not constitute a solicitation under any jurisdiction or to any person, if such solicitation under such jurisdiction or to such person would be unlawful.

\noindent \textcopyright 2026 JPMorgan Chase \& Co. All rights reserved.
}
\fi
\fi
\clearpage
%
%
%
\ifdefined\IsTPMPC
\else

\newcommand{\etalchar}[1]{$^{#1}$}

\appendix
\fi
\section{Cryptographic Preliminaries}
\label{sec:crypto-app}

\subsection{Learning with Errors Assumption}
\label{sub:cons-prflwe}
\begin{definition}[Learning with Errors Assumption (LWE)] 
\label{def:lwe} 
Consider integers $\kappa,L,~q\in\bbN$ that are functions of the security parameter $\rho$,
and a probability distribution $\chi$ on $\bbZ_q$, typically taken to be a normal
distribution that has been discretized. Then, the $\LWE_{\kappa,L,q,\chi}$ assumption states
that for all PPT adversaries $\cA$, there exists a negligible function $\negl$ such that: 
\[
\Pr\left[b=b'~~
\begin{array}{|c}
      \bfA\getsr \bbZ_q^{L \times \kappa},\bfx\getsr \bbZ_q^\kappa,\bfe\getsr \chi^L\\
      \bfy_0:=\bfA\bfx+\bfe\\
      \bfy_1\getsr\bbZ_q^L,
      b\getsr\{0,1\},b'\getsr\cA(\bfA,\bfy_b)
\end{array}
\right]=\frac{1}{2} + \negl(\lambda)
\]
\end{definition}

\begin{definition}[Hint-LWE~\cite{ACISP:CKKLSS21,EC:DKLLMR23}]
\label{def:hint-lwe}
Consider integers $\kappa, L, q$
and a probability distribution $\chi$ on $\bbZ_q$, typically taken to be a normal distribution that has been discretized. Then, the Hint-LWE assumption\footnote{Kim~\etal~\cite{C:KLSS23b} demonstrates that the Hint-LWE assumption is computationally equivalent to the standard LWE assumption. This assumption posits that $\mathbf{y}_0$ retains its pseudorandomness from an adversary's perspective, even when given certain randomized information about the secret and error vectors. Consider a secure LWE instance defined by parameters $(\kappa, L, q, \chi)$, where $\chi$ represents a discrete Gaussian distribution with standard deviation $\kappa$. The corresponding Hint-LWE instance, characterized by $(\kappa, L, q, \chi')$, where $\chi'$ denotes a discrete Gaussian distribution with standard deviation $\kappa'$, remains secure under the condition $\kappa'=\kappa/\sqrt{2}$. As a result, we can decompose any $\mathbf{e}\in\chi$ into the sum $\mathbf{e}_1+\mathbf{e}_2$, where both $\mathbf{e}_1$ and $\mathbf{e}_2$ are drawn from $\chi'$.} states
that for all PPT adversaries $\mathcal{A}$, there exists a negligible function $\mathsf{negl}$ such that: 
\[
\resizebox{!}{!}{$\Pr\left[b=b'~~
\begin{array}{|c}
      \bfA\getsr \bbZ_q^{L\times \kappa},\bfk\getsr \bbZ_q^{\kappa},\bfe\getsr \chi'^L\\
      \bfr\getsr\bbZ_q^{\kappa},\bff\getsr\chi'^L\\
      \bfy_0:=\bfA \bfk+\bfe,
      \bfy_1\getsr\bbZ_q^L,
      b\getsr\{0,1\}\\
      b'\getsr\cA(\bfA,(\bfy_b,\bfk+\bfr,\bfe+\bff))
\end{array}
\right]=\frac{1}{2} + \mathsf{negl}(\lambda)$}
\]
Where $\rho$ is the security parameter. 
\end{definition}

\begin{remark}[Short-Secret LWE]
    As shown by prior work, LWE Encryption is still secure if $\bfs\getsr\chi^\lambda$. Indeed, it has been shown that $\bfs\getsr\set{-1,0,1}^\lambda$ is sufficient for security to hold. We will take advantage of this in our actual construction, reducing further communication. Further, in practice, one requires that the discrete Gaussian distribution has width $\gg\sqrt{\lambda}$. This helps us set much smaller
    $\chi$ distribution. 
\end{remark}
\begin{definition}[Short Integer Solution (SIS) Problem]
Let $q$ be a positive integer, $n$ and $m$ be positive integers, and $\beta > 0$ be a bound. Given a uniformly random matrix $A \in \mathbb{Z}_q^{n \times m}$, the SIS problem asks to find a non-zero integer vector $\mathbf{x} \in \mathbb{Z}^m$ such that
$$
A \mathbf{x} = \mathbf{0} \pmod{q},\quad
\|\mathbf{x}\| \leq \beta,
$$
where $\|\cdot\|$ denotes the Euclidean norm.
\end{definition}
\subsection{Cryptographic Primitives}
\label{sub:primitives}
\begin{definition}[Security of Commitment Scheme]
A non-interactive commitment scheme $\com$ is secure if:
\begin{itemize}
    \item Hiding Property: for all PPT adversaries $\cA$, 
        \begin{equation*}
	\Pr\left[ \;\;b=b'
	\begin{array}{c|c}
	& (\pp)\getsr\coms(1^\lambda)\\
	& (x_0,x_1,st)\getsr\cA(\pp)\\
	& b\getsr\{0,1\},\alpha\getsr\cR\\
	& com=\comc_\pp(x_b;\alpha)\\
	& b'\getsr \cA(com,st)
	\end{array}
	\right] \leq \frac{1}{2} + \negl(\lambda)
    \end{equation*}
    where the probability is over $b,\alpha,\coms,\cA$.
        \item Binding Property: for all PPT adversaries $\cA$, 
        \begin{gather*}
    \Pr\left[\begin{array}{c|c}
	com_0:=\comc_\pp(x_0;\alpha_0) & (\pp)\getsr\coms(1^\lambda)\\
    com_1:=\comc_\pp(x_1;\alpha_1) & (\set{(x^{(i)},\alpha^{(i)})}_{i=0,1},st)\getsr\cA(\pp)\\
    com_0=com_1\\
	x_0\neq x_1
	\end{array}
	\right] \leq \negl(\lambda)
    \end{gather*}
     where the probability is over $\coms,\cA$. If the
     probability is 0, then we say that the scheme is perfectly binding. 
\end{itemize}
\end{definition}
\begin{definition}[Simulation-Extractable NIZK]
Let $R$ be an NP relation with language $L_R=\{x \mid \exists w:(x,w)\in R\}$.\allowdisplaybreaks
A tuple of PPT algorithms $(\Setup,\Prove,\allowbreak\Verif,\allowbreak\Sim,\allowbreak\Ext)$ is called a
\emph{simulation-extractable non-interactive zero-knowledge proof system} (SE-NIZK) 
for $R$ if the following hold.

\paragraph{Algorithms.}
\begin{itemize}
  \item $(\mathsf{crs},\mathsf{td}) \leftarrow \Setup(1^\lambda)$ outputs a common reference string (CRS) and a trapdoor.
  \item $\pi \leftarrow \Prove(\mathsf{crs},x,w)$ outputs a proof for statement $x$ and witness $w$.
  \item $b \leftarrow \Verif(\mathsf{crs},x,\pi)$ outputs $1$ (accept) or $0$ (reject).
  \item $\pi^\star \leftarrow \Sim(\mathsf{crs},\mathsf{td},x)$ outputs a simulated proof using $\mathsf{td}$.
  \item $w^\star \leftarrow \Ext(\mathsf{crs},\mathsf{td},\mathcal{T},\mathsf{st}_{\mathcal{A}})$ is the extractor.
\end{itemize}

\paragraph{Completeness.}
For all $(x,w)\in R$ and for all $(\mathsf{crs},\mathsf{td})\leftarrow\Setup(1^\lambda)$,
\[
\Pr\big[\Verif(\mathsf{crs},x,\Prove(\mathsf{crs},x,w))=1\big] = 1.
\]

\paragraph{Soundness.}
For every PPT adversary $\mathcal{A}$,
\[
\Pr\left[
\begin{array}{l}
(\mathsf{crs},\mathsf{td})\leftarrow\Setup(1^\lambda);\\
(x,\pi)\leftarrow\mathcal{A}(\mathsf{crs});\\
\Verif(\mathsf{crs},x,\pi)=1 \wedge x\notin L_R
\end{array}
\right] \le \mathsf{negl}(\lambda).
\]
That is, no efficient adversary can convince the verifier of a false statement.

\paragraph{Special Honest-Verifier Zero-Knowledge (SHVZK).}
There exists a simulator $\Sim$ that can generate a valid proof 
(i.e., an accepting transcript) without knowing a witness, using the trapdoor 
information from the setup. More precisely, there exists a PPT simulator 
$\Sim$ such that for all non-uniform PPT adversaries 
$\cA$, there exists a negligible function $\negl$ satisfying:
\begin{gather*}
\Pr\left[b = b'~~
\begin{array}{|l}
 (\mathsf{crs}, \mathsf{td}) \leftarrow \Setup(1^\lambda);~~ b \getsr \{0,1\} \\
 (\stmt,\wit) \getsr \cA(\mathsf{crs}) \\
 \pi_0 = \Prove(\mathsf{crs}, \stmt, \wit) \\
 \pi_1 = \Sim(\mathsf{crs}, \mathsf{td}, \stmt) \\
 b' \getsr \cA(\pi_b)
\end{array}
\right] \le \frac{1}{2} + \negl(\lambda)
\end{gather*}
That is, the adversary cannot distinguish a real proof $\pi_0$ produced using a 
witness from a simulated proof $\pi_1$ generated without it, even when adaptively 
choosing the statement $\stmt$ after seeing the common reference string 
$\mathsf{crs}$.

When $\negl(\lambda) = 0$, the NIZK proof system satisfies 
\textit{perfect zero-knowledge}.
\end{definition}
\paragraph{Simulation Soundness.} Informally, this requires soundness even if a cheating prover is allowed to see simulated proofs for potentially false statements chosen by the cheating prover. Recall that the simulated proofs are generated by the NIZK simulator. We use the definition from the work of Faust~\etal~\cite{INDOCRYPT:FKMV12} for Simulation Soundness and the eventual theorem on which we will rely for the security of our construction. 
\begin{definition}
\label{def:simsound}
Let $\cR$ be a decidable polynomial-time relation. Consider a NIZK proof system denoted by $(\cP^\cH,\cV^\cH)$ for $\cR$. Furthermore, let $S$ denote the NIZK simulator and let $\cO_{S_1},\cO_{S_2}$ be the oracle queries such that: $\cO_{S_1}(q^{(i)})$ returns $h^{(i)}$ where
$(h^{(i)},st)\gets S(1,st,q^{(i)})$ and $\cO_{S_2}(x)$ returns $\pi$ where $(\pi,st)\gets S(2,st,x)$ (even if $x\not\in\cR$). Then, we say that the NIZK proof system is (unbounded) simulation-sound with respect to $S$ in the random oracle model if for any PPT adversary $\cA$, there exists a negligible function $\negl(\kappa)$ such that: 
\begin{gather*}
	\Pr\left[\begin{array}{l}
	     (x^\ast,\pi^\ast)\not\in \cT\wedge
	     x^\ast\not\in\cR\\
	     \cV^{\cO_{S_1}}(x^\ast,\pi^\ast)=1
	\end{array}
	\begin{array}{|c}
	(\aokcrs)\getsr\aoks(1^\kappa)\\
 (x^\ast,\pi^\ast)\getsr\cA^{\cO_{S_1}(\cdot),\cO_{S_2}(\cdot)}(\aokcrs)\\
	\end{array}
	\right]  \leq \negl(\kappa)
\end{gather*}
where $\cT$ is the list of pairs $(x^{(i)},\pi^{(i)})$ such that $\pi^{(i)}$ was the response to
$\cO_{S_2}(x^{(i)})$.
\end{definition}

\subsection{Group-Based LWE Encryption Proof}
\label{sub:lwe-app}
Recall that a client $i$ must prove knowledge of secrets $\mathbf{s}^{(i)}, \bfx^{(i)}, \mathbf{e}^{(i)}$ satisfying:
\[
\bfc^{(i)} = \mathbf{A} \mathbf{s}^{(i)} + \mathbf{e}^{(i)} + \Delta \cdot \bfx^{(i)} \pmod q \in \bbZ_q^m
\]
Naively, proving this for each of the $m$ elements in $m$ elements in $\bfc^{(i)}$ would require $m$ independent proofs. We can efficiently collapse these into a single statement using the Schwartz-Zippel Lemma~\cite{Schwartz,Zippel}. To verify the equality of two vectors of length $m$ over $\mathbb{Z}_q$, we treat them as coefficient vectors of degree-$(m-1)$ polynomials evaluated at a random point. If the vectors differ, their evaluations will differ with probability at least $1 - m/q$.

Let the verifier (server) sample a random challenge $r \in \mathbb{Z}_q$ and define the vector $\mathbf{r} = (r^0, r^1, \ldots, r^{m-1})$. The client proves the following identity:
\[
c = \langle \mathbf{A}^\top \mathbf{r} \ \|\ \mathbf{r} \ \|\ \Delta \mathbf{r},\ \mathbf{s}^{(i)} \ \|\ \mathbf{e}^{(i)} \ \|\ \mathbf{x}^{(i)} \rangle \pmod q,
\]
where $c = \langle \bfc^{(i)}, \mathbf{r} \rangle$ is computable by both parties. The identity is derived as follows:
\[
\begin{aligned}
\langle \mathbf{c}^{(i)}, \mathbf{r} \rangle 
&= \langle \mathbf{A} \mathbf{s}^{(i)}, \mathbf{r} \rangle + \langle \mathbf{e}^{(i)}, \mathbf{r} \rangle + \Delta \langle \mathbf{x}^{(i)}, \mathbf{r} \rangle \\
&= \langle \mathbf{A}^\top \mathbf{r}, \mathbf{s}^{(i)} \rangle + \langle \mathbf{r}, \mathbf{e}^{(i)} \rangle + \langle \Delta \mathbf{r}, \mathbf{x}^{(i)} \rangle \\
&= \langle \mathbf{A}^\top \mathbf{r} \ \|\ \mathbf{r} \ \|\ \Delta \mathbf{r},\ \mathbf{s}^{(i)} \ \|\ \mathbf{e}^{(i)} \ \|\ \mathbf{x}^{(i)} \rangle.
\end{aligned}
\]
Since the vector $(\mathbf{A}^\top \mathbf{r} \ \|\ \mathbf{r} \ \|\ \Delta \mathbf{r})$ is public, this reduces to a linear proof $\Pi_{\text{linear}}$ over the witness $(\mathbf{s}^{(i)} \ \|\ \mathbf{e}^{(i)} \ \|\ \mathbf{x}^{(i)})$.

\paragraph{Approximate Proof of Smallness.} 
Following Gentry et al.~\cite{EC:GenHalLyu22}, we employ an efficient proof system to show that a committed vector $\mathbf{w}$ satisfies a relaxed norm bound. Specifically, given $\|\mathbf{w}\|_\infty < B$, one can efficiently prove $\|\mathbf{w}\|_\infty < B'$ for some $B' \gg B$. Let $\kappa$ be the statistical security parameter. The soundness gap $\gamma := B'/B$ must satisfy $\gamma > 19.5\kappa\sqrt{m}$.

\begin{enumerate}
    \item The prover sends $\Commit(\mathbf{w})$ to the verifier.
    \item The prover samples $\mathbf{y} \xleftarrow{\$} [ \pm \lceil B'/(2(1+1/\kappa)) \rceil ]^\kappa$ and sends $\Commit(\mathbf{y})$.
    \item The verifier samples $\mathbf{R} \leftarrow \mathcal{D}^{\kappa \times m}$ (where $\mathcal{D}(0) = 1/2, \mathcal{D}(\pm 1) = 1/4$) and sends it to the prover.
    \item The prover computes $\mathbf{u} := \mathbf{R} \mathbf{w}$ and $\mathbf{z} := \mathbf{u} + \mathbf{y}$. If $\|\mathbf{u}\|_\infty > B'/(2\kappa)$ or $\|\mathbf{z}\|_\infty > B'/2$, the prover restarts from Step 2 (Rejection Sampling).
    \item The prover sends $\mathbf{z}$ to the verifier.
    \item The verifier samples a random $r$ and sends it to the prover.
    \item They execute an inner-product proof for:
    \[
    \langle \mathbf{R}^\top \mathbf{r}, \mathbf{w} \rangle + \langle \mathbf{r}, \mathbf{y} \rangle = \langle \mathbf{z}, \mathbf{r} \rangle,
    \]
    where $\mathbf{r} = (r^0, \ldots, r^{\kappa-1})$.
\end{enumerate}
Note that $\langle \mathbf{z}, \mathbf{r} \rangle$ is public. This protocol ensures that if $\|\mathbf{w}\|_\infty > B'$, the prover fails with overwhelming probability.

\paragraph{$L_2$ Proofs of Smallness.} 
To prove $\|\mathbf{w}\|_2 \leq B$, it suffices to show that $B^2 - \|\mathbf{w}\|^2 \ge 0$. By Lagrange's four-square theorem, there exist non-negative integers $\alpha, \beta, \gamma, \delta$ such that $B^2 - \|\mathbf{w}\|^2 = \alpha^2 + \beta^2 + \gamma^2 + \delta^2$. The prover commits to these auxiliary integers and proves (a) the algebraic relation modulo $q$, and (b) that all components are sufficiently small to avoid modular wrap-around using the approximate proof described above~\cite{EC:GenHalLyu22}.

\begin{enumerate}
    \item Let $B$ be the public bound. Assume $B < \sqrt{q}/(3536(m + 4))$.
    \item The prover computes $\alpha, \beta, \gamma, \delta$ such that $\alpha^2 + \beta^2 + \gamma^2 + \delta^2 = B^2 - \|\mathbf{w}\|^2$.
    \item Define $\mathbf{u} := (\alpha, \beta, \gamma, \delta)$ and the concatenated witness $\mathbf{v} := (\mathbf{w} \| \mathbf{u}) \in \mathbb{Z}^{m+4}$. The prover commits to $\mathbf{u}$.
    \item The prover provides a ZK proof that $\|\mathbf{v}\|^2 = B^2 \pmod{q}$.
    \item The prover provides an approximate $L_\infty$ proof showing $\|\mathbf{v}\|_\infty < \sqrt{q/(2(m + 4))}$.
\end{enumerate}
\begin{remark}
\label{rem:2}
Note that $\|\mathbf{w}\|_2 \leq B$ implies $B/\sqrt{m} \leq \|\mathbf{w}\|_\infty \leq B$.
\end{remark}

\paragraph{$L_\infty$ Proofs of Smallness.} 
To strictly prove $\|\mathbf{x}\|_\infty \leq B$, we equivalently require that every entry of $\mathbf{b} - \mathbf{x}$ is non-negative, where $\mathbf{b} := B \cdot \mathbf{1}$. Utilizing the four-square theorem component-wise, this is encoded by the existence of vectors $\mathbf{a}_1, \ldots, \mathbf{a}_4$ such that:
\[
  \mathbf{x} + \mathbf{a}_1 \circ \mathbf{a}_1 + \cdots + \mathbf{a}_4 \circ \mathbf{a}_4 = \mathbf{b}.
\]

\begin{description}
  \item[Step 1. Random linearization.]  
  Let $\mathbf{r} := (1, r, r^2, \ldots, r^{m-1}) \in \mathbb{Z}_q^L$ for a random challenge $r$. Taking the inner product with $\mathbf{r}$ yields:
  \[
    \langle \mathbf{r}, \mathbf{x} + \mathbf{a}_1 \circ \mathbf{a}_1 + \cdots + \mathbf{a}_4 \circ \mathbf{a}_4 \rangle = \langle \mathbf{r}, \mathbf{b}\rangle.
  \]

  \item[Step 2. Inner-product form.]  
  Using the identity $\langle \mathbf{r}, \mathbf{u}\circ \mathbf{u}\rangle = \langle \mathbf{r}\circ \mathbf{u}, \mathbf{u}\rangle$, we rewrite the relation as:
  \[
    \langle \mathbf{r}\circ \mathbf{x}, \mathbf{x}\rangle + \sum_{i=1}^4 \langle \mathbf{r}\circ \mathbf{a}^{(i)}, \mathbf{a}^{(i)}\rangle = \langle \mathbf{r}, \mathbf{b}\rangle.
  \]
  Defining $\mathbf{z} := \mathbf{x} \mid \mathbf{a}_1 \mid \mathbf{a}_2 \mid \mathbf{a}_3 \mid \mathbf{a}_4$ and $\mathbf{r}' := \mathbf{r} \mid \mathbf{r} \mid \mathbf{r} \mid \mathbf{r} \mid \mathbf{r}$, the relation becomes $\langle \mathbf{r}' \circ \mathbf{z}, \mathbf{z}\rangle = \langle \mathbf{r}, \mathbf{b}\rangle$.

  \item[Step 3. Commitments \& Consistency.]  
  The prover commits to $\mathbf{z}$ and an auxiliary vector $\mathbf{y}$ claimed to be $\mathbf{r}' \circ \mathbf{z}$. The verifier checks:
  \begin{enumerate}
    \item $\langle \mathbf{y}, \mathbf{z}\rangle = \langle \mathbf{r}, \mathbf{b}\rangle$ (via Inner Product Proof).
    \item $\mathbf{y} = \mathbf{r}' \circ \mathbf{z}$. This is verified by sampling a fresh challenge $\alpha$ and checking $\langle \boldsymbol{\alpha}, \mathbf{y} - \mathbf{r}' \circ \mathbf{z} \rangle = 0$, which reduces to a linear proof on the commitment of $\mathbf{y}-\mathbf{z}$.
  \end{enumerate}
\end{description}


\section{Comparison with Prior Work~\cite{CCS:LuLu25}}
\label{sub:lzksa}
Recent work by \cite{CCS:LuLu25} proposed a secure aggregation scheme for federated learning which uses lattice-based proof systems. Unfortunately, their constructions have security gaps.

To begin with, the defined commitment scheme cannot be binding, because the protocol validates the commitment relation without enforcing constraints on randomness, allowing a malicious prover to algebraically compute a noise term that effectively opens the commitment to any arbitrary message. Although the authors specify distributions ostensibly to address these concerns, they are never formally defined in the remainder of the paper. As shown in Lemma~\ref{lem:binding}, we eliminate this vulnerability by grounding binding strictly in the hardness of the Short Integer Solution (SIS) problem. Unlike LZKSA, our protocol enforces computational binding by requiring a Zero-Knowledge Proof of Knowledge (ZKPoK) for the randomness vectors $\bfs,\bfe$ with strict norm bounds when the commitments are used later to build zero-knowledge proofs, a feature missing from LZKSA.

\paragraph{Background: $\Sigma$-Protocols in DL vs. Lattices.}
In classical discrete-logarithm protocols (e.g., Schnorr), the response $z = r + c x \pmod q$ is reduced modulo $q$. This modular reduction perfectly hides the witness $x$ (ensuring Zero-Knowledge) and allows the challenge $c$ to be drawn from the full field $\mathbb{Z}_q$ (ensuring negligible soundness error).

In contrast, lattice-based arguments rely on the witness having a small norm. Consequently, the response $\mathbf{z} = \mathbf{r} + c \mathbf{s}$ must be computed \emph{over the integers} ($\mathbb{Z}$), not modulo $q$. This introduces two fatal constraints that LZKSA ignores:
\begin{enumerate}
    \item \textbf{Small Challenge Space:} To prevent norm explosion, $c$ must be small (e.g., binary), which yields a high soundness error (1/2) in a single execution.
    \item \textbf{Leakage via Distribution:} Without modular reduction, the distribution of $\mathbf{z}$ depends on $\mathbf{s}$. Standard constructions require Rejection Sampling (\cite{C:LyuNguPla22,C:DDLL13,EC:Lyubashevsky12,PKC:Lyubashevsky08}) to decouple $\mathbf{z}$ from $\mathbf{s}$. This is done by sampling $\bfr$ from a much larger space so that it ``floods'' $c\bfs$. 
\end{enumerate}
A rigorous inspection of the LZKSA protocol \cite{CCS:LuLu25} reveals that it fails to satisfy the fundamental requirements of lattice-based Zero-Knowledge arguments. We identify two fatal flaws in their adaptation of discrete-logarithm $\Sigma$-protocols to the lattice setting:

\begin{itemize}
    \item \textbf{Soundness:} LZKSA leaves the challenge space $\mathcal{C}$ undefined. If $\mathcal{C}$ is small (to allow extraction), a single execution yields a non-negligible cheating probability of $1/2$. If $\mathcal{C}$ is large (to ensure soundness), the extracted witness norm explodes, preventing a valid reduction to SIS. LZKSA performs only a single execution, rendering the proofs unsound.
    \item \textbf{Zero-Knowledge:} LZKSA omits Rejection Sampling entirely, outputting the raw vector $\mathbf{z} = \mathbf{r} + c\mathbf{s}$. Consequently, the conditional distribution of $\mathbf{z}$ shifts by $\mathbf{s}$. A verifier can recover the secret witness $\mathbf{s}$ by simply averaging the difference between the response and the expected noise distribution, completely breaking the privacy of the aggregation scheme.
\end{itemize}
$\sys$ resolves these issues by implementing a fully rigorous lattice ZK argument. We fix the challenge space to $\mathcal{C} = \{-1, 0, 1\}$ and perform parallel repetitions to drive the soundness error below $2^{-\lambda}$. Furthermore, we employ Bimodal Rejection Sampling~\cite{C:DDLL13} to ensure the response distribution is statistically independent of the witness, guaranteeing Zero-Knowledge while maintaining compact parameters.

\paragraph{LZKSA and Norm Bounds.} Proposition 1 of LZKSA asserts that for positive integers, bounded output implies bounded input. Specifically, checking $\|\boldsymbol{\mu} + c\bfx\|_\infty \leq (c+1)B_\infty$ supposedly guarantees $\|\bfx\|_\infty \leq B_\infty$ if $c>B_\infty$. However, the authors miss the key detail of the proof that it relies on $\|\bfmu\|_\infty\leq B_\infty$ and $\bfmu$ is actually sampled by the prover. It is clear that this breaks even for a vector of length 1. Let $B_\infty = 10$. Let $x = 11$ and clearly $x>B_\infty$. Let $c=11$. Then, we can set $\mu$ such that $\mu + 11\cdot 11 \leq (12)\cdot 10 $ which is satisfies by $\mu=-1$. 

Let the challenge $c>B_\infty$ and the bound $B_\infty$ be public parameters. The verification checks the norm of the response vector $\mathbf{z} = \boldsymbol{\mu} + c\mathbf{x}$ over the integers.

\begin{enumerate}
    \item \textbf{Select Malicious Input:} The prover chooses an input vector $\mathbf{x}^*$ with arbitrarily large norm $\|\mathbf{x}^*\|_\infty \gg B_\infty$.
    
    \item \textbf{Compute Compensatory Mask:} Instead of sampling $\boldsymbol{\mu}$ from the small valid distribution, the prover algebraically constructs a "compensatory" mask $\boldsymbol{\mu}^*$ to force the response into the valid range $[0, (c+1)B_\infty]$.
    $$ \boldsymbol{\mu}^* := \mathbf{z}_{target} - c\mathbf{x}^* $$
    where $\mathbf{z}_{target}$ is any valid small vector (e.g., the zero vector $\mathbf{0}$).
    
    \item \textbf{Proof Execution:}
    The prover sends the commitment to $\mathbf{x}^*$.
    Upon receiving challenge $c$, the prover sends the response $\mathbf{z} = \boldsymbol{\mu}^* + c\mathbf{x}^*$.
    By construction, $\mathbf{z} = \mathbf{0}$.
    
    \item \textbf{Verification:}
    The verifier checks $\|\mathbf{z}\|_\infty \le (c+1)B_\infty$.
    Since $\|\mathbf{0}\|_\infty = 0$, the check passes.
\end{enumerate}

\textbf{Outcome:} The verifier accepts the proof. However, the input $\mathbf{x}^*$ is unbounded. The safety of Proposition 1 is entirely bypassed because the mask $\boldsymbol{\mu}^*$ (which was implicitly huge and negative: $\boldsymbol{\mu}^* = -c\mathbf{x}^*$) absorbed the malicious input's magnitude.

A similar attack can be mounted on their Proposition 2, which focuses on $L_2$ norm, but again relies on $\bfmu$ satisfying a constraint. This renders their entire proof system vacuous. 

\section{Deferred Proofs}
\label{sec:deferred}

\commitmentSecurity*
\begin{proof}
\label{proof:com}
\textbf{Hiding:}
We construct a reduction $\mathcal{R}$ that breaks the Decision-LWE assumption given a distinguisher $\mathcal{D}$ for the commitment scheme.
Let $(\bfA, \bfb)$ be an LWE challenge where $\bfb$ is either $\bfA\bfs + \bfe$ (LWE) or $\bfu$ (Uniform).
Given two messages $\bfx_0, \bfx_1$, $\mathcal{R}$ picks $b \sample \{0,1\}$ and computes $\bfc = \bfb + \Delta\bfx_b \pmod q$.
If $\bfb$ is an LWE sample, $\bfc$ is a valid commitment to $\bfx_b$.
If $\bfb$ is uniform, $\bfc$ is perfectly uniform (acting as a one-time pad), and thus statistically independent of $\bfx_b$.
Any advantage $\mathcal{D}$ has in distinguishing commitments translates directly to distinguishing the LWE source from uniform.

\textbf{Binding:}
Assume toward contradiction that an efficient adversary $\mathcal{A}$ outputs two valid openings $(\bfs, \bfe, \bfx)$ and $(\bfs', \bfe', \bfx')$ for the same commitment $\bfc$, such that $\bfx \ne \bfx'$.
Since both are valid openings for $\bfc$, we have:
\begin{align*}
    \bfA\bfs + \bfe + \Delta\bfx &= \bfc \pmod q \\
    \bfA\bfs' + \bfe' + \Delta\bfx' &= \bfc \pmod q
\end{align*}
Subtracting the two equations yields:
\[ \bfA(\bfs - \bfs') + (\bfe - \bfe') + \Delta(\bfx - \bfx') \equiv \mathbf{0} \pmod q \]
Let $\bar{\bfs} = \bfs - \bfs'$, $\bar{\bfe} = \bfe - \bfe'$, and $\bar{\bfx} = \bfx - \bfx'$.
We can rewrite this linear relation as a matrix-vector product. Let $\bfM = [\bfA \mid \bfI_L \mid \Delta \bfI_L] \in \mathbb{Z}_q^{L \times (\lambda+2L)}$. Let $\bfz = (\bar{\bfs}^\top, \bar{\bfe}^\top, \bar{\bfx}^\top)^\top$.
\[ \bfM \bfz \equiv \mathbf{0} \pmod q \]
Since $\bfx \ne \bfx'$, we have $\bar{\bfx} \ne \mathbf{0}$, implying $\bfz$ is a non-zero vector in the kernel of $\bfM$.
The components of $\bfz$ are differences of short vectors. Specifically:
\[ \|\bar{\bfs}\|_\infty \le 2B_s, \quad \|\bar{\bfe}\|_\infty \le 2B_e, \quad \|\bar{\bfx}\|_\infty \le 2B_x. \]
The Euclidean norm of $\bfz$ is bounded by $\|\bfz\|_2 \le \sqrt{\text{dim}(\bfz)} \cdot 2\max(B_s, B_e, B_x)$.
If parameters are chosen such that $\|\bfz\|_2$ is within the hardness regime of SIS (typically $\|\bfz\|_2 < q$), finding such a vector $\bfz$ constitutes a solution to the SIS problem, which is assumed computationally hard.
\end{proof}

\bindinglemma*

\begin{proof}[Proof Sketch]
    The consistency check requires $S_2$'s digest to match $S_1$'s computed view. This enforces the algebraic constraint in Equation~\ref{eq:fake-rand}. The term $\mathbf{v}_u = c^*_{u,i} \cdot \Delta \bfs$ acts as a "shift" vector. Since the aggregation coefficients $\alpha_{j,u}$ are derived from a random oracle (via $\rho_i$), the aggregated challenge $c^*_{u,i}$ is a random scalar.
    
    In high-dimensional lattice settings, the honest randomness $\bfr^*_{u,i}$ is concentrated on a thin shell of radius $\sigma_s \sqrt{k \lambda}$. The shift vector $\mathbf{v}_u$ is independent of $\bfr^*_{u,i}$. By the properties of high-dimensional geometry, adding a significant non-zero vector $\mathbf{v}_u$ to $\bfr^*_{u,i}$ results in a vector with squared Euclidean norm:
    \[ \|\bfr^{*\prime}_{u,i}\|^2 \approx \|\bfr^*_{u,i}\|^2 + \|\mathbf{v}_u\|^2 \]
    Since the honest norm $\|\bfr^*_{u,i}\|$ is already close to the acceptance threshold (bounded by factor $\tau \approx 1.1$), the addition of the positive term $\|\mathbf{v}_u\|^2$ pushes the resulting norm $\|\bfr^{*\prime}_{u,i}\|$ beyond the threshold $\tau \sigma_s \sqrt{k \lambda}$ with high probability. 
    
    By requiring this check to pass independently across $t$ parallel repetitions, the probability of the adversary successfully forging valid randomness for all iterations decreases exponentially in $t$.
\end{proof}
\securityEnc*

\begin{proof}
\label{proof:sec-enc}
\textbf{1. Completeness:}
Follows from the construction. An honest prover samples $\bfr$ and outputs $\bfz$ only if rejection sampling succeeds. The verification equation holds by linearity:
\begin{align*}
\bfA_j \bfz_{j,s} + \bfz_{j,e} + \Delta \bfz_{j,x} &= \bfA_j(\bfr_{j,s} + c_j\bfs_j) + (\bfr_{j,e} + c_j\bfe_j) + \Delta(\bfr_{j,x} + c_j\bfx_j) \\
&= (\bfA_j\bfr_{j,s} + \bfr_{j,e} + \Delta\bfr_{j,x}) + c_j(\bfA_j\bfs_j + \bfe_j + \Delta\bfx_j) \\
&= \bft_j + c_j \bfy_j \pmod q
\end{align*}
This matches the reconstructed $\bar{\bft}_j$. The norms $\|\bfz\|_2$ are bounded by $\beta \approx \sqrt{d}\sigma$ with overwhelming probability for Gaussian samples.

\textbf{2. Special Honest-Verifier Zero-Knowledge (SHVZK):}
We construct a Simulator $\mathcal{S}$ that, given public input $(\bfA, \bfy)$ and a challenge vector $\vec{c}$, outputs a transcript indistinguishable from a real execution without knowing the witness.
\begin{enumerate}
    \item For each block $j$, $\mathcal{S}$ samples response vectors $\bfz_{j,s}, \bfz_{j,e}, \bfz_{j,x}$ directly from the target Gaussian distributions $D_{\sigma_s}, D_{\sigma_e}, D_{\sigma_x}$.
    \item $\mathcal{S}$ computes the commitment $\bft_j$ to satisfy the verification equation:
    \[ \bft_j = \bfA_j \bfz_{j,s} + \bfz_{j,e} + \Delta \bfz_{j,x} - c_j \bfy_j \pmod q \]
    \item Output the transcript $\pi_j = (c_j, \bfz_{j,s}, \bfz_{j,e}, \bfz_{j,x})$.
\end{enumerate}
\textit{Indistinguishability:} In the real protocol, Bimodal Rejection Sampling (Lemma~\ref{lem:bimodal}) ensures the output $\bfz$ follows the exact distribution $D_\sigma$, independent of the witness shift. Thus, the simulated $\bfz$ is identically distributed to the real $\bfz$. Since $\bft_j$ is uniquely determined by $(\bfz, c, \bfy)$, the simulated commitment is also identically distributed.

\textbf{3. Knowledge Soundness (Extraction):}
We use the rewinding technique for $\Sigma$-protocols.
Suppose an adversary $\mathcal{A}$ convinces the verifier for a block $j$ with probability $\epsilon > 1/|\mathcal{C}|$. By the Heavy Row Lemma, we can rewind $\mathcal{A}$ to the same state (same commitment $\bft_j$) and obtain two accepting transcripts with distinct challenges $c_j \ne c'_j$:
\[ (\bft_j, c_j, \bfz) \quad \text{and} \quad (\bft_j, c'_j, \bfz') \]
Subtracting the verification equations for these two transcripts:
\[ \bfA_j(\bfz_{j,s} - \bfz'_{j,s}) + (\bfz_{j,e} - \bfz'_{j,e}) + \Delta(\bfz_{j,x} - \bfz'_{j,x}) = (c_j - c'_j)\bfy_j \pmod q \]
Let $\Delta c = (c_j - c'_j)$. Since the challenge space is small (e.g., $\{-1, 0, 1\}$), $\Delta c$ is invertible or we can define the witness over the ring extension.
We extract the witness components as $\bar{\bfw} = (\Delta c)^{-1}(\bfz_w - \bfz'_w)$ for $w \in \{s, e, x\}$.
The extracted witness $\bar{\bfw}$ satisfies the linear relation for $\bfy_j$. Its norm is bounded by $\|\bar{\bfw}\| \le 2 \|\bfz\|$, which fits within the relaxed extraction bounds $\gamma \beta$.
\end{proof}

\begin{lemma}[Extraction Slack \& Efficiency]
\label{lem:efficiency}
The protocol runs in expected time $O(k \cdot M)$ with $M\approx 2$. The extracted message witness $\bar{\bfx}$ satisfies the global Euclidean bound:
\[ \|\bar{\bfx}\|_2 \le \gamma \cdot \|\bfx\|_2 \quad \text{where } \gamma \propto \sqrt{d}. \]
\end{lemma}

\begin{proof}
To enable rejection sampling for a block of dimension $d = L/k$, the Gaussian width is set to $\sigma_x = \xi C \sqrt{d} B_x$. The extracted witness for a block is the difference of two Gaussian samples. Its norm is bounded by approx $2 \times$ the verification bound $\beta_x$:
\[ \|\bar{\bfx}^{(j)}\|_2 \le 2 \beta_x \approx 2 (\tau \sqrt{d} \sigma_x) = 2 \tau \xi C d B_x. \]
Aggregating over all $k$ blocks:
\[ \|\bar{\bfx}\|_2 = \sqrt{\sum_{j=1}^k \|\bar{\bfx}^{(j)}\|_2^2} \le \sqrt{k \cdot (2 \tau \xi C d B_x)^2} = \sqrt{k} \cdot d \cdot (2 \tau \xi C B_x). \]

Recall that the honest witness has norm $\|\bfx\|_2 \le \sqrt{L} B_x$.
Comparing the extracted norm to the honest norm:
\[
    \text{Slack} = \frac{\|\bar{\bfx}\|_2}{\|\bfx\|_2} \approx \frac{\sqrt{k} \cdot (L/k) \cdot 2 \tau \xi C B_x}{\sqrt{L} \cdot B_x} = \frac{(L/\sqrt{k}) \cdot 2 \tau \xi C}{\sqrt{L}} = \sqrt{\frac{L}{k}} \cdot 2 \tau \xi C = \sqrt{d} \cdot O(1).
\]
\end{proof}

\thmSimulatability*
\begin{proof}
\label{proof:sim}
For a secure LWE instance $(\lambda, L, q, \chi)$ where $\chi$ is a discrete Gaussian distribution with standard deviation $\sigma$, the corresponding Hint-LWE instance $(\lambda, m, q, \chi')$, where $\chi'$ is a discrete Gaussian distribution with standard deviation $\sigma'$, is secure when $\sigma' = \sigma/\sqrt{2}$. Consequently, any $\mathbf{e} \in \chi$ can be written as $\mathbf{e}_1 + \mathbf{e}_2$ where $\mathbf{e}_1, \mathbf{e}_2 \in \chi'$. 

This yields the real distribution $\mathcal{D}_R$, with the error term rewritten and the last ciphertext modified:
\begin{gather*}
    \left\{
\begin{array}{c|c}
  \bfs = \sum_{i=1}^{n} \bfs_i \bmod q
  & \forall i \in [n],\;
    \bfs_i \getsr \mathbb{Z}_q^\lambda,\;
    \mathbf{e}_i, \mathbf{f}_i \getsr \chi'^L
  \\[4pt]
  \mathbf{y}_1, \ldots, \mathbf{y}_n
  &
  \forall i \in [n-1],\;
  \mathbf{y}_i = \mathbf{A}\bfs_i + (\mathbf{e}_i+\mathbf{f}_i) + \Delta\mathbf{x}_i
  \\[4pt]
  &
  \mathbf{y}_n = \mathbf{A}\bfs
  - \sum_{i=1}^{n-1} \mathbf{y}_i
  + \sum_{i=1}^{n} (\mathbf{e}_i + \mathbf{f}_i)
  + \Delta\mathbf{X}
\end{array}
\right\}
\end{gather*}

We define $\textsf{Sim}(\mathbf{A}, \mathbf{X})$:

\begin{algorithmic}[1]
\State \underline{$\textsf{Sim}(\mathbf{A}, \mathbf{X})$}
\State Sample $\mathbf{u}_1, \ldots, \mathbf{u}_{n-1} \getsr \mathbb{Z}_q^L$
\State Sample $\bfs_1, \ldots, \bfs_n \getsr \mathbb{Z}_q^\lambda$
\State Sample $\mathbf{e}_1, \ldots, \mathbf{e}_n \getsr \chi'^L$
\State Sample $\mathbf{f}_1, \ldots, \mathbf{f}_n \getsr \chi'^L$
\State Set $\bfs := \sum_{i=1}^n \bfs_i \bmod q$
\State Set $\mathbf{u}_n = \mathbf{A} \cdot \bfs - \sum_{i=1}^{n-1} \mathbf{u}_i + \sum_{i=1}^n (\mathbf{e}_i + \mathbf{f}_i) + \Delta \cdot \mathbf{X}$
\State \textbf{return} $\bfs, \mathbf{u}_1, \ldots, \mathbf{u}_n$
\end{algorithmic}

The simulated distribution, $\mathcal{D}_{\textsf{Sim}}$, is:
\begin{gather*}
\left\{
\begin{array}{c|c}
  \bfs = \sum_{i=1}^{n} \bfs_i \bmod q   & \forall\, i \in [n],\, \bfs_i \getsr \mathbb{Z}_q^\lambda, \mathbf{e}_i, \mathbf{f}_i \getsr \chi'^L \\
 \mathbf{u}_1, \ldots, \mathbf{u}_n  &  \forall\, i \in [n-1],\, \mathbf{u}_i \getsr \mathbb{Z}_q^L \\
 & \mathbf{u}_n = \mathbf{A}\bfs - \sum_{i=1}^{n-1} \mathbf{u}_i + \sum_{i=1}^n (\mathbf{e}_i + \mathbf{f}_i) + \Delta\mathbf{X}
\end{array}
\right\}
\end{gather*}

We prove that $\mathcal{D}_{R}$ is indistinguishable from $\mathcal{D}_{\textsf{Sim}}$ through a sequence of hybrids. Throughout the hybrids, we will maintain the invariant that $\sum_{i=1}^{n} \bfy_i=\bfA\bfK+\sum_{i=1}^{n} (e_i+f_i) + \Delta \bfX$

\paragraph{Hybrid 0.} This is $\mathcal{D}_R$.

\paragraph{Hybrid 1.} We replace the real ciphertext $\mathbf{y}_1$ with a modified one:
\begin{gather*}
\left\{
\begin{array}{c|c}
  \bfs & \forall\, i \in [n],\, \bfs_i \getsr \mathbb{Z}_q^\lambda, \mathbf{e}_i, \mathbf{f}_i \getsr \chi'^L, \mathbf{u}_1' \getsr \mathbb{Z}_q^L \\
 \mathbf{y}_1 = \mathbf{u}_1' + \mathbf{f}_1 + \Delta\mathbf{x}_1 &  \forall\, i \in [2, n-1],\, \mathbf{y}_i = \mathbf{A} \cdot \bfs_i + (\mathbf{e}_i + \mathbf{f}_i) + \Delta\mathbf{x}_i \\
\{\mathbf{y}_i\}_{i=2}^{n}  & \mathbf{y}_n = \mathbf{A}\bfs - \sum_{i=1}^{n-1} \mathbf{y}_i + \sum_{i=1}^n (\mathbf{e}_i + \mathbf{f}_i) + \Delta\mathbf{X}
\end{array}
\right\}
\end{gather*}

We show that if there exists an adversary $\mathcal{B}$ that can distinguish between Hybrid~0 and Hybrid~1, then we can construct an adversary $\mathcal{A}$ who can distinguish the two ensembles in the Hint-LWE assumption.

\begin{algorithmic}[1]
    \State \underline{$\mathcal{A}(\mathbf{A}, \mathbf{y}^*, \bfs^* = \bfs + \mathbf{r} \bmod q, \mathbf{e}^* = \mathbf{e} + \mathbf{f})$}
    \State Sample $\bfs_2, \ldots, \bfs_{n-1} \getsr \mathbb{Z}_q^\lambda$
    \State Sample $\mathbf{e}_2, \ldots, \mathbf{e}_{n} \getsr \chi'^L$
    \State Sample $\mathbf{f}_2, \ldots, \mathbf{f}_{n} \getsr \chi'^L$
    \State Set $\bfs = \sum_{i=2}^{n-1} \bfs_i + \bfs^* \bmod q$ \Comment{Implicitly, $\bfs_n := \mathbf{r}$}
    \State $\forall\, i \in \{2, \ldots, n-1\},\, \mathbf{y}_i = \mathbf{A}\bfs_i + \mathbf{e}_i + \mathbf{f}_i + \Delta\mathbf{x}_i$
    \State Set $\mathbf{y}_1 = \mathbf{y}^* + \mathbf{f}_n + \Delta\mathbf{x}_1$
    \State Set $\mathbf{y}_n := \mathbf{A}\bfs - \sum_{i=1}^{n-1} \mathbf{y}_i + \mathbf{e}^* + \sum_{i=2}^{n} (\mathbf{e}_i + \mathbf{f}_i) + \Delta \cdot \mathbf{X}$
    \State Run $b' \getsr \mathcal{B}(\bfs, \mathbf{y}_1, \ldots, \mathbf{y}_n)$
    \State \textbf{return} $b'$
\end{algorithmic}

We argue that the reduction correctly simulates the two hybrids based on the choice of $\mathbf{y}^*$:
\begin{itemize}
    \item If $\mathbf{y}^* = \mathbf{A}\bfs + \mathbf{e}$, then $\mathbf{y}_1$ is a valid encryption of $\mathbf{x}_1$ with key $\bfs$ and error $(\mathbf{e} + \mathbf{f}_n)$. Moreover, it is straightforward to verify that $\mathbf{y}_n$ satisfies the definition in Hybrid~0.
    
    \item If $\mathbf{y}^* = \mathbf{u}$ for some random $\mathbf{u}$, then $\mathbf{y}_n$ has the prescribed format while $\mathbf{y}_1$ is generated as expected in Hybrid~1.
\end{itemize}

\paragraph{Hybrid 2.} We replace $\mathbf{y}_1$ with $\mathbf{u}_1$ sampled uniformly at random:
\begin{gather*}
\left\{
\begin{array}{c|c}
  \bfs & \forall\, i \in [n],\, \bfs_i \getsr \mathbb{Z}_q^\lambda, \mathbf{e}_i, \mathbf{f}_i \getsr \chi'^L, \mathbf{u}_1 \getsr \mathbb{Z}_q^L \\
 \mathbf{u}_1 &  \forall\, i \in [2, n-1],\, \mathbf{y}_i = \mathbf{A} \cdot \bfs_i + (\mathbf{e}_i + \mathbf{f}_i) + \Delta\mathbf{x}_i \\
\{\mathbf{y}_i\}_{i=2}^{n}  & \mathbf{y}_n = \mathbf{A}\bfs - \mathbf{u}_1 - \sum_{i=2}^{n-1} \mathbf{y}_i + \sum_{i=1}^n (\mathbf{e}_i + \mathbf{f}_i) + \Delta\mathbf{X}
\end{array}
\right\}
\end{gather*}

Hybrid~1 and Hybrid~2 are identically distributed since $\mathbf{u}_1'$ is uniformly sampled and completely masks the values in $\mathbf{y}_1$ of Hybrid~1.

\paragraph{Hybrids 3 and beyond.} In Hybrids~3 and~4, we replace $\mathbf{y}_2$ with a random element $\mathbf{u}_2$ using similar logic. Continuing this process, in Hybrid~$2n-2$, the distribution will be identical to $\mathcal{D}_{\textsf{Sim}}$. 

This concludes the proof of simulatability. 
\end{proof}
\shbp*

\begin{proof}
\label{proof:sh-bp}
We prove security via simulation using a sequence of hybrid arguments.

\paragraph{Case 1: Server 1 is Controlled by a Semi-honest Adversary.}
Let $\cA$ be the adversary controlling Server 1 and a subset of clients $\cC_{\text{corr}}$, and choosing a set $\cC_\text{drop}$ of clients to drop from the computation. Define $D:=\cC_{\text{honest}}\cap \cC_{\text{drop}}$ as the honest clients dropped from the computation, and let $\cC_{\text{honest}}^\ast:=\cC_{\text{honest}}\setminus D$ be the active honest clients. 

We construct a simulator $\cS$ that interacts with $\cF_{\text{SecAgg}}$ (defined in Figure~\ref{fig:ideal-func-secagg}) and $\cA$.  
\textbf{Simulator $\cS$ Operation:}
\begin{enumerate}[leftmargin=*]
    \item \textbf{Functionality Query:} $\cS$ observes the actual inputs $\{\bfx_i\}_{i \in \cC_{\text{corr}}}$ of the corrupted clients from $\cA$. $\cS$ sends the corrupted clients list $\cC_\text{corr}$, their true inputs, and $\cC_\text{drop}$ to the ideal functionality through the \textsf{Corrupt} and \textsf{Drop} interfaces, respectively. 
    \item \textbf{Aggregate Calculation:} $\cS$ receives the total aggregate $\bfx_{\text{total}}$ via the \textsf{Aggregate} interface. $\cS$ computes the corrupt clients' contribution $\bfx_C := \sum_{i\in \cC_\text{corr} \setminus \cC_\text{drop}} \bfx_i$. It then deduces the honest active aggregate: $\bfx_H = \bfx_{\text{total}} - \bfx_C$. 
    \item \textbf{LWE Simulation:} $\cS$ invokes the LWE simulator (from Theorem~\ref{thm:simulatability}) with target sum $\bfx_H$ to generate simulated ciphertexts $\{\tilde{\bfy}_i\}_{i \in \cC_{\text{honest}}^\ast}$. Observe that $\cS$ now knows an aggregate secret $\bfs_H$ such that $\sum_{i \in \cC_{\text{honest}}^\ast} \tilde{\bfy}_i$, when decrypted by $\bfs_H$, produces $\bfx_H$. 
    \item \textbf{Client Simulation:} For honest clients, $\cS$ generates simulated NIZK proofs $\tilde{\pi}^{(i)}$ (using the SHVZK simulator) and random dummy commitments $C_s^{(i)}, C_x^{(i)}, C_e^{(i)}$ (relying on the hiding property).
    \item \textbf{Server 2 Simulation:} $\cS$ provides these client values to $\cA$. Acting as the honest Server 2, $\cS$ provides hash values $z_i$ consistent with the dummy commitments and an aggregate blinding factor $\tilde{\bfs}$ (derived from $\bfs_H$) such that the final decryption executed by $\cA$ yields exactly $\bfx_{\text{total}} = \bfx_H + \bfx_C$.
\end{enumerate}

\textbf{Indistinguishability Hybrids:}
\begin{description}
    \item[\textbf{$H_0$ (Real):}] Real protocol execution.
    
    \item[\textbf{$H_1$ (Simulated Proofs):}] Replace honest clients' NIZK proofs with simulated proofs. 
    \emph{Justification:} Indistinguishable due to the Special Honest-Verifier Zero-Knowledge (SHVZK) property of the proof system.
    
    \item[\textbf{$H_2$ (Random Commitments):}] Replace honest clients' commitments with commitments to random values.
    \emph{Justification:} Indistinguishable due to the Hiding property of the commitment scheme.
    
    \item[\textbf{$H_3$ (Simulated Ciphertexts / Ideal):}] Replace honest clients' ciphertexts with the output of the LWE simulator (dependent only on the honest sum $\bfx_H$).
    \emph{Justification:} Indistinguishable under the Hint-LWE assumption (Theorem~\ref{thm:simulatability}). 
    
    In $H_3$, the view depends only on the aggregate $\bfx_H$, perfectly matching the Ideal World.
\end{description}

\paragraph{Case 2: Server 2 is Controlled by Semi-Honest Adversary.}
Let $\cA$ control Server 2. The view of Server 2 consists strictly of the blinding factors $\{\bfs_i\}_{i \in \cC_{\text{honest}}^\ast}$ and their openings.
In the real protocol, $\bfs_i \xleftarrow{\$} \bbZ_q^\lambda$ are chosen uniformly at random, independent of the inputs $\bfx_i$.
The simulator $\cS$ simply samples uniform random vectors $\tilde{\bfs}_i$ and sends them to $\cA$.
Since the distribution of $\bfs_i$ is identically distributed in both the real and simulated worlds (uniform random), the views are perfectly indistinguishable.
\end{proof}

\maliciousSOne*

\begin{proof}
\label{proof:mal-bp}
We prove security via simulation. Let $\cA$ be a malicious adversary controlling Server 1 ($S_1$) and a set of corrupt clients $\cC_{\text{corr}}$. Let $\cC_{\text{honest}}$ denote the set of honest clients. As before, let $\cC_\text{drop}$ denote the set of clients dropped from the computation, as specified by the adversary. 

\noindent\textbf{Simulator $\cS$ Construction:}
\begin{enumerate}[leftmargin=*]
    \item \textbf{Setup \& Trapdoors:} $\cS$ generates the NIZK CRS with an extraction trapdoor $\xi$. For the Pedersen Vector Commitment, $\cS$ samples random $\tau_1, \dots, \tau_m \in \mathbb{Z}_q$ and sets $g_j = h^{\tau_j}$ for all $j \in [m]$. The vector $\boldsymbol{\tau} = (\tau_1, \dots, \tau_m)$ serves as the \textit{equivocation trapdoor}.

    \item \textbf{Honest Client Simulation (Round 0):} For each honest client $i \in \cC_{\text{honest}}\setminus \cC_\text{drop}$, $\cS$ samples dummy inputs $\tilde{\bfx}^{(i)} = \mathbf{0}$ and generates dummy ciphertexts $\tilde{\bfy}^{(i)}$ and equivocal commitments. $\cS$ invokes the NIZK simulator to produce proofs $\tilde{\pi}^{(i)}$. 
    
    \item \textbf{Consistency Check \& Online Set (Round 1):} $\cS$ (acting as $S_2$) sends digests to $S_1$ only for clients $i \notin \cC_\text{drop}$. $\cA$ (as $S_1$) determines the final online set $\cV$. $\cS$ sets the final dropout set $D = \{i \mid i \notin \cV\}$.
    
    \item \textbf{Witness Extraction:} For every corrupted client $i \in \cC_{\text{corr}} \cap \cV$ that provided a valid NIZK proof $\pi^{(i)}$ $\cS$ observes the proof and uses the extractor $\mathcal{E}$ to extract the witness $(\bfx^{(i)})$ 

    \item \textbf{Functionality Query:} $\cS$ queries $\cF_{\text{SecAgg}}$ with the extracted inputs for $\cC_{\text{corr}}$ and the dropout set $D$. $\cS$ receives $\bfx_{\text{ideal}} = \sum_{i \in \cV} \bfx^{(i)}$.

    \item \textbf{Equivocation (Round 2):} $\cS$ forces the transcript to yield $\bfx_{\text{ideal}}$ and also provide an opening for the aggregate mask $\mathbf{w}^{\Sigma} = \sum_{i \in \cV} \mathbf{w}^{(i)}$. 

    \item \textbf{Equivocation (Round 2):} $\cS$ forces the transcript to yield $\bfx_{\text{ideal}}$ and provides an opening for the required aggregate mask.
    \begin{itemize}
        \item \textbf{Target Mask Calculation:} $\cS$ calculates the ``target'' aggregate mask $\bfw^*$ required to satisfy the decryption equation. By construction, the sum of the ciphertexts $\sum_{i\in\cV} \bfy^{(i)}$ currently contains the honest dummy masks $\tilde{\bfw}_H = \sum_{i \in \cV \cap \cC_{\text{honest}}} \tilde{\bfw}^{(i)}$, the corrupt masks $\bfw_C = \sum_{i \in \cV \cap \cC_{\text{corr}}} \bfw^{(i)}$, and the corrupt inputs $\Delta\bfx_C$. To ensure the final decode yields $\bfx_{\text{ideal}} = \bfx_H + \bfx_C$, the target aggregate mask must offset the missing honest inputs:
        \[ \bfw^* = \tilde{\bfw}_H + \bfw_C - \Delta\bfx_H \]
        where $\bfx_H = \bfx_{\text{ideal}} - \bfx_C$ and $\bfx_C=\sum_{i\in\cC_\text{corr}\cap \cV}\bfx^{(i)}$
        
        \item \textbf{RO Programming:} $\cS$ selects one active honest client $i^* \in \cC_{\text{honest}} \cap \cV$ and programs the Random Oracle $\mathcal{H}'$ for their seed to output $\tilde{\bfw}^{(i^*)} - \Delta\bfx_H$. For all other honest clients $j \neq i^*$, it programs $\mathcal{H}'(\term{seed}^{(j)}) = \tilde{\bfw}^{(j)}$. This ensures the aggregate programmed mask exactly equals $\bfw^*$.
        
        \item \textbf{Vector Equivocation:} Since $\cS$ already sent dummy commitments $C_w^{(i)}$ in Round 0, it uses the trapdoor vector $\boldsymbol{\tau} = (\tau_1, \dots, \tau_m)$ to calculate a forged aggregate opening $R'_w$ such that:
        \[ \Commit(\bfw^*; R'_w) = \prod_{i \in \cV} C_w^{(i)} \]
        Because the difference between the originally committed values $(\tilde{\bfw}_H + \bfw_C)$ and the target $\bfw^*$ is exactly $\Delta\bfx_H$, this is computed cleanly as:
        \[ R'_w = \left(\sum_{i \in \cV} r_w^{(i)}\right) + \sum_{j=1}^m \tau_j (\Delta x_{H, j}) \]
        where $\Delta x_{H, j}$ is the $j$-th component of the vector $\Delta\bfx_H$.
    \end{itemize}
\end{enumerate}

\noindent\textbf{Indistinguishability Hybrids:}
\begin{description}[leftmargin=1em]
    \item[Hybrid $H_0$ (Real World):] The real protocol execution.
    \item[Hybrid $H_1$ (Simulated NIZK):] For all honest clients $i \in \cC_{\text{honest}}$, replace the real NIZK proofs with simulated proofs generated using the CRS simulation trapdoor. Indistinguishable from $H_0$ due to the \textbf{Zero-Knowledge} property of the NIZK.
    \item[Hybrid $H_2$ (Extraction):] $\cS$ uses the soundness extractor to extract the witness $(\bfx^{(i)}, \bfs^{(i)}, \dots)$ from the proofs provided by $\cA$ for $i \in \cC_{\text{corr}} \cap \cV$. $\cS$ sends these extracted inputs to the ideal functionality $\cF_{\text{SecAgg}}$ to obtain $\bfx_{\text{ideal}}$. Indistinguishable from $H_1$ due to the \textbf{Simulation-Extractability} of the NIZK; the probability that $\cA$ produces a valid proof for a statement from which a valid witness cannot be extracted is negligible.
    \item[Hybrid $H_3$ (Hint-LWE Simulation):] Instead of using the individual honest inputs $\bfx^{(i)}$, $\cS$ invokes the simulator from Theorem~\ref{thm:simulatability}. Given the honest sum $\bfx_H$, the Hint-LWE simulator generates simulated ciphertexts $\{\bfy^{(i)}\}_{i \in \cC_{\text{honest}}}$ such that their LWE components correctly sum to an encryption of $\Delta\bfx_H$. Indistinguishable from $H_2$ under the \textbf{Hint-LWE assumption}.
    \item[Hybrid $H_4$ (RO Programming \& Equivocation):] $\cS$ modifies the Hint-LWE simulation to encrypt $\mathbf{0}$ for all honest clients instead of $\bfx_H$. To ensure the aggregate transcript still decrypts to $\bfx_{\text{ideal}}$, $\cS$ shifts the required offset into the Random Oracle computation. $\cS$ programs $\mathcal{H}'$ for the honest seeds such that the aggregate target mask becomes $\bfw^* = \tilde{\bfw}_H + \bfw_C - \Delta\bfx_H$. $\cS$ then uses the discrete log trapdoor vector $\boldsymbol{\tau}$ to compute a forged aggregate opening $R'_w$ that equivocates the dummy commitments to this programmed mask $\bfw^*$. 
\end{description}

\noindent
The view in $H_4$ is statistically indistinguishable from $H_3$ because: 
(1) The Random Oracle is programmed on honest seeds $\term{seed}^{(i)}$ that are uniformly random and unqueried by $\cA$ prior to Phase 3; 
(2) Pedersen vector commitments are \textbf{perfectly hiding}, so the commitments sent in Phase 1 reveal no information about the shift; and 
(3) The forged opening $R'_w$ is distributed identically to a true opening. In Hybrid $H_4$, the simulated view relies solely on $\bfx_{\text{ideal}}$ and is identical to the Ideal World. This concludes the proof.
\end{proof}

\maliciousSOneLat*

\begin{proof}
\label{proof:mal-lat}
We prove security via simulation. Let $\cA$ be a malicious adversary controlling Server 1 ($S_1$) and a subset of clients $\cC_{\text{corr}}$.  Let $\cC_{\text{honest}}$ denote the set of honest clients. As before, let $\cC_\text{drop}$ denote the set of clients dropped from the computation, as specified by the adversary. 

\noindent\textbf{Simulator $\cS$ Construction:}
\begin{enumerate}[leftmargin=*]
    \item \textbf{Setup \& Lattice Trapdoors:} Instead of generating standard uniformly random matrices, $\cS$ uses the Ajtai trapdoor generation algorithm (e.g., \textsf{TrapGen}) to generate $\bfB_\rho \in \bbZ_q^{\lambda \times m}$ along with a \textbf{short basis trapdoor} $\bfT_\rho$ for the lattice $\Lambda^\perp(\bfB_\rho)$. The matrices $\bfA, \bfB_s, \bfB_w$ are generated uniformly.
    
    \item \textbf{Round 0 (Honest Simulation):} For $i \in \cC_{\text{honest}}\setminus \cC_\text{drop}$, $\cS$ generates dummy inputs $\tilde{\bfx}^{(i)} = \mathbf{0}$, dummy masks $\tilde{\bfw}^{(i)}$, and dummy secrets $\tilde{\bfs}^{(i)} = \mathbf{0}$. 
    \begin{itemize}
        \item It computes dummy ciphertexts $\tilde{\bfy}^{(i)}$ using these zeros.
        \item It computes dummy Ajtai commitments $\tilde{\cC}_j^{(i)} = \bfB_\rho \tilde{\bfrho}_j^{(i)}$ using short Gaussian randomness $\tilde{\bfrho}_j^{(i)} \getsr D_\sigma^{km}$.
        \item It uses the NIZK simulation trapdoor to produce proofs $\tilde{\pi}^{(i)}$.
    \end{itemize}
    
    \item \textbf{Round 1 (Lattice Extraction):} For each corrupt client $i \in \cC_{\text{corr}} \cap \cV$, $\cS$ observes the proof $\pi^{(i)}$. Using the NIZK extractor $\mathcal{E}$ with trapdoor $\xi$, $\cS$ extracts the short vectors $\mathbf{w}_i = (\bfx^{(i)}, \bfs^{(i)}, \bfe^{(i)}, \bfw^{(i)}, \bfrho^{(i)})$. By the soundness of the lattice-based ZK proof, these extracted values correctly bind to $\bfy^{(i)}$ and $\cC_j^{(i)}$.
    
    \item \textbf{Functionality Query:} $\cS$ sends the extracted inputs $\{\bfx^{(i)}\}$ for active corrupted clients and the dropout set $D$ to $\cF_{\text{SecAgg}}$, receiving the true aggregate $\bfx_{\text{ideal}}$.

    \item \textbf{Round 2 (RO Programming \& Ajtai Equivocation):}
    $\cS$ must force the final sum to decrypt to $\bfx_{\text{ideal}}$.
    \begin{itemize}
        \item \textbf{Target Mask:} $\cS$ calculates the required target mask $\bfw^*$ to offset the difference $\Delta\bfx_H = \Delta(\bfx_{\text{ideal}} - \bfx_C)$. 
        \item \textbf{RO Programming:} $\cS$ selects an honest client $i^*$ and programs $\cH'(\term{seed}^{(i^*)})$ such that the aggregate mask equals $\bfw^*$.
        \item \textbf{Ajtai Equivocation:} Because $S_1$ only checks the aggregate secret $\bfs^\Sigma$ and the programmed seeds in Phase 3, explicit equivocation of the individual $\bfrho$ to $S_1$ is not strictly required by the transcript. However, to ensure the internal state of the simulated honest clients is consistent with the dummy commitments $\tilde{\cC}_j^{(i)}$ sent in Phase 1, $\cS$ utilizes the lattice trapdoor $\bfT_\rho$. For a programmed mask $\bfw_j^*$, $\cS$ defines the target syndrome $\bfu = \tilde{\cC}_j^{(i)} - \bfB_w \bfw_j^*$. It then uses Gaussian preimage sampling (\textsf{SamplePre}$(\bfB_\rho, \bfT_\rho, \bfu, \sigma)$) to find a short vector $\bfrho_j^*$ such that $\bfB_\rho \bfrho_j^* = \bfu$.
    \end{itemize}
\end{enumerate}

\noindent\textbf{Indistinguishability Hybrids:}
\begin{description}[leftmargin=1em]
    \item[$H_0 \to H_2$ (NIZK \& Extraction):] Replace honest proofs with simulated ones, and use the extractor for $\cA$'s inputs. Indistinguishable under the \textbf{Zero-Knowledge} and \textbf{Simulation-Extractability} properties of the lattice ZK system.
    \item[$H_3$ (Hint-LWE Simulation):] Replace the honest ciphertexts with those generated by the Hint-LWE simulator targeting the sum $\Delta\bfx_H$. Indistinguishable under the \textbf{Hint-LWE assumption}.
    \item[$H_4$ (RO Programming \& Lattice Equivocation):] Shift the LWE offset to $\mathbf{0}$ and program the Random Oracle $\cH'$ to absorb the $\Delta\bfx_H$ difference. The dummy Ajtai commitments sent in Phase 1 are computationally indistinguishable from real commitments under the \textbf{Module-SIS/LWE assumption}. Furthermore, the Gaussian preimage sampling ensures the equivocated randomness $\bfrho^*$ is statistically close to the real distribution $D_\sigma^{km}$.
\end{description}
The view in $H_4$ relies solely on $\bfx_{\text{ideal}}$ and is indistinguishable from the real execution.
\end{proof}

\end{document}